\documentclass[11pt,a4paper]{article}

\usepackage{preamble}
\usepackage[left=1in, right=1in, top=1in, bottom=1in]{geometry}
\usepackage{graphicx}
\usepackage{hyperref}
\usepackage{xcolor}
\usepackage{adjustbox}[export]
\usepackage{tikz-cd}
\usepackage{authblk}
\usepackage{wrapfig}
\usepackage{enumitem}

\usepackage{todonotes}

\newcommand{\boris}[1]{}
\newcommand{\alex}[1]{}
\newcommand{\note}[1]{}

\newcommand{\interval}[2]{{#2}\rightarrow{#1}}

\newcommand{\psup}[2]{\prescript{#1}{}{#2}}
\newcommand{\und}[1]{\underline{#1}}

\newcommand{\Bd}{\mathrm{Bd}}

\newcommand{\ii}{{\mathsf{i}}}
\newcommand{\f}{{\mathsf{f}}}

\newcommand{\Dr}{\mathrm{Dr}}

\newcommand{\Hloc}{\pi^\1(\caA^{\loc})\ket{\Omega}}
\newcommand{\mn}[2]{{[#2 \to #1]}}
\newcommand{\coneR}[1]{\pi^{\I}(\caA_{#1})''}

\DeclareMathOperator{\anchor}{  \adjincludegraphics[valign=B, width = 0.25cm]{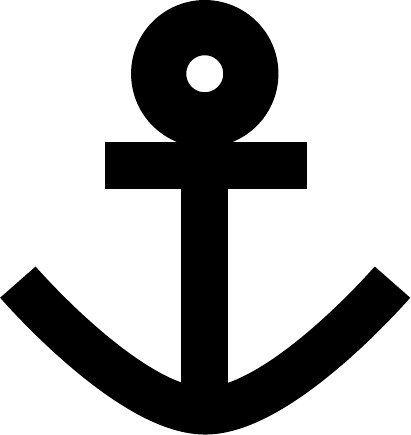}  }
\DeclareMathOperator{\anchorino}{  \adjincludegraphics[valign=B, width = 0.18cm]{anchor_symbol.pdf}  }

\DeclareMathOperator{\tr}{\mathrm{tr}}
\DeclareMathOperator{\Ob}{\mathrm{Ob}}
\DeclareMathOperator{\Irr}{\mathrm{Irr}}

\DeclareMathOperator{\Tube}{\rm{Tube}}

\newcommand{\1}{\mathbbm{1}}

\DeclareMathOperator{\DHR}{\mathsf{SSS}}
\DeclareMathOperator{\SSS}{\mathsf{SSS}}
\renewcommand{\Vec}{\mathsf{Vec}}
\DeclareMathOperator{\Nat}{\mathsf{Nat}}
\DeclareMathOperator{\Fun}{\mathsf{Fun}}

\DeclareMathOperator{\op}{\mathrm{op}}
\DeclareMathOperator{\coev}{\mathrm{coev}}

\DeclareMathOperator{\inn}{\mathrm{in}}

\newcommand{\ZC}{Z(\mathcal{C})}
\DeclareMathOperator{\gl}{\mathrm{gl}}

\newtheorem{excitation lemma}[theorem]{the Excitation Lemma}
\newtheorem{concatenation lemma}[theorem]{the Concatenation Lemma}
\newtheorem{multiplicativity lemma}[theorem]{the Multiplicativity Lemma}
\newtheorem{inclusion lemma}[theorem]{the Inclusion Lemma}
\newtheorem{isotopy lemma}[theorem]{the Isotopy Lemma}

\begin{document}

\title{Sector Theory of Levin-Wen Models II \\ Fusion and Braiding}

\author[1]{Alex Bols \thanks{email: \href{abols01@phys.ethz.ch}{abols01@phys.ethz.ch}}}
\author[2]{Boris Kj\ae r \thanks{email: \href{bbk@math.ku.dk}{bbk@math.ku.dk}}}
\affil[1]{Institute for Theoretical Physics, ETH Z{\"u}rich}
\affil[2]{QMATH, Department of Mathematical Sciences, University of Copenhagen}

\maketitle

\begin{abstract}
    This is the continuation of our study of the Levin-Wen model based on an arbitrary unitary fusion category $\caC$ on the infinite plane. The ground state of the Levin-Wen model hosts anyonic excitations whose fusion and braiding properties are captured by the associated braided $\rm C^*$-tensor category of superselection sectors $\DHR$. By constructing explicit isomorphisms between the fusion spaces of $\DHR$ and those of the Drinfeld center $Z(\caC)$, we show that these two categories have isomorphic $F$- and $R$-symbols. It follows that the full subcategory of finite sectors is unitarily braided monoidally equivalent to the Drinfeld center,  $$\,\DHR_f \simeq Z(\caC).$$ This provides the first complete characterisation of the category of superselection sectors for a class of two-dimensional lattice models supporting anyons with non-integer quantum dimensions.
\end{abstract}

\section{Introduction} \label{sec:introduction}

The ambition to understand gapped phases of quantum spin systems has generated a vast literature. Within mathematical physics, this program is cast as a classification problem; we want to establish a complete set of invariants of gapped phases so that if two gapped ground states have the same invariant, then they belong to the same gapped phase. An important invariant for this classification problem in two dimensions is the \emph{anyon content} of a gapped ground state, or more generally, the topological defects admitted by the ground state. It has become clear over the past decade that these topological defects, together with their fusion and braiding properties, can be captured rigorously using \emph{sector theory} \cite{naaijkens2011localized, ogata2022derivation, bhardwaj2025posets, jones2024dhr, kawagoe2024operator, rubio2024anyonic}.

In this paper we continue our study, initiated in \cite{bols2025levinI} (hereafter `Part I'), of the sector theory of Levin-Wen models on the infinite square lattice based on a unitary fusion category (UFC) $\caC$. These models are believed to represent all gapped phases in two spatial dimensions that admit a gapped boundary. In Part I we classified the irreducible anyon sectors of Levin-Wen models, namely the unitary equivalence classes of irreducible representations of the observable algebra $\caA$ that satisfy the superselection criterion with respect to the GNS-representation $\pi^{\I}$ of the unique Levin-Wen ground state $\omega^{\I}$. Concretely, for each simple object $X \in \Irr \ZC$ of the \emph{Drinfeld center} of $\caC$ we constructed an endomorphism $\rho^X \in \End(\caA)$ so that $\{\pi^X := \pi^\1 \circ \rho^X\}_{X \in \Irr Z(\caC)}$ is a complete set of disjoint irreducible anyon representations.

Under the assumption of (approximate) Haag duality, one can construct a braided $\rm C^*$-tensor category $\DHR$ of \emph{superselection sectors} \cite{naaijkens2011localized, Fiedler2014, ogata2022derivation} associated to $\omega^{\I}$. This construction adapts the sector theory of algebraic quantum field theory \cite{doplicher1969fields, doplicher1971local, buchholz1982locality, frohlich1990braid} to the setting of lattice spin systems. 
The category $\DHR$ rigorously captures the anyon content of the theory, describing the fusion and braiding of anyons, and establishes these data as an invariant of gapped phases \cite{nachtergaele2019quasilocality, ogata2022derivation}.

The irreducible anyon representations $\rho^X$ constructed in $\cite{bols2025levinI}$ are the simple objects of $\DHR$. The results of \cite{bols2025levinI} therefore establish that $\DHR_f$, the full semisimple subcategory of $\DHR$ generated by its simple objects, is linearly equivalent to $Z(\caC)$. The goal of the present paper is to analyse the fusion and braiding properties of the Levin-Wen model's anyonic excitations by establishing the equivalence $\DHR_f \simeq Z(\caC)$ as braided $\rm C^*$-tensor categories. We achieve this by identifying the fusion spaces of these categories by constructing isomorphisms 
\begin{equation*}
    \Phi_{XY}^Z : \ZC(X \otimes Y \to Z) \to \DHR_f(\rho^X \otimes \rho^Y \to \rho^Z),
\end{equation*}
and showing that these isomorphisms preserve $F$- and $R$-symbols. In Appendix \ref{app:equivalence}, we review the details of the familiar statement that the $F$- and $R$-symbols determine semisimple braided monoidal categories up to equivalence. 

Under the assumption of bounded spread Haag duality for the ground state of the Levin-Wen model (\Cref{ass:bounded spread Haag duality} below), we therefore obtain
\begin{theorem}
    \label[theorem]{thm:main theorem}
    There is a unitary braided monoidal equivalence
    \begin{equation*}
        \ZC \simeq \DHR_f.
    \end{equation*}
\end{theorem}
 The proof appears in Section \ref{sec:proof of main theorem}. Note that this equivalence pushes the unitary modular tensor category (UMTC) structure of $Z(\caC)$ forward to $\DHR_f$, showing in particular that anyons of the Levin-Wen model have conjugates (antiparticles).

A similar result has recently been obtained \cite{bols2025classification, bols2025category} for Kitaev's quantum double models \cite{kitaev2003fault} on the plane. For these models one can explicitly construct localized and transportable \emph{amplimorphisms}, originally introduced in \cite{naaijkens2015kitaev}, representing all equivalence classes of objects in $\DHR_f$. The arguments in \cite{bols2025category} rely crucially on the fact that these amplimorphisms moreover provide an action of the model's anyon theory, namely the representation category of the quantum double $\caD(G)$ of the gauge group, on the observable algebra. It follows from the discussion in \cite[Section 5.1]{chen2022q} that the analogous property \emph{cannot} hold for string operators of the Levin-Wen models covered here (indeed, $Z(\caC)$ may have non-integer quantum dimensions). Our methods for computing the braided monoidal structure of $\DHR_f$ for Levin-Wen models therefore necessarily differ significantly from the strategy pursued in $\cite{bols2025category}$. We believe that the strategy presented here can be used to compute the sector theory of representative ground states of \emph{all} gapped phases of two dimensional spin systems \cite{sopenko2023chiral}.

This paper is structured as follows. After introducing the model and its category of superselection sectors in Section \ref{sec:model and DHR}, we summarise in Section \ref{sec:basic notions} some important concepts from Part I that will be used in this second part. The simple objects of $\DHR$ are characterised in Section \ref{sec:simple objects of DHR}. In Section \ref{sec:fusion}, we construct the isomorphisms $\Phi_{XY}^Z$ of fusion spaces and show that they preserve $F$-symbols. In Section \ref{sec:braiding}, we establish that they also preserve $R$-symbols. Finally, the short Section \ref{sec:proof of main theorem} gives the proof of the main theorem \ref{thm:main theorem} by appealing to the result of Appendix \ref{app:equivalence} which reviews how the maps $\Phi_{X Y}^Z$ can be used to construct a unitary braided monoidal equivalence between $Z(\caC)$ and $\DHR_f$.

\subsection*{Acknowledgements} 
We thank Corey Jones and David Penneys for useful conversations. B. K. was supported by the
Villum Foundation through the QMATH Center of Excellence (Grant No. 10059) and the Villum Young Investigator (Grant No. 25452) programs.


\setcounter{tocdepth}{2}
\tableofcontents

\section{The Levin-Wen model and its category of superselection sectors} \label{sec:model and DHR}

This section recaps the definition of the Levin-Wen model, and follows with a detailed definition of the braided C$^*$-tensor category of superselection sectors associated with the gapped ground state of the model under the assumption of bounded spread Haag duality \cite{naaijkens2011localized, ogata2022derivation, bhardwaj2025posets}.
The reader familiar with this material can skip immediately to \Cref{sec:basic notions}. 
\subsection{The Levin-Wen model}

The Levin-Wen model is defined based on a unitary fusion category (UFC). 
Throughout, we fix an arbitrary UFC $\caC$ with representative simple objects $\Irr \caC$ (see Section \ref{ssubsec:UFC} for details).

\subsubsection{Local degrees of freedom}

To each lattice site $v \in \Z^2 \subset \R^2$, we associate a local Hilbert space
\[
\caH_v = \bigoplus_{a, b, c, d \in \Irr \caC} \, \caC(a \otimes b \rightarrow c \otimes d).
\]
An element $\phi \in \caH_v$ in the subspace $\caC(a \otimes b \rightarrow c \otimes d)$ is represented graphically as
\[
\adjincludegraphics[valign=c, height=0.8cm]{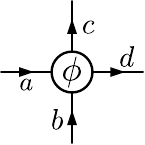},
\]
where the morphism $\phi$ is associated to the site $v$, and the diagram is read from bottom left to top right.
Following \cite{kong2014universal, christian2023lattice, green2024enriched}, we equip $\caH_v$ with the \emph{skein inner product} given by
\[
\langle \phi, \psi \rangle := \frac{\tr \{ \phi^{\dag} \circ \psi \}}{\sqrt{d_a d_b d_c d_d}} 
= \frac{1}{\sqrt{d_a d_b d_c d_d}} \, \adjincludegraphics[valign=c, height=1.0cm]{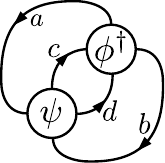}
\qquad\text{for }\phi, \psi \in \caC(a \otimes b \rightarrow c \otimes d).
\]
For any $v \in \Z^2$, let $\caA_v = \End(\caH_v)$. For finite $V \subset \Z^2$, define
\[
\caH_V = \bigotimes_{v \in V} \caH_v, \quad \caA_V = \bigotimes_{v \in V} \caA_v \simeq \End(\caH_V).
\]
If $V \subset W$ are finite subsets of $\Z^2$, there is a natural embedding $\caA_V \hookrightarrow \caA_W$ by tensoring with the identity on $\caA_{W \setminus V}$. For any (possibly infinite) $X \subset \Z^2$, these inclusions form a directed system of matrix algebras with direct limit
\[
\caA_X^{\loc} = \varinjlim_{V \subset\subset X} \caA_V, \quad \caA_X = \overline{\caA_X^{\loc}}^{\norm{\cdot}}.
\]
The *-algebra $\caA_X^{\loc}$ is called the algebra of \emph{local observables} on $X$, and the $\rm C^*$-algebra $\caA_X$ is called the \emph{quasi-local algebra} on $X$. We also write $\caA = \caA_{\Z^2}$ and $\caA^{\loc} = \caA_{\Z^2}^{\loc}$.

For any $S \subset \R^2$, define $\overline{S} = S \cap \Z^2$ and set
\[
\caA_S := \caA_{\overline S}, \quad \caH_S := \caH_{\overline S} \quad \text{if $\overline S$ is finite}.
\]

\subsubsection{The Levin-Wen Hamiltonian and its ground state}

Let $\bse_1 = (1,0)$ and $\bse_2 = (0,1)$ be the unit vectors of the square lattice $\Z^2$. Denote the set of oriented edges by $\vec \caE = \{ (v_0, v_1) \in \Z^2 \times \Z^2 : \dist(v_0, v_1) = 1 \}$.

For $e \in \vec \caE$, write $\bar e = (v_1, v_0)$, $\partial_\ii e = v_0$, and $\partial_\f e = v_1$. For each oriented edge $e = (v_0, v_1) \in \vec \caE$, define the projector $A_e$ on $\caH_{v_0} \otimes \caH_{v_1}$ by
\[
A_e \,\, \adjincludegraphics[valign=c, height=0.8cm]{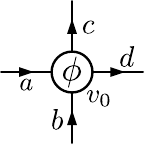} \otimes 
\adjincludegraphics[valign=c, height=0.8cm]{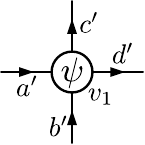} = 
\delta_{d\, a'} \, \adjincludegraphics[valign=c, height=0.8cm]{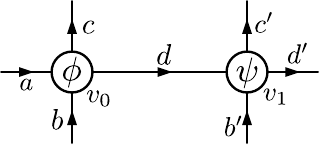}.
\]
Clearly $A_e = A_{\bar e}$. We say that $A_e$ enforces the \emph{string-net constraint} at $e$.

Denote by $\caF$ the set of faces of $\Z^2$. For each face $f \in \caF$, define $\caH_f = \bigotimes_{v \in f} \caH_v$ and an orthogonal projector $B_f$ on $\caH_f$ which annihilates vectors orthogonal to $\prod_{e \in f} A_e \caH_f$ and acts on $\prod_{e \in f} A_e \caH_f$ by inserting the regular element of $\caC$ and using local relations:
\[
	B_f \, \adjincludegraphics[valign=c, height=1.0cm]{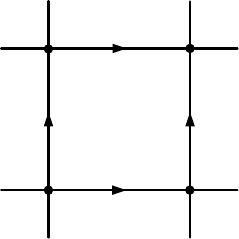} :=
\frac{1}{\caD^2} \sum_{a \in \Irr \caC} d_a \, \adjincludegraphics[valign=c, height=1.0cm]{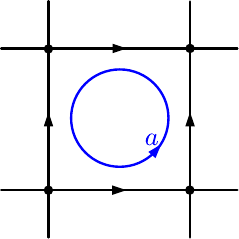},
\]
where $d_a$ is the quantum dimension of $a$ and $\caD^2 = \sum_a d_a^2$. See Section \ref{ssubsec:Bf defined} for a precise definition. The projector $B_f$ is orthogonal. Moreover, all $\{B_f\}_{f \in \caF}$ commute by Lemma \ref{lem:commutativity lemma}, see also \cite[Proposition 5.14]{zhang2016temperley}.

The \emph{Levin-Wen Hamiltonian} is the formal commuting projector Hamiltonian
\[
H_{LW} = - \sum_{f \in \caF} B_f.
\]

A state $\omega: \caA \to \C$ is a \emph{frustration-free ground state} of $H_{LW}$ if $\omega(B_f) = 1$ for all $f \in \caF$. The following Proposition has been proved in \cite{jones2023localtopologicalorderboundary} (see section 2.3 and Theorem 4.8 of that paper) and more recently using different methods in \cite[~Proposition 2.1]{bols2025levinI}:
\begin{proposition} \label{prop:unique ffgs}
	The Levin-Wen Hamiltonian has a unique frustration free ground state which we denote by $\omega^{\I}$. This frustration free ground state is pure.
\end{proposition}

Since $H_{LW}$ is a commuting projector Hamiltonian, it follows that $\omega^{\I}$ is in fact a gapped ground state of $H_{LW}$.
Let $(\pi^{\I}, \caH, \ket\Omega)$ be the GNS triple of the unique frustration free ground state $\omega^{\I}$ of the Levin-Wen Hamiltonian. Since $\omega^{\I}$ is pure, $\pi^{\I}$ is irreducible.

\subsection{The category of superselection sectors} \label{subsec:DHR}

For completeness, we review the construction of the braided $\rm C^*$-tensor category of superselection sectors, in detail, under the assumption of bounded spread Haag duality. See also \cite[~Section 6.1]{bhardwaj2025posets} for a slightly different presentation.

Note that while this section deals with the $\DHR$ category associated with the ground state $\omega^{\I}$ of our Levin-Wen model, the discussions in this section are actually valid for any pure gapped ground state $\omega^{\I}$ on a two-dimensional spin system.

\subsubsection{Cones and anyon representations}

A \emph{cone} $\Lambda \subset \R^2$ is a subset of the plane of the form
\begin{equation*}
	\Lambda = \Lambda_{a, \hat v, \theta} :=  \{ x \in \R^2 \, : \, (x - a) \cdot \hat v > \norm{x-a} \cos (\theta/2)   \} \quad\quad\quad\quad \adjincludegraphics[width=2.5cm,valign=c]{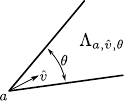}
\end{equation*}
for some \emph{apex} $a \in \R^2$, a unit vector $\hat v \in \R^2$ called the \emph{axis}, and an \emph{opening angle} $\theta \in (0, 2\pi)$. The closure of a cone is also called a cone with the same apex, axis, and opening angle.

Given any cone $\Lambda = \Lambda_{a, \hat v, \theta}$ we write $\Lambda^{+s} := \Lambda_{b, \hat v, \theta}$ with $b = a - s \sin(\theta/2)^{-1} \cdot \hat v$ for the cone obtained by moving the bounding rays of $\Lambda$ perpendicularly outwards by a distance $s$.\\

The basic objects of study in sector theory are so-called anyon representations.
\begin{definition} \label[Definition]{def:anyon representation}
	A representation $\pi: \caA \to B(\caH)$ is localized in a cone $\Lambda$ (with respect to $\pi^{\I}$) if 
	\begin{equation*}
		\pi(x) = \pi^{\I}(x) \qquad \text{for all } x\in \caA_{\Lambda^c}.
	\end{equation*}
	A representation $\pi$ satisfies the superselection criterion (with respect to $\pi^{\I}$) if for any cone $\Lambda$ there is a unitarily equivalent representation localized in $\Lambda$.
	Representations satisfying the superselection criterion are called anyon representations, and their unitary equivalence classes are called anyon sectors.
\end{definition}

\subsubsection{Localized and transportable endomorphisms of the allowed algebra}

Anyon sectors capture the types of anyonic excitations that are supported by the ground state $\omega^{\I}$. In order to investigate the fusion of these anyons, we need a way to compose anyon representations. It turns out that under the assumption of (approximate) Haag duality, any anyon representation can be extended from $\caA$ to an \emph{endomorphism} of a slightly larger algebra. The fusion of anyons then corresponds to the composition of such endomorphisms. \\

For any cone $\Lambda$ we put $\caR_{\Lambda} := \coneR{\Lambda}$ which we call the \emph{cone algebra} of $\Lambda$. Since $\omega^{\I}$ is pure and a gapped ground state of a local Hamiltonian, it follows from \cite[~Lemma 5.3]{ogata2022derivation} that all the cone algebras $\caR_{\Lambda}$ are infinite factors.
The following assumption is adapted from \cite{jones2024dhr}.

\begin{assumption}[Bounded spread Haag duality] \label{ass:bounded spread Haag duality}
    There is $s \geq 0$ so that
    $$ \caR_{\Lambda^c}' \subseteq \caR_{\Lambda^{+s}} $$
    holds for all cones $\Lambda$.
\end{assumption}

\begin{remark}
    (Bounded spread) Haag duality was proven for a wide class of commuting projector models based on $\rm C^*$-weak Hopf algebras in \cite{ogata2025haag}. This class includes models which are believed to be in the same phase as the Levin-Wen model considered here. That is, the Levin-Wen ground state can be transformed into the ground state of one of the models of $\cite{ogata2025haag}$ by finite depth quantum circuits and adding/removing decoupled ancillas. Since bounded spread Haag duality is stable with respect to such transformations, we expect that the Levin-Wen ground state considered here does satisfy bounded spread Haag duality.
    Note, moreover, that the results presented in this paper suffice to derive a meaningful result on the anyon theory even without making the assumption of bounded spread Haag duality; see Remark \ref{rem:no haag} below.
\end{remark}

Fix a unit vector $\hat f \in \R^2$ which we call the \emph{forbidden direction}. We say a cone $\Lambda = \Lambda_{a, \hat v, \theta}$ is \emph{allowed} if $\hat v \cdot \hat f < \cos(\theta/2)$, in which case we write $\Lambda \perp \hat f$. Define the \emph{allowed algebra}
$$ \caB := \overline{ \bigcup_{\Lambda \perp \hat f} \, \caR_{\Lambda} }^{\norm{\cdot}}. $$
Note that $\pi^{\I}(\caA) \subset \caB$ is WOT-dense.

\begin{lemma} \label[lemma]{lem:extension to B}
    Let $\pi : \caA \to B(\caH)$ be an anyon representation. There is a unique WOT-continuous representation $\rho : \caB \to B(\caH)$ such that $\rho \circ \pi^{\I} = \pi$. We call $\rho$ the extension of $\pi$ to $\caB$.
\end{lemma}

\begin{proof}
    Since $\pi$ is an anyon representation there is a unitary $V_{\Lambda}$ for every cone $\Lambda$, such that $\Ad[V_{\Lambda}] \circ \pi^{\I}$ is localized in $\Lambda^c$. Equivalently, $\pi|_{\caA_{\Lambda}} = \Ad[V_{\Lambda}] \circ \pi^{\I}|_{\caA_{\Lambda}}$. Note that if $\Lambda \subset \Lambda'$ then
    $$ \Ad[V_{\Lambda'}] \circ \pi^{\I}|_{\caA_{\Lambda}} = \pi|_{\caA_{\Lambda}} = \Ad[V_{\Lambda}] \circ \pi^{\I}|_{\caA_{\Lambda}}. $$
    Since the adjoint action of a unitary is WOT-continuous, this implies that $\Ad[V_{\Lambda}]|_{\caR_{\Lambda}} = \Ad[V_{\Lambda'}]|_{\caR_{\Lambda}}$ whenever $\Lambda \subset \Lambda'$. Noting that the set of allowed cones, partially ordered by inclusion, is directed, this implies that we can consistently define $\rho$ on $\bigcup_{\Lambda \perp \hat f} \caR_{\Lambda}$ by setting $\rho(x) = \Ad[V_{\Lambda'}](x)$ whenever $x \in \caR_{\Lambda'}$ for an allowed cone $\Lambda'$. Finally, $\rho$ extends to $\caB$ by norm-continuity.

    Any $x \in \caA^{\loc}$ belongs to some $\caA_{\Lambda}$ for an allowed cone $\Lambda$, so we have $\rho(\pi^{\I}(x)) = (\Ad[V_{\Lambda}] \circ \pi^{\I})(x) = \pi(x)$. The equality $\rho \circ \pi^{\I} = \pi$ therefore holds on $\caA^{\loc}$, and it extends to the whole of $\caA$ by continuity.
\end{proof}

At this point we have shown that each anyon representation can be uniquely extended to a representation of the allowed algebra. We will now show, under the assumption of bounded spread Haag duality, that if the anyon representation is localized in an allowed cone, this extension is actually an endomorphism of $\caB$.

\begin{lemma} \label{lem:extension is endomorphism}
	Let $\pi$ be an anyon representation localized in an allowed cone $\Lambda$, and let $\rho : \caB \rightarrow B(\caH)$ be its extension to $\caB$. Under \Cref{ass:bounded spread Haag duality} we have $\rho(\caR_{\Delta}) \subset \caR_{\Delta^{+s}}$ for any allowed cone $\Delta$ that contains $\Lambda$. In particular, $\rho$ is an endomorphism of $\caB$.
\end{lemma}

\begin{proof}
    By Lemma \ref{lem:extension to B} there is a unique WOT-continuous representation $\rho : \caB \to B(\caH)$ such that $\rho \circ \pi^{\I} = \pi$.

    Suppose $\Delta$ is an allowed cone containing $\Lambda$, and take $x \in \caA_{\Delta}$. For any $y \in \caA_{\Delta^c}$ we have
    $$ [\rho(\pi^{\I}(x)), \pi^{\I}(y)] = [ \pi(x), \pi(y) ] = 0$$
    where we used that $\pi$ is localized in $\Lambda$. This implies that $\rho( \pi^{\I}( \caA_{\Delta} ) ) \subseteq \pi^{\I}(\caA_{\Delta^c})' \subseteq \caR_{\Delta^{+s}}$, where the last inclusion is based on the assumption of bounded spread Haag duality. By continuity we obtain $\rho(\caR_{\Delta}) \subseteq \caR_{\Delta^{+s}}$ for any allowed cone $\Delta$ containing $\Lambda$. Again by continuity, this implies $\rho(\caB) \subseteq \caB$ as required.
\end{proof}

Extensions of anyon representations satisfy the following localization and transportability conditions:
\begin{lemma} \label[lemma]{lem:transportability}
    Let $\pi$ be an anyon representation localized in an allowed cone $\Lambda$. Then the extension $\rho$ of $\pi$ to $\caB$  is localized in $\Lambda$ in the sense that $\rho \circ \pi^{\I}|_{\caA_{\Lambda^c}} = \pi|_{\caA_{\Lambda^c}} = \pi^{\I}|_{\caA_{\Lambda^c}}$. Moreover, the extension $\rho$ is transportable in the sense that for any allowed cone $\Delta$ there is a representation $\rho' : \caB \rightarrow \caB(\caH)$ localized in $\Delta$ which is unitarily equivalent to $\rho$.
\end{lemma}

\begin{proof}
    The localization claim is immediate from $\rho \circ \pi^{\I} = \pi$. For transportability, since $\pi$ is an anyon representation there is a unitary $V_{\Delta}$ so that $\pi' := \Ad[V_{\Delta}] \circ \pi$ is localized in $\Delta$. The representation $\pi'$ is an anyon representation, and by Lemma \ref{lem:extension to B} it has a unique WOT-continuous extension $\rho'$ to $\caB$ which is localized in $\Delta$.
    The unitary $V_{\Delta}$ intertwines $\pi$ and $\pi'$, and by continuity this implies that $V_{\Delta}$ intertwines $\rho$ and $\rho'$ as well. That is, $\rho$ and $\rho'$ are unitarily equivalent. This concludes the proof.
\end{proof}

These preliminaries motivate the following definition.
\begin{definition} \label{def:superselection category}
    An endomorphism $\rho \in \End(\caB)$ is localized in a cone $\Lambda$ if $\rho \circ \pi^{\I}|_{\caA_{\Lambda^c}} = \pi^{\I}|_{\caA_{\Lambda^c}}$. The endomorphism $\rho$ is transportable if for any cone $\Delta$ there is an endomorphism $\rho' : \caB \to \caB$ localized in $\Delta$ which is unitarily equivalent to $\rho$.
    
    The category of superselection sectors $\DHR$ is the $\rm C^*$-category whose objects are endomorphisms of $\caB$ which are localized in an allowed cone, and are transportable. The morphisms of $\DHR$ are intertwiners. For any $\rho, \rho' \in \Ob \DHR$ we denote the space of intertwiners from $\rho$ to $\rho'$ by
    $$\DHR(\rho \to \rho') = \{ V \in B(\caH) \, |\, V\rho(x) = \rho'(x) V \text{ for all }x\in \caB\}.$$
    
    Define $\DHR_f$ as the full subcategory of $\DHR$ whose objects $\rho$ have a finite-dimensional endomorphism space $\DHR(\rho \to \rho)$.
\end{definition}

\begin{remark}
    If $\rho \in \Ob \DHR$ is localized in an allowed cone $\Lambda$ then $\rho \circ \pi^{\I}$ is an anyon representation localized in $\Lambda$, of which $\rho$ is the unique extension to $\caB$. It therefore follows from Lemma \ref{lem:transportability} that, under the assumption of bounded spread Haag duality, the objects of $\DHR$ are precisely the extensions of anyon representations localized in allowed cones.
\end{remark}

\begin{remark} \label{rem:why finite endomorphism spaces}
    The reason for introducing $\DHR_f$ is that we want to prove an equivalence with the Drinfeld center $Z(\caC)$, all of whose objects have finite dimensional endomorphism spaces. Since $\DHR$ admits infinite direct sums, restricting to $\DHR_f$ is necessary for the equivalence to hold. We believe that all infinite objects of $\DHR$ are infinite direct sums or integrals of objects of $\DHR_f$.
\end{remark}

\begin{remark} \label[remark]{rem:no haag}
	In Section \ref{sec:string operators} we give explicit constructions of localized and transportable endomorphisms of $\caB$ representing each isomorphism class of simple objects in $\DHR$. 
    It can be shown directly that the construction indeed gives endomorphisms of $\caB$ without using Haag duality (see e.g. \cite[~Proposition 4.6]{naaijkens2011localized}).
    Transportability of these endomoprhisms also holds independently of Haag duality. Indeed, we construct explicit transporters in Section \ref{subsec:transporters} that are localized in the sense of Lemma \ref{lem:locality of intertwiners}.
    
    The subcategory $\caD \subset \DHR$ generated by all finite products and direct sums of those explicitly constructed endomorphisms, naturally has the structure of a braided $C^*$-tensor category. Indeed, Eq. \eqref{eq:braiding defined} can be adapted to define the braiding, direct sums exist by virtue of the cone algebras being properly infinite (cf. Section \ref{subsec:direct sums and subobjects}), while \Cref{lem:Phi dagger lemma}, proven below, can be applied to construct subobjects. See \cite{bols2024double} for a similar approach.

    Having equipped $\caD$ with the structure of a braided $\rm C^*$-tensor category in this way, The arguments in this paper show that $\caD$ is unitary braided monoidally equivalent to $Z(\caC)$, without any (weakened) assumption of Haag duality.
\end{remark}

\subsubsection{Tensor product and braiding} \label{ssubsec:tensor product and braiding}

We show how the category $\DHR$ of localized and transportable endomorphisms is naturally equipped with a tensor product and a braiding. Let us note first the following locality property of intertwiners.
\begin{lemma} \label{lem:locality of intertwiners}
    Let $\rho$ and $\rho'$ be endomorphisms of $\caB$ both localized in an allowed cone $\Lambda$. If \Cref{ass:bounded spread Haag duality} holds, then $\DHR(\rho \to \rho') \subset \caR_{\Lambda^{+s}} \subset \caB$.
\end{lemma}

\begin{proof}
    Let $T \in \DHR(\rho \to \rho')$. Take $y \in \caA_{\Lambda^c}$. Since $\rho$ and $\rho'$ are both localized in $\Lambda$ we have
    $$ T \pi^{\I}(y) = T \rho(\pi^{\I}(y)) = \rho'(\pi^{\I}(y)) T = \pi^{\I}(y) T, $$
    which shows that $T \in \pi^{\I}(\caA_{\Lambda^c})' = \caR_{\Lambda^c}' \subseteq \caR_{\Lambda^{+s}}$, where the last inclusion is by bounded spread Haag duality.
\end{proof}

It follows immediately from this lemma that the following strict tensor product on $\DHR$ is well defined:
\begin{align*}
    \rho \times \sigma &:= \rho \circ \sigma  \quad\quad\quad\quad\quad\qquad\,\, \rho, \sigma \in \Ob \SSS \\
    S \times T &:= S \rho(T) = \rho'(T)S, \quad\quad S \in \DHR(\rho \to \rho'), T \in \DHR(\sigma \to \sigma').
\end{align*}

\begin{remark}
    It will follow from \Cref{prop:fusion rules in DHR} below that $\DHR_f$ is closed under taking tensor products.
\end{remark}

The tensor product of endomorphisms that are localized in disjoint cones that are suitably far apart commutes.

\begin{lemma} \label{lem:commuting endomorphisms}
    Let $\rho$ and $\sigma$ be endomorphisms of $\caB$ localized in allowed cones $\Lambda$ and $\Delta$ respectively, so that $\Lambda^{+s} \cap \Delta^{+s} = \emptyset$. Under \Cref{ass:bounded spread Haag duality} we have $\rho \times \sigma = \sigma \times \rho$.
\end{lemma}

\begin{proof}
    If $x \in \pi^{\I}( \caA_{\Lambda^c \cap \Delta^c} )$ then by localization we have $\rho(x) = x$ and $\sigma(x) = x$, so $\rho(\sigma(x)) = \rho(x) = x = \sigma(x) = \sigma(\rho(x))$ as required.

    If $x \in \pi^{\I}(\caA_{\Delta})$ then by localization $\rho(x) = x$ and therefore $\sigma(\rho(x)) = \sigma(x)$. By Lemma \ref{lem:extension is endomorphism} we have $\sigma(x) \in \caR_{\Delta^{+s}}$, and again by localization $\rho(\sigma(x)) = \sigma(x)$. This shows $\rho(\sigma(x)) = \sigma(\rho(x))$ in this case as well. We find $\rho(\sigma(x)) = \sigma(\rho(x))$ for $x\in \pi^{\I}(\caA_{\Lambda})$ in the same way.
    
    Since $\rho \circ \sigma$ and $\sigma \circ \rho$ are endomorphisms and $\caA$ is generated by $\caA_{\Lambda}, \caA_{\Delta}$, and $\caA_{\Lambda^c \cap \Delta^c}$, it follows from the above that $\rho \circ \sigma = \sigma \circ \rho$ on $\pi^{\I}(\caA)$. Finally, the equality extends to the whole of $\caB$ by continuity.
\end{proof}

For the braiding, note first that for any allowed cone $\Lambda$ there is an allowed cone $\Delta$ lying clockwise from $\Lambda$ w.r.t. the forbidden direction and such that $\Lambda^{+s} \cap \Delta^{+s} = \emptyset$, see Figure \ref{fig:braiding cones}. Given $\rho, \sigma \in \Ob \DHR$ localized in $\Lambda$ and assuming bounded spread Haag duality (\Cref{ass:bounded spread Haag duality}), Lemmas \ref{lem:transportability} and \ref{lem:extension is endomorphism} provide an endomorphism $\sigma' \in \Ob \DHR$ localized in $\Delta$ and a unitary $V \in \SSS(\sigma \to \sigma')$. Then we may define
\begin{equation} \label{eq:braiding defined}
    b_{\rho, \sigma} := V^* \rho(V) \in \DHR(\rho \times \sigma \to \sigma \times \rho).
\end{equation}
To see that $b_{\rho, \sigma}$ intertwines $\rho \times \sigma$ and $\sigma \times \rho$, compute for any $x \in \caB$
\begin{align*}
    b_{\rho, \sigma} \, (\rho \times \sigma)(x) &= V^* \rho(V) \rho(\sigma(x)) = V^* \rho( V \sigma(x) ) = V^* \rho( \sigma'(x) ) \rho(V) = V^* \sigma'( \rho(x) ) \rho(V) \\
    &= \sigma(\rho(x)) V^* \rho(V) = (\sigma \times \rho)(x) \, b_{\rho, \sigma}
\end{align*}
where we used Lemma \ref{lem:commuting endomorphisms} in the fourth step.

\begin{lemma}
    Under \Cref{ass:bounded spread Haag duality}, the morphism defined in Eq. \eqref{eq:braiding defined} does not depend on the choice of 
    $\Delta$, $\sigma'$, or $V$.
\end{lemma}

\begin{proof}
    Suppose $\Delta'$, $\sigma''$, and $W$ satisfy the same hypotheses as $\Delta, \sigma'$ and $V$ above. 
    
    Assume there is an allowed cone $\Gamma \supset \Delta, \Delta'$ such that $\Gamma^{+s} \cap \Lambda^{+s} = \emptyset$ and lies clockwise from $\Lambda$ w.r.t. the forbidden direction. Now, $VW^* \in \DHR(\sigma'' \to \sigma')$, and both $\sigma'$ and $\sigma''$ are localized in $\Gamma$. It therefore follows from Lemma \ref{lem:locality of intertwiners} that $VW^* \in \caR_{\Gamma^{+s}}$. Since $\rho$ is localized in $\Lambda$, we have $(VW^*)^*\rho(VW^*) =1$. Using this, we obtain
    $$ W^* \rho(W) = W^* (VW^*)^* \rho( VW^* W ) = V^* \rho(V), $$
    as required.
    
    In general, one can always find a \emph{zig-zag} from $\Delta$ to $\Delta'$ which is disjoint from $\Lambda^{+s}$ \cite{bhardwaj2025posets}. That is, a sequence of allowed cones $\Delta_1 = \Delta, \Delta_2, \cdots, \Delta_{n-1}, \Delta_n = \Delta'$ and allowed cones $\Gamma_1, \cdots, \Gamma_{n-1}$ such that $\Gamma_i^{+s} \cap \Lambda^{+s} = \emptyset$ for each $i = 1, \cdots, n-1$, and such that each $\Gamma_i$ lies clockwise from $\Lambda$ w.r.t. the forbidden direction, and contains $\Delta_{i}$ and $\Delta_{i+1}$. The result then follows from a repeated application of the above.
\end{proof}

\begin{lemma} \label{lem:braiding}
    Under \Cref{ass:bounded spread Haag duality}, the morphisms given by Eq. \eqref{eq:braiding defined} define a unitary braiding on $\DHR$. That is, for all $\rho, \sigma \in \Ob \DHR$ the morphism $b_{\rho, \sigma}$ is unitary, and we have
    $$ b_{\rho', \sigma'} \, S \times T = T \times S \, b_{\rho, \sigma} $$
    for any $\rho',\sigma' \in \Ob\DHR$,  $S \in \DHR(\rho \to \rho')$, and  $T \in \DHR(\sigma \to \sigma')$, and
    \begin{align*}
        b_{\rho \times \sigma, \tau} &= ( b_{\rho, \tau} \times \I_{\sigma} ) \, ( \I_{\rho} \times b_{\sigma, \tau} ) \\
        b_{\rho, \sigma \times \tau} &= (\I_{\sigma} \times b_{\rho, \tau}) \, ( b_{\rho, \sigma} \times \I_{\tau} )
    \end{align*}
    for all $\rho, \sigma, \tau \in \Ob \DHR$.
\end{lemma}

\begin{figure}
    \begin{center}
        \includegraphics[width=5cm]{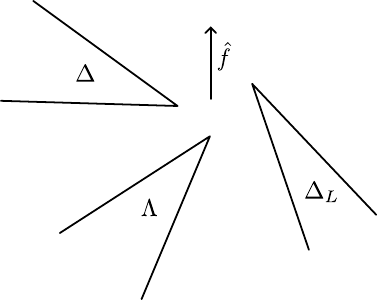}
    \end{center}
    \caption{An allowed cone $\Lambda$ together with allowed cones $\Delta$ and $\Delta_L$, both disjoint from $\Lambda^{+s}$ and lying respectively clockwise and counter clockwise from $\Lambda$ w.r.t. the forbidden direction $\hat f$. }
    \label{fig:braiding cones}
\end{figure}

\begin{proof}
    Unitarity follows immediately from the definition.
    To show $b_{\rho', \sigma'} S \times T = T \times S b_{\rho, \sigma}$, we suppose that the endomorphisms $\rho, \rho', \sigma$ and $\sigma'$ are all localized in an allowed cone $\Lambda$, and fix an allowed cone $\Delta$ that lies clockwise from $\Lambda$ so that $
    \Lambda^{+s} \cap \Delta^{+s}=\emptyset$, as well as endomorphisms $\sigma_R, \sigma'_R$ localized in $\Delta$ and unitaries $V \in \DHR(\sigma \rightarrow \sigma_R)$ and $W \in \DHR(\sigma' \rightarrow \sigma'_R)$.
    
    Let us first show the special case where $\sigma = \sigma'$ and $T = \I_{\sigma}$. then $S \times T = S$ and $T \times S = \sigma(S)$. We must therefore show $b_{\rho', \sigma} S = \sigma(S) b_{\rho, \sigma}$. This we do as follows:
    $$ b_{\rho', \sigma} S = V^* \rho'(V) S = V^* S \rho(V) = V^* \sigma_R(S) \rho(V) = \sigma(S) V^* \rho(V) = \sigma(S) b_{\rho, \sigma},$$
    where we noted that Lemma \ref{lem:locality of intertwiners} and the localization of $\sigma_R$ implies that $\sigma_R(S) = S$.

    Let us now show the special case where $\rho = \rho'$ and $S = \I_{\rho}$. Then $S \times T = \rho(T)$ and $T \times S = T$, so we must show $b_{\rho, \sigma'} \rho(T) = T b_{\rho, \sigma}$. To do this, note first that we can pick an allowed cone $\Delta_L$ which lies counter clockwise from $\Lambda$ and such that $\Delta_L^{+s} \cap \Lambda^{+s} = \emptyset$ (see Figure \ref{fig:braiding cones}), as well as an endomorphism $\rho_L$ localized on $\Delta_L$ and a unitary $U \in \DHR(\rho \to \rho_L)$. The previous special case then yields
    $$ b_{\rho, \sigma} = \sigma(U^*) b_{\rho_L, \sigma} U \quad \text{and} \quad b_{\rho, \sigma'} = \sigma'(U^*) b_{\rho_L, \sigma'} U. $$
    Using this, we compute
    $$ b_{\rho, \sigma'} \rho(T) = \sigma'(U^*) b_{\rho_L, \sigma'} U \rho(T) = \sigma'(U^*) W^* \rho_L(W) \rho_L(T) U = \sigma'(U^*) T S $$
    where we used Lemma \ref{lem:locality of intertwiners} and the localization of $\rho_L$ to conclude that $\rho_L(T) = T$ and $\rho_L(W) = W$. Similarly,
    $$ T b_{\rho, \sigma} = T \sigma(U^*) b_{\rho_L, \sigma} U = \sigma'(U^*) T V^* \rho_{L}(V) U = \sigma'(U^*) T S, $$
    so both special cases are now proven. The general case follows because
    \begin{align*}
        b_{\rho', \sigma'} \, S \times T &= b_{\rho', \sigma'} (S \times \I_{\sigma'}) (\I_{\rho} \times T) = (\I_{\sigma'} \times S) b_{\rho, \sigma'} (\I_{\rho} \times T) \\
        &= (\I_{\sigma'} \times S) (T \times \I_{\rho}) b_{\rho, \sigma} = T \times S \, b_{\rho, \sigma}.
    \end{align*}

    Fix $\rho', \sigma', \tau'$ localized in $\Delta$ with unitary transporters $V_{\rho} \in \DHR(\rho \to \rho')$, $V_{\sigma} \in \DHR(\sigma \to \sigma')$, and $V_{\tau} \in \DHR(\tau \to \tau')$. Then
    $$ (b_{\rho, \tau} \times \I_{\sigma})(\I_{\rho} \times b_{\sigma, \tau}) = \beta_{\rho, \tau} \, \rho( \beta_{\sigma, \tau} ) = V_{\tau}^* \rho(V_{\tau}) \, \rho( V_{\tau}^* \sigma(V_{\tau}) ) = V_{\tau}^* (\rho \times \sigma)( V_{\tau} ) = b_{\rho \times \sigma, \tau}, $$
    which shows the first hexagon identity. For the second, note that $V_{\sigma \times \tau} := V_{\sigma} \times V_{\tau} = V_{\sigma} \sigma(V_{\tau}) \in \DHR( \sigma \times \tau \to \sigma' \times \tau' )$ and compute
    \begin{align*}
        (\I_{\sigma} \times b_{\rho, \tau})(b_{\rho, \sigma} \times \I_{\tau}) &= \sigma \big( V_{\tau}^*  \rho(V_{\tau}) \big) \, V_{\sigma}^* \rho(V_{\sigma}) = \sigma(V_{\tau})^* V_{\sigma}^* \, \sigma' \big( \rho(V_{\tau}) \big) \rho(V_{\sigma}) \\
        &=  \big( V_{\sigma} \sigma(V_{\tau}) \big)^*\, \rho \big( \sigma'(V_{\tau}) V_{\sigma} \big) = V^*_{\sigma \times \tau} \, \rho \big(  V_{\sigma \times \tau} \big) = b_{\rho, \sigma \times \tau}
    \end{align*}
    where we used Lemma \ref{lem:commuting endomorphisms} in the third step.
\end{proof}

\subsubsection{Direct sums and subobjects} \label{subsec:direct sums and subobjects}

Let $\rho_1, \cdots, \rho_n \in \Ob \DHR$, then there is an allowed cone $\Lambda$ so that each $\rho_i$ is localized in $\Lambda$. Since the cone algebra $\caR_{\Lambda}$ is properly infinite we can use the Halving Lemma \cite[~Lemma 6.3.3]{kadison16fundamentals} to construct isometries $\{ v_i \}_{i = 1}^n$ generating an isomorphic copy of the Cuntz algebra $\caO_n$. One then easily checks that
$$ \rho := \sum_{i = 1}^n \Ad[v_i] \circ \rho_i $$
is an endomorphism of $\caB$ which is localized in $\Lambda$ and transportable. This shows that $\DHR$ has direct sums. It is clear from the construction that $\DHR_f$ is closed under direct sums.

Now let $\rho \in \Ob \DHR$ be localized in an allowed cone $\Lambda$ and let $p \in \DHR(\rho \to \rho)$ be a non-trivial orthogonal projection. Assuming bounded spread Haag duality (\Cref{ass:bounded spread Haag duality}), Lemma \ref{lem:locality of intertwiners} yields $p \in \caR_{\Lambda^{+s}}$, so $\tilde \sigma(x) := p \rho(x) = \rho(x) p$ defines a (non-unital) endomorphism of $\caB$. In particular, $\tilde \sigma$ restricts to a non-zero representation of $\pi^{\I}(\caA)$ on $p \caH$. Since $\caA$ is simple, this representation is faithful. This implies that $p$ must be infinite. Since $\caR_{\Lambda^{+s}}$ is a properly infinite factor, it follows that $p \sim \I$, that is, there is a partial isometry $w \in \caR_{\Lambda^{+s}}$ such that $p = w w^*$ and $\I = w^* w$. Consider now the endomorphism $\sigma = \Ad[w] \circ \rho$. One easily checks that $\sigma$ is localized in $\Lambda^{+s}$, and that it is transportable. This shows that $\DHR$ admits subobjects. Note finally that if $\rho \in \Ob \DHR_f$, then also $\sigma \in \Ob \DHR_f$, so $\DHR_f$ is closed under taking subobjects. \\

In summary, we conclude that if \Cref{ass:bounded spread Haag duality} holds, then both $\DHR$ and $\DHR_f$ are braided $\rm C^*$-tensor categories, and $\DHR_f$ is moreover semisimple.

\section{Summary of Part I} \label{sec:basic notions}

In this section, we review general facts about skein theory of unitary fusion categories as well as results from \cite{bols2025levinI} that are used in this manuscript.

\subsection{Skein modules}

\subsubsection{Unitary fusion categories and the Drinfeld center}\label{ssubsec:UFC}

Fix a unitary fusion category (UFC) $\caC$ with its canonical spherical structure; see
\cite{etingof2015tensor,turaev2017monoidal,penneysUFCnotes,penneys2018unitary}.
For objects $x,y \in \Ob\caC$, denote the morphism space by $\caC(x\to y)$.
Each endomorphism space $\caC(x\to x)$ carries the spherical trace $\tr$, and any
nonzero object has quantum dimension $d_x := \tr(\id_x) > 0$.
Choose representatives $\Irr\caC$ of simple objects containing the tensor unit $\I$.
The total quantum dimension is
\[
\caD := \Bigl(\sum_{a\in\Irr\caC} d_a^2\Bigr)^{1/2}.
\]

Duals are written $x^*$ with evaluation and coevaluation maps
$\ev_x : x^*\otimes x \to \I$ and $\coev_x : \I \to x\otimes x^*$.
For each simple $a$ there is a unique $\bar a\in\Irr\caC$ isomorphic to $a^*$.
The unitary structure provides adjoints
$f \mapsto f^\dagger : \caC(x\to y)\to\caC(y\to x)$,
making $\caC(x\to y)$ a Hilbert space with the \emph{trace inner product}
$(f,g)_{\tr} := \tr(f^\dagger g)$.

We employ the graphical calculus throughout (cf.~\cite[Ch.~I.2]{turaev2017monoidal}),
suppressing associators, unitors, and pivotal maps as usual.

\medskip
The \emph{Drinfeld center} $Z(\caC)$ has objects $(X,\sigma)$, where
$X\in\Ob\caC$ and $\sigma : X\otimes - \Rightarrow -\otimes X$ is a half-braiding
with components
\begin{equation}\label{eq:half-braidings}
\sigma_x =
\adjincludegraphics[valign=c,height=1.0cm]{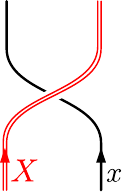},
\qquad
\sigma_x^{-1} =
\adjincludegraphics[valign=c,height=1.0cm]{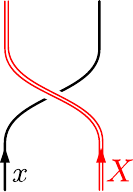}.
\end{equation}
These satisfy $\sigma_{\I}=\id_X$ and
$\sigma_{x\otimes y}=(\id_y\otimes\sigma_z)\circ(\sigma_y\otimes\id_z)$.
A morphism $(X,\sigma)\to(X',\sigma')$ is a map $f\in\caC(X\to X')$ intertwining the half-braidings:
\[
\adjincludegraphics[valign=c,height=1.5cm]{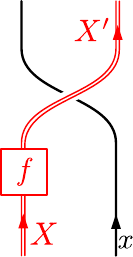}
\;=\;
\adjincludegraphics[valign=c,height=1.5cm]{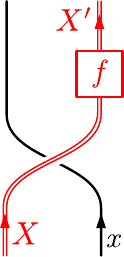}.
\]
We generally omit half-braidings from the notation for objects.

The category $Z(\caC)$ is a unitary modular tensor category whose dagger and
pivotal structures agree with those of $\caC$
\cite{muger2003subfactorsquantumdouble}; see also
\cite[Prop.~2.3]{henriques2015categorifiedtracemoduletensor} in the non-strict
pivotal case.
Let $\Irr Z(\caC)$ be representatives of simple objects containing $\I$; for each
$X \in \Irr \ZC$ there is a unique $\bar X \in \Irr \ZC$ isomorphic to $X^*$.
The braiding in $Z(\caC)$ is
\begin{equation}\label{eq:Drinfeld center braiding}
\beta_{(X,\sigma),(X',\sigma')}=\sigma_{X'}=
\adjincludegraphics[valign=c,height=1.2cm]{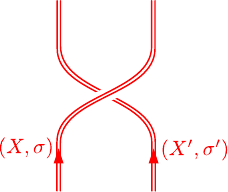}.
\end{equation}

\subsubsection{Skein modules on extended surfaces}

We define skein modules for decorated surfaces; compare
\cite[App.~A]{koenig2010quantum}, \cite[Sec.~2]{kirillov2011string},
\cite{walker2021universal,walker2006tqft}. All manifolds are taken to be
piecewise linear. A \emph{decorated $1$-manifold} is a compact oriented
$1$-manifold $\caN$ (possibly with boundary) equipped with finitely many
\emph{signed marked points} $m_{\caN}\subset\caN$.  
A \emph{decorated surface} is an oriented surface $\Sigma$ together with a
decoration of $\partial\Sigma$. We freely refer to the topology of the
underlying surface (e.g.\ a sphere with $m$ disks removed).

A \emph{string diagram} in $\Sigma$ consists of an embedded graph $\Gamma$
whose edges are labeled by objects of $\caC$ and whose internal vertices are
labeled by compatible morphisms in $\caC$, following the graphical calculus.
Via the forgetful functor $Z(\caC)\to\caC$, labels from $Z(\caC)$ are also
allowed. The graph intersects $\partial\Sigma$ transversely at marked points,
with edge orientation pointing into (resp.\ out of) the boundary at positive
(resp.\ negative) points. These boundary labels determine the
\emph{boundary condition} $b : m_{\partial\Sigma}\to\Ob\caC$.

This boundary condition is \emph{simple} when all labels lie in $\Irr\caC$.  
Let $\Bd(\Sigma)$ denote the set of connected boundary components, viewed as
decorated $1$-manifolds. For $\caN\subset\partial\Sigma$, write $b_{\caN}$ for
the restriction of $b$, and for a diagram $x$ let $x_{\caN}$ denote its induced
boundary data.

For a decorated $1$-manifold $\caN$, define $\hat\caN$ by reversing orientation
and switching all signs; similarly obtain $\hat\Sigma$ from $\Sigma$, so that
$\partial\hat\Sigma=\widehat{\partial\Sigma}$.  
A diagram $x$ on $\Sigma$ determines $\hat x$ on $\hat\Sigma$ by reversing edge
orientations and replacing each internal vertex label with its dagger.

Let $\scrS(\Sigma;b)$ be the set of string diagrams on $\Sigma$ with boundary
condition $b$. Define
\[
A(\Sigma;b) := \C[\scrS(\Sigma;b)]/\!\sim,
\]
the vector space spanned by such diagrams modulo the relations generated by  
\begin{enumerate}[label=(\roman*)]
    \item isotopy in $\Sigma$ fixing $\partial\Sigma$, and
    \item local relations applied inside contractible disks intersecting
edges transversely.
\end{enumerate}

The \emph{skein module} of $\Sigma$ is
\[
A(\Sigma) := \bigoplus_{b\in\caB(\Sigma)} A(\Sigma;b),
\]
where $\caB(\Sigma)$ is the finite set of simple boundary conditions.

For any diagram $x$ in $\Sigma$ we denote its class in $A(\Sigma)$ by
$[x]_{\Sigma}$, or simply by $[x]$ when the surface is clear from context.

\begin{convention} \label[convention]{conv:dotted line}
    We allow string diagrams to have edges coloured by a dotted line, defined by
    \begin{equation} \label{eq:dotted line defined}
        \adjincludegraphics[valign=c, height = 0.8cm]{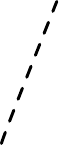} = \frac{1}{\caD^{2}} \sum_{a \in \Irr \caC} d_a \,\,\, \adjincludegraphics[valign=c, height = 0.8cm]{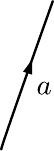}\quad.
    \end{equation}
    This dotted line satisfies the following local relations for any $x \in \Ob \caC$ and any $X \in \Ob Z(\caC)$ (\cite[Corollary~3.5]{kirillov2011string}, \cite[Lemma~2.2]{kirillov2010turaev}):
    \begin{equation} \label{eq:normalization and cloaking properties}
        \adjincludegraphics[valign=c, width = 1.5cm]{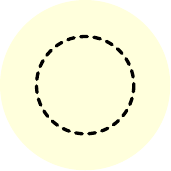} = \adjincludegraphics[valign=c, width = 1.5cm]{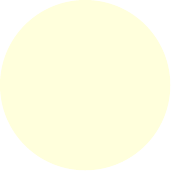}, \quad 
        \adjincludegraphics[valign=c, width = 1.5cm]{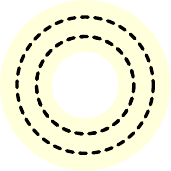} = \adjincludegraphics[valign=c, width = 1.5cm]{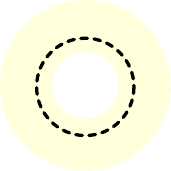}, \quad \adjincludegraphics[valign=c, width = 1.5cm]{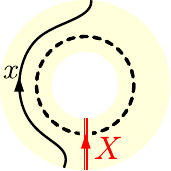} = \adjincludegraphics[valign=c, width = 1.5cm]{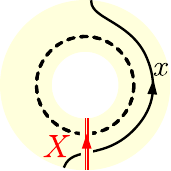}
    \end{equation}
    on embedded disks and annuli in any decorated surface. These are called the normalization, projector, and cloaking properties of the dotted line.
\end{convention}

\subsubsection{Gluing}\label{ssubsec:gluing}

Let $\Sigma$ be a decorated surface and let $\caN,\caM\subset\partial\Sigma$
be disjoint decorated boundary submanifolds.  
We say $\caN$ and $\caM$ are \emph{gluable} if
$\caN \cong \hat{\caM}$ as decorated $1$-manifolds.  
In this case there exists an orientation-reversing homeomorphism
$\psi:\caN\to\caM$ preserving marked points, and we denote by
$\Sigma_{\psi}$ the surface obtained by identifying $\caN$ with $\caM$
via $\psi$.  
The resulting decorated surface depends only on
the isotopy class of $\psi$ \cite[Lem.~4.1.1]{bakalov2001lectures}.

Fix such a map $\psi$ and write $\Sigma_{\gl}$ for the glued surface.
For a string diagram $x$ on $\Sigma$, the boundary data along
$\caN$ and $\caM$ are said to \emph{match} when each pair of corresponding
marked points $(m,\psi(m))$ carries the same label.  
When this holds, $x$ naturally determines a diagram $x_{\gl}$ on
$\Sigma_{\gl}$:
\[
\adjincludegraphics[valign=c,height=1.0cm]{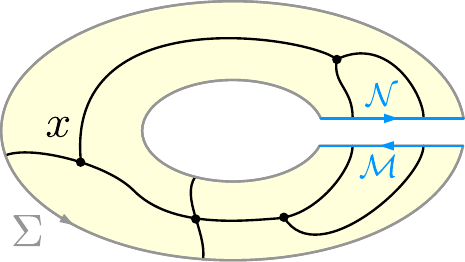}
\;\xmapsto{\gl}\;
\adjincludegraphics[valign=c,height=1.0cm]{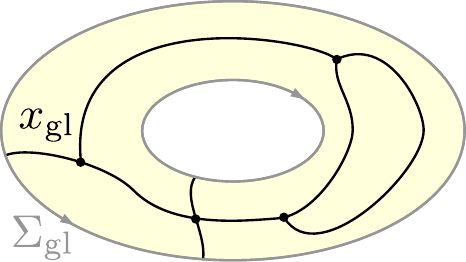}.
\]

This construction descends linearly to skein modules, giving a well-defined
map
\[
\gl : A(\Sigma) \to A(\Sigma_{\gl}), \qquad
\gl([x]) :=
\begin{cases}
[x_{\gl}], & \text{if the boundary labels match},\\
0, & \text{otherwise}.
\end{cases}
\]

\subsubsection{Tube algebras and actions on skein modules}

Let $\caS$ be a decorated circle.  
Forming the product with the unit interval $I$ produces a decorated cylinder
$\caS\times I$ whose bottom and top boundaries are
\[
\partial_b(\caS\times I)=\hat\caS\times\{0\}, 
\qquad
\partial_t(\caS\times I)=\caS\times\{1\}.
\]
The skein module
\[
\Tube_{\caS}:=A(\caS\times I)
\]
carries a natural $\mathrm{C}^*$-algebra structure: multiplication is defined by
gluing the bottom of one representative to the top of another, and the involution
is given by reflection
\[
[x]^*=[f(\hat x)], \qquad
f(\theta,r)=(\theta,1-r).
\]
This algebra is the \emph{Tube algebra} associated to $\caS$
\cite{izumi2000,izumi2001examples,muger2003subfactorsquantumdouble}.

To describe matrix units, we introduce \emph{extended} manifolds:  
a decorated $1$-manifold together with a chosen fiducial point (distinct from
marked points) is called an extended $1$-manifold, and a decorated surface each of whose
boundary components is equipped with a fiducial point is an \emph{extended surface}.
Fiducial points are indicated graphically by an $\mathrm{X}$.

Choose a fiducial point on $\caS$, and suppose $\caS$ has marked points
$m_{\caS}=\{m_1,\dots,m_n\}$ ordered from the fiducial point in the direction
opposite to the orientation of $\caS$.  
Let $\sigma_i\in\{+,-\}$ denote the sign of $m_i$.
For a boundary condition $\underline a:m_{\caS}\to\Irr\caC$, define
\[
\otimes\underline a
:= a(m_1)^{\sigma_1}\otimes\cdots\otimes a(m_n)^{\sigma_n}, 
\qquad
d_{\underline a}:=d_{\otimes\underline a},
\]
and set
\[
\chi^{\otimes\caS}
=\bigoplus_{\underline a:m_{\caS}\to\Irr\caC}\otimes\underline a.
\]

For each $X\in\Irr Z(\caC)$ and boundary condition $\underline a$, fix an
orthonormal basis $\{w^{X\,\underline a}_i\}_i\subset\caC(X\to\otimes\underline a)$ with respect to the trace inner product, represented graphically by
$$ w^{X \underline{a}}_{i} = \adjincludegraphics[valign=c, height = 1.2cm]{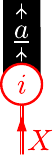}, \quad\quad (w^{X \underline{a}}_{i})^{\dag} =  \adjincludegraphics[valign=c, height = 1.2cm]{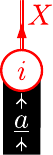}, \quad\quad \text{where} \quad \quad  \adjincludegraphics[valign=c, height = 1.0cm]{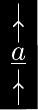} = \adjincludegraphics[valign=c, height = 1.0cm]{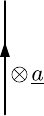} = \adjincludegraphics[valign=c, height = 1.0cm]{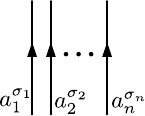}\quad.
$$

\begin{proposition}[{\cite[Prop. 3.4]{bols2025levinI}}]
\label[proposition]{prop:matrix units for Tube_n}
The elements
\[
E^{X}_{\underline b,j;\underline a,i}
:= d_X\,
\adjincludegraphics[valign=c,width=2.0cm]{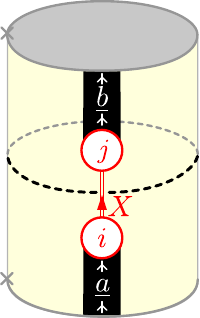}
\]
form a complete system of matrix units in $\Tube_{\caS}$. They satisfy
\begin{equation}\label{eq:Tube_n matrix unit properties}
E^{Y}_{\underline c,k;\underline b',j'}
\,E^{X}_{\underline b,j;\underline a,i}
=
\delta_{XY}\,\delta_{\underline b\,\underline b'}\,\delta_{jj'}\,
E^{X}_{\underline c,k;\underline a,i},
\qquad
\sum_{X,\underline a,i}E^{X}_{\underline a,i;\underline a,i}
=\id_{\Tube_{\caS}}.
\end{equation}
Consequently, the minimal central projections are
\[
P^X=\sum_{\underline a,i}E^{X}_{\underline a,i;\underline a,i},
\qquad X\in\Irr Z(\caC).
\]
\end{proposition}

A projection $p\in\Tube_{\caS}$ is said to be of \emph{type} $X$ if
$p \le P^X$.

\subsubsection{The TQFT inner product}\label{ssubsec:TQFT inner product}

The skein modules $A(\Sigma)$ serve as state spaces for the
Turaev-Viro–Barrett-Westbury TQFT \cite{turaev1992state,barrett1996invariants}.
Since $\caC$ is unitary, the TQFT partition function $Z$ induces an inner
product on $A(\Sigma)$ \cite{walker2021universal,walker2006tqft} which we now recall.

The partition function $Z$ assigns a complex number $Z(M;x)$ to any 3-manifold $M$
equipped with a string diagram $x\in\scrS(\partial M)$.  
This depends only on the class $[x]\in A(\partial M)$, giving a linear map
\[
Z_M : A(\partial M) \to \C.
\]

For a surface $\Sigma$, define the pinched product $\Sigma\times I$ by
collapsing $\partial\Sigma\times I\simeq\partial\Sigma$, so that
\[
\partial(\Sigma\times I) = \Sigma \cup \hat\Sigma.
\]  
Given $[x]\in A(\Sigma;b)$ and $[y]\in A(\Sigma;b')$, define
\[
([x],[y])_{A(\Sigma)} := \delta_{b\,b'} \; Z_{\Sigma\times I}(\hat x \cup y),
\]
where $\hat x \cup y$ is the diagram on $\partial(\Sigma\times I)$ consisting
of $\hat x$ on $\hat\Sigma$ and $y$ on $\Sigma$.  
This yields a well-defined positive-definite inner product on $A(\Sigma)$,
called the \emph{TQFT inner product}.

With respect to this inner product, the Tube algebra actions along boundary
components are *-representations, endowing $A(\Sigma)$ with the structure of a
unitary $A(\partial\Sigma)$-module.

\subsubsection{Skein modules on punctured disks}\label{ssubsec:skein modules on punctured disks}

Let $\Sigma$ be an extended surface homeomorphic to a disk with $m$ interior holes.

An \emph{anchor} $\anchor$ is an embedded graph in $\Sigma$ with a single vertex, the \emph{anchor point}, and edges that connect it transally to each boundary
component $\caS\in\Bd(\Sigma)$ at a point distinct from the fiducial point.
The edges are ordered clockwise around the anchor, inducing an enumeration
$\{\caS^{\anchorino}_{\kappa}\}_{\kappa=0}^m$ of $\Bd(\Sigma)$, with
$\caS^{\anchorino}_0$ being the outer boundary.

\begin{definition}[\cite{bols2025levinI}]\label[definition]{def:equivalence of anchors}
Two anchors are \emph{equivalent} if they induce the same enumeration of
$\Bd(\Sigma)$ and can be deformed into each other via isotopy fixing all
fiducial points.  
They are \emph{equivalent up to a boundary component} $\caS$ if the subgraphs
obtained by removing the edge that connects to $\caS$ from each anchor are isotopic.
\end{definition}

For any extended circle $\caS$ and $X\in\Ob Z(\caC)$, the space
$\caC(X\to\chi^{\otimes\caS})$ is a unitary left $\Tube_{\caS}$-module via
\begin{equation}\label{eq:Tube_n_module}
\adjincludegraphics[valign=c,width=2.0cm]{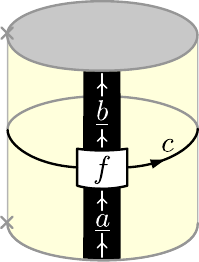} 
\;\triangleright\; 
\adjincludegraphics[valign=c,width=0.7cm]{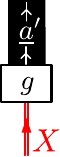}
=
\delta_{\underline a,\underline a'} \;\adjincludegraphics[valign=c,width=2.5cm]{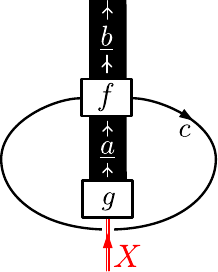} \,\,.
\end{equation}

Define
\begin{align}\label{eq:definition of algebraic module}
\begin{split}
\caC_{\anchorino}(\Sigma) := \bigoplus_{X_0,\dots,X_m\in\Irr(Z(\caC))}
Z(\caC)(X_0^*\to X_1\otimes\dots\otimes X_m) 
\;\otimes\;&\caC(X_0\to \chi^{\otimes\caS^{\anchorino}_0})\\
&\otimes\cdots\otimes \caC(X_m\to \chi^{\otimes\caS^{\anchorino}_m}).
\end{split}
\end{align}
(This space was called $\caC_{\anchorino}^*$ in \cite{bols2025levinI}). Each boundary component $\caS^{\anchorino}_\kappa$ carries a unitary
$\Tube_{\caS^{\anchorino}_\kappa}$ action $\triangleright_\kappa$, making
$\caC_{\anchorino}(\Sigma)$ a unitary $A(\partial\Sigma)$-module. Define a linear map (This map was called $\Phi_{\Sigma}^{\anchorino, *}$ in \cite{bols2025levinI})
\begin{equation}\label{eq:Phi_anchor_defined}
\Phi_{\Sigma}^{\anchorino} : \al \otimes w_1 \otimes \dots \otimes w_m
\;\mapsto\;
\Bigl(\prod_{\kappa=1}^m d_{X_\kappa}\Bigr)^{1/2}\,\caD^{m-1}\;
\adjincludegraphics[valign=c,width=6.0cm]{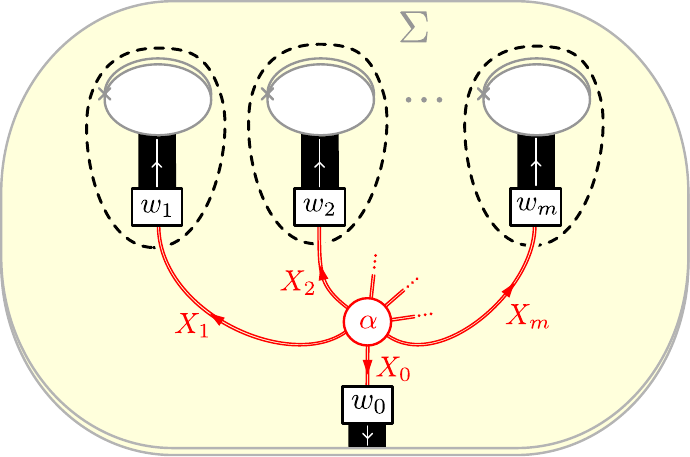},
\end{equation}
where $\al$ sits at the anchor point and each strand $X_\kappa$ follows the
$\kappa$-th edge to resolve into $w_\kappa$ near $\caS^{\anchorino}_\kappa$.
The canonical isomorphism
$Z(\caC)(X_0^*\to Y)\simeq Z(\caC)(\I\to X_0\otimes Y)$ is used to interpret
$\al$ as a morphism from $\I$.  
Equivalent anchors yield the same map $\Phi_\Sigma^{\anchorino}$.

\begin{proposition}[{\cite[Prop. 3.5]{bols2025levinI}}]\label[proposition]{prop:characterization of skein modules}
The map
\[
\Phi_\Sigma^{\anchorino} : \caC_{\anchorino}(\Sigma)\longrightarrow A(\Sigma)
\]
is a unitary isomorphism of $A(\partial\Sigma)$-modules.
\end{proposition}

\subsection{Skein subspaces}

\subsubsection{Regions and string-net subspaces}\label{ssubsec:regions and string-net subspaces}

Let $C^{\Z^2}$ be the standard cell complex with vertices $C_0^{\Z^2}=\Z^2$,
edges $\caE$, and faces $\caF$. The set of oriented edges is $\vec E_C := \{ e\in \vec\caE : e\in E_C \text{ or } \bar e \in E_C \}$. A \emph{region} is a subcomplex $C\subset
C^{\Z^2}$, with vertex, edge, and face sets denoted $V_C$, $E_C$, $F_C$. We write $\caA_C = \caA_{V_C}$ for the algebra of observables supported on $C$, and if $C$ is finite, $\caH_C := \caH_{V_C}$.  

The \emph{string-net subspace} associated with a finite region is
\[
H_{C_1} := \Bigl(\prod_{e\in E_C} A_e\Bigr)\caH_C,
\]
depending only on the 1-skeleton $C_1$ of $C$.

\subsubsection{Extended surfaces assigned to regions}\label{ssubsec:extended surface of region}

To each region $C$ we assign the extended surface
\[
\Sigma_C := \{ v\in\R^2 : \dist_\infty(v,C) \le 1/3\},
\quad \dist_\infty(v,C) := \inf_{w\in C} \|v-w\|_\infty,
\]
the $1/3$-fattening of $C \subset \R^2$. The marked boundary points are the intersections of $\partial\Sigma_C$ with edges
of $\Z^2$. A marked boundary point is \emph{positive} if the corresponding edge (oriented
toward the top-left) exits $\Sigma_C$, and \emph{negative} if it enters $\Sigma_C$.

Each connected boundary component $\caN\in\Bd(\Sigma_C)$ is a union of
horizontal and vertical segments. For compact components, the fiducial point is
placed at the top end of top leftmost vertical segment; for non compact
components, it is chosen arbitrarily.  

When $C$ is finite and connected, $\Sigma_C$ has a distinguished \emph{outer}
boundary; the other boundaries are called \emph{inner}.

\subsubsection{Isomorphism of string-net subspaces with skein modules}
\label{ssubsec:isomorphism of string-net with skein}

Let $C$ be a finite region. Product vectors in the string-net subspace $H_{C_1}$
naturally correspond to string diagrams on $\Sigma_{C_1}$.  
For a labeling $l\in\caL_{SN}(C)$, let $f_v\in\caH_v(l)$ be unit vectors for each
vertex $v\in C_0$, and set $f=(f_v)_{v\in C_0}$ with
\[
\phi_f := \bigotimes_{v\in C_0} f_v \in \caH_C(l) \subset H_{C_1}.
\]
Define the associated string diagram $x_f$ on $\Sigma_{C_1}$ using the intersection
of the $\Z^2$ graph with $\Sigma_{C_1}$, edges directed top-right, labeled by
$l$, and vertices labeled by $f_v$.  
The map
\begin{equation}\label{eq:pi_C defined}
\pi_{C_1}(\phi_f) := [x_f]_{\Sigma_{C_1}}
\end{equation}
identifies product vectors with their diagrams.

Also define
\begin{equation}\label{eq:sigma_C defined}
\sigma_{C_1}(\phi_f) := d_{\partial l}^{-1/4}\,[x_f]_{\Sigma_{C_1}},
\qquad
d_{\partial l} := \prod_{e\in \vec\partial C} d_{l(e)},
\end{equation}
and extend linearly. By construction, $\sigma_C(H_{C_1}^b) \subset A(\Sigma_{C_1};b)$
for any boundary condition $b$.  

The following Lemma combines \cite[~Lem. A. 5]{bols2025levinI} and a special case of \cite[~Prop 4.4]{bols2025levinI}:
\begin{lemma}\label[lemma]{lem:string-net isomorphism}
    The map $\pi_{C_1}: H_{C_1} \to A(\Sigma_{C_1})$ is a vector space isomorphism,
    and $\sigma_{C_1}: H_{C_1} \to A(\Sigma_{C_1})$ is a Hilbert space isomorphism.
\end{lemma}

\begin{convention}\label[convention]{conv:graphical representation}
    We freely use $\pi_{C_1}$ to represent states in $H_{C_1}$ by string diagrams
    on $\Sigma_{C_1}$, employing the graphical calculus.  
    Conversely, for a string diagram $x$ on $\Sigma_{C_1}$, write
    \[
    \phi_x := \pi_{C_1}^{-1}([x])
    \]
    for the corresponding vector in $H_{C_1}$.
\end{convention}

\subsubsection{Actions of Tube algebras on string-net subspaces}
\label{ssubsec:Tube actions on collars}

Let $C$ be a region such that $\Sigma_C$ has compact boundary.  
For each connected boundary component $\caS \in \Bd(\Sigma_C)$, define the
\emph{collar region} $C^{\caS}$ as the minimal subregion with
$\caS \in \Bd(\Sigma_{C^{\caS}})$.  

If $C$ contains a face $f$, removing it yields a boundary component $\caS_f$; the
corresponding collar region is $C^f := C^{\caS_f}$.  
Similarly, removing an edge along with its two neighbouring faces belonging to $C$, 
produces a boundary component $\caS^e$ with collar region $C^e := C^{\caS_e}$. We refer to $\caS_e$ as the puncture at $e$.

For each $\caS \in \Bd(\Sigma_C)$, define a representation
\[
\frt_{\caS} : \Tube_{\caS} \to \mathrm{End}(\caH_{C^{\caS}})
\]
by letting $\frt_{\caS}(a)$ vanish on $H_{C^{\caS}}^\perp$, and act on the string-net
subspace $H_{C^{\caS}}$ by
\begin{equation}\label{eq:frt_caS defined}
\frt_{\caS}(a)\, \phi := \sigma_{C^{\caS}}^{-1} \bigl( a \triangleright_{\caS} \sigma_{C^{\caS}}(\phi) \bigr),
\quad \phi \in H_{C^{\caS}}.
\end{equation}
By Lemma~\ref{lem:string-net isomorphism}, each $\frt_{\caS}$ is a *-representation
of $\Tube_{\caS}$.

We write $\frt_f := \frt_{\caS_f}$ and $\frt_e := \frt_{\caS_e}$
for any face $f$ and edge $e$.

\subsubsection{Definition of \texorpdfstring{$B_f$}{Bf} projectors}
\label{ssubsec:Bf defined}

For a boundary component $\caS_f$ corresponding to a face $f\in\caF$, let
$\frt_f := \frt_{\caS_f}$ be the associated $\Tube_{\caS_f}$-representation on
$\caH_{C^f}$.  
Define the orthogonal projector
\begin{equation}\label{eq:Bf defined}
B_f := \frt_f(P^{\I}), \qquad
P^{\I} = \adjincludegraphics[valign=c,width=0.8cm]{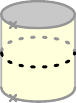} \in \Tube_{\caS_f}.
\end{equation}

\subsubsection{Skein subspaces} \label{ssubsec:skein subspaces defined}

Let $C$ be a finite region. Its \emph{skein subspace} is
\begin{equation}\label{eq:skein subspace defined}
H_C := B_C H_{C_1} = P_C \caH_C,
\end{equation}
where $B_C := \prod_{f\in F_C} B_f$ and $P_C := \prod_{f\in F_C} B_f \,\,
\prod_{e\in E_C} A_e$. This subspace inherits the skein inner product from
$\caH_C$.

Since local relations and isotopies in $\Sigma_{C_1}$ are also valid in $\Sigma_C$, the inclusion $\Sigma_{C_1} \subset \Sigma_C$ induces a well-defined surjective map
\[
\iota_C : A(\Sigma_{C_1}) \to A(\Sigma_C), \quad [y]_{\Sigma_{C_1}} \mapsto [y]_{\Sigma_C}.
\]
Recalling \eqref{eq:sigma_C defined}, define
\[
\sigma_C := \caD^{-|F_C|}\, \iota_C \circ \sigma_{C_1} : H_{C_1} \to A(\Sigma_C).
\]

\begin{proposition}[{\cite[Prop. 4.4]{bols2025levinI}}]
\label[proposition]{prop:isomorphism of skein subspace and skein module}
The map $\sigma_C|_{H_C} : H_C \to A(\Sigma_C)$ is unitary.  
For any connected boundary $\caS \in \Bd(\Sigma_C)$, we have
\[
a \triangleright_{\caS} \sigma_C(\psi) = \sigma_C(\frt_{\caS}(a) \psi),
\quad \psi \in H_C,\, a \in \Tube_{\caS}.
\]
Thus, $H_C$ with the $\Tube$-actions $\{\frt_{\caS}\}$ is a unitary
$A(\partial \Sigma_C)$-module, and $\sigma_C|_{H_C}$ is an intertwiner of unitary modules.
\end{proposition}

\begin{lemma}[{\cite[Lemma 4.5]{bols2025levinI}}]
\label[lemma]{lem:commutativity lemma}
For distinct boundary components $\caS, \caS' \in \Bd(\Sigma_C)$ and any
$a \in \Tube_{\caS}$, $b \in \Tube_{\caS'}$,
\[
    [\frt_{\caS}(a), \frt_{\caS'}(b)] = 0.
\]
In particular, all $B_f$ projectors commute.
\end{lemma}

A \emph{boundary condition} $\underline p = \{p_{\caS}\}_{\caS \in \Bd(\Sigma_C)}$ for $H_C$ assigns a projector $p_{\caS} \in \Tube_{\caS}$ to each boundary component.
Define
\begin{equation}\label{eq:def H_C(p)}
H_C(\underline p) := \Bigl(\prod_{\caS \in \Bd(\Sigma_C)} \frt_{\caS}(p_{\caS})\Bigr) H_C,
\quad
P_C(\underline p) := B_C \prod_{\caS \in \Bd(\Sigma_C)} \frt_{\caS}(p_{\caS}).
\end{equation}
Given an enumeration $\{\caS_\kappa\}_{\kappa=0}^m$ of $\Bd(\Sigma_C)$, we also write
\[
H_C(p_{\caS_0},\dots,p_{\caS_m}) = H_C(\underline p),\quad
P_C(p_{\caS_0},\dots,p_{\caS_m}) = P_C(\underline p).
\]

\subsubsection{Characterization of skein subspaces} 
\label{ssubsec:characterization of skein subspaces}

Let $C$ be a finite connected region. Then $\Sigma_C$ is homeomorphic to a disk with $m$ holes for some $m \in \N_0$.  
We refer to an anchor for $\Sigma_C$ as an anchor for $C$.

\begin{proposition}[{\cite[Prop. 4.6]{bols2025levinI}}]
\label[proposition]{prop:the great interface}
Let $C$ be a finite connected region with anchor $\anchor$. Then
\begin{equation}\label{eq:great interface isomorphism}
\Psi_C^{\anchorino} := \sigma_C^{-1} \circ \Phi_C^{\anchorino} : 
\caC^*_{\anchorino}(\Sigma_C) \longrightarrow H_C
\end{equation}
is a unitary isomorphism of $A(\partial \Sigma_C)$-modules.
\end{proposition}

\subsection{Drinfeld insertions} \label{subsec:Drinfeld insertions}

For each $X \in \Irr Z(\caC)$ and each extended circle $\caS$, fix a unit vector
$w^X_{\caS} \in \caC(X \rightarrow \chi^{\otimes \caS})$ with respect to the trace inner product, 
and let $p_{\caS}^X \in \Tube_{\caS}$ be the corresponding minimal projector 
(see Proposition~\ref{prop:matrix units for Tube_n}).  

For $X = \I$ and $\caS = \caS_e$, we take
\begin{equation} \label{eq:ground state boundary condition at puncture}
w^{\I}_{\caS_e} = \frac{1}{\caD} \sum_{a \in \Irr \caC} d_a^{1/2} 
\adjincludegraphics[valign=c, width=1.0cm]{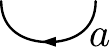},
\end{equation}
and define
\begin{equation} \label{eq:p^one defined}
p^{\I}_{\caS_e} := v v^*, \quad v := \frac{1}{\caD} \sum_{a \in \Irr \caC} d_a^{1/2} 
\adjincludegraphics[valign=c, width=2.0cm]{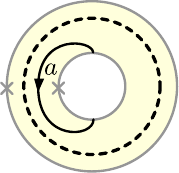} \in \Tube_{\caS_e}.
\end{equation}
Here $\caS_e \times I$ is presented as an annulus with the outer boundary identified with the bottom of $\caS_e \times I$. By \cite[Lemma 5.4]{bols2025levinI} $\frt_e(p^{\I}_{\caS_e}) = B_{f_1} B_{f_2}$ enforces ground state constraints on the two faces bordering $e$. We drop the $\caS$ subscript when it is clear from context: $w^X := w^X_{\caS}$ and $p^X := p^X_{\caS}$.

\begin{convention} \label[convention]{conv:Drinfeld strands attaching to boundary components}
Given the fixed boundary vectors $w^X$, we adopt the following graphical convention for string diagrams:
\begin{equation}
\adjincludegraphics[valign=c, height=1.0cm]{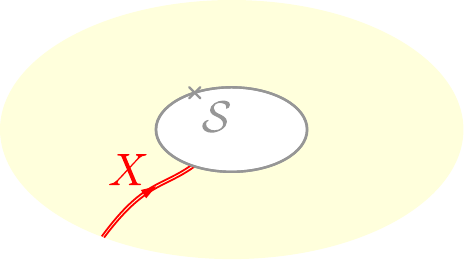} 
= \adjincludegraphics[valign=c, height=1.0cm]{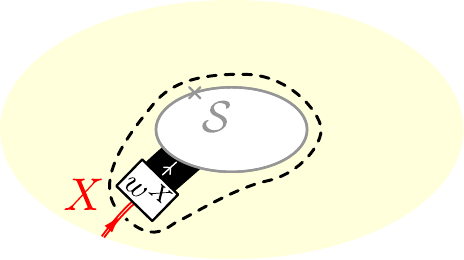}, 
\quad
\adjincludegraphics[valign=c, height=1.0cm]{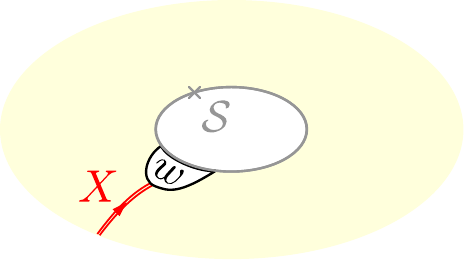} 
= \adjincludegraphics[valign=c, height=1.0cm]{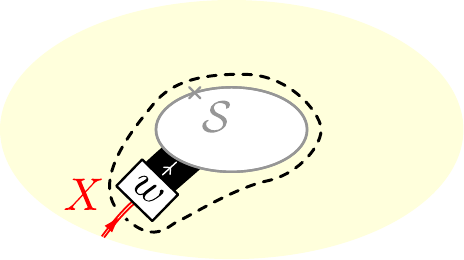}.
\end{equation}
\end{convention}

Let $C$ be a finite connected region with anchor $\anchor$. Proposition~\ref{prop:the great interface} implies that vectors of the form
\begin{align*}
\Psi_C^{\anchorino} \big( \al \otimes w_0 & \otimes w^{X_1} \otimes \cdots \otimes w^{X_m} \big) \\ 
&= \left( \prod_{\kappa=0}^m d_{X_\kappa} \right)^{1/2} \caD^m 
\sigma_C^{-1} \left( \adjincludegraphics[valign=c, height=2.0cm]{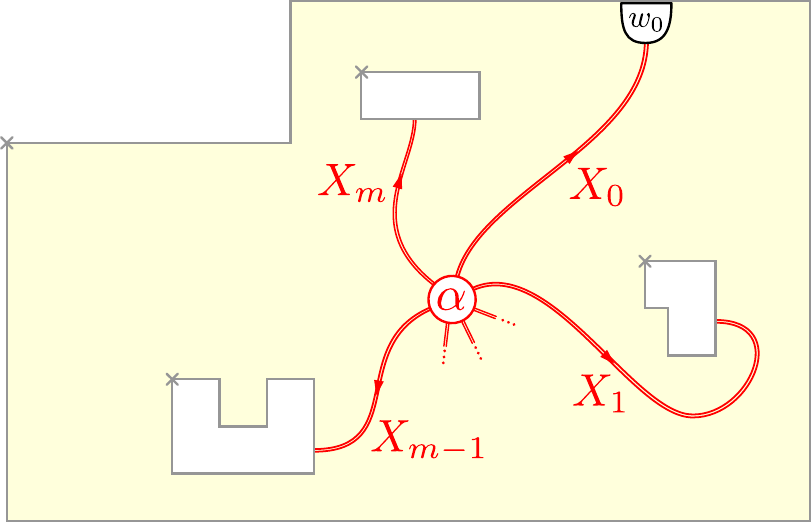} \right),
\end{align*}
for $X_0 \in \Irr Z(\caC)$, $w_0 \in \caC(X_0 \rightarrow \chi^{\otimes \caS_0^{\anchorino}})$ 
and $\al \in Z(\caC)(X_0^* \rightarrow X_1 \otimes \cdots \otimes X_m)$ span
\begin{equation} \label{eq:skein subspace with free outer boundary}
H_C^{(X_1, \dots, X_m)} := H_C(\id, p^{X_1}, \dots, p^{X_m}).
\end{equation}
The orthogonal projector onto this subspace is $P_C^{(X_1, \dots, X_m)} := P_C(\id, p^{X_1}, \dots, p^{X_m})$.

For any $\beta \in Z(\caC)(X_1 \otimes \dots \otimes X_m \rightarrow Y_1 \otimes \dots \otimes Y_m)$, 
the \emph{Drinfeld insertion} $\Dr_C^{\anchorino}[\beta] \in \caA_C$ acts on $H_C^{(X_1, \dots, X_m)}$ by
\begin{align} \label{eq:Drinfeld insertion defined}
\Dr_C^{\anchorino}[\beta] : 
\Psi_C^{\anchorino} \big( \al \otimes w_0 \otimes w^{X_1} \otimes \dots \otimes w^{X_m} \big) 
\mapsto 
\Psi_C^{\anchorino} \big( (\beta \circ \al) \otimes w_0 \otimes w^{Y_1} \otimes \dots \otimes w^{Y_m} \big),
\end{align}
and annihilates the orthogonal complement of $H_C^{(X_1, \dots, X_m)}$.  
Graphically, $\sigma_C \circ \Dr_C^{\anchorino}[\beta] \circ \sigma_C^{-1}$ is represented by
\begin{equation} \label{eq:Drinfeld insertion graphical}
\adjincludegraphics[valign=c, height=1.7cm]{graphical_vector_in_H_C.pdf} \mapsto 
\left( \prod_{\kappa=1}^m \frac{d_{Y_\kappa}}{d_{X_\kappa}} \right)^{1/2} 
\adjincludegraphics[valign=c, height=1.7cm]{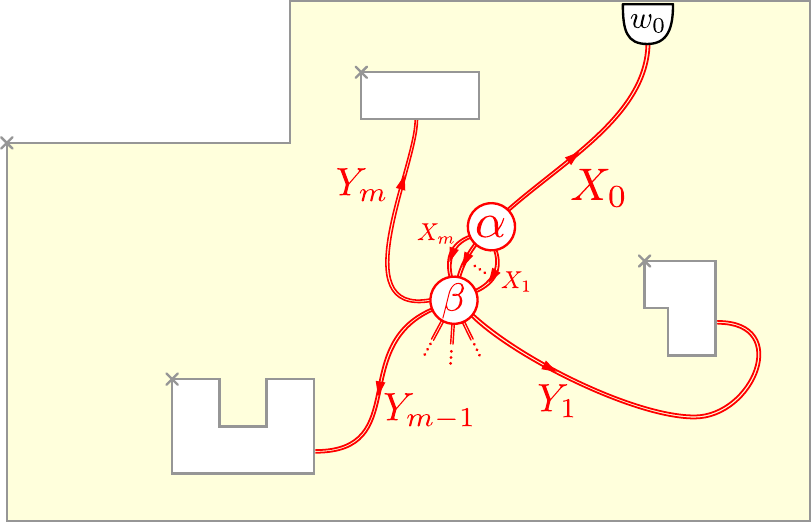}.
\end{equation}

\begin{lemma}[Multiplicativity {\cite[Lemma 6.2]{bols2025levinI}}]
\label[multiplicativity lemma]{lem:Drinfeld insertion multiplicativity}
Let $C$ be a finite connected region with $|\Bd(\Sigma_C)| = m+1$ and anchor $\anchor$. 
If $\Theta = \bigoplus_{X \in \Irr(Z(\caC))} X$, then
\[
\Dr_C^{\anchorino} : \End_{Z(\caC)}(\Theta^{\otimes m}) \longrightarrow \End\Big( \bigoplus_{X_1, \dots, X_m} H_C^{(X_1, \dots, X_m)} \Big)
\]
defines a *-representation.
\end{lemma}

\subsection{The Inclusion Lemma} \label{subsec:inclusion lemma}

A \emph{dual path} $I = \{f_i\}_{i=1}^l$ is a sequence of faces with
$f_i$ neighbouring $f_{i+1}$ for all $i$, and is \emph{self-avoiding} if all
$f_i$ are distinct. Denote $\partial_{\ii} I = f_1$, $\partial_{\f} I = f_l$,
and $I_{\inn} = \{ f_2, \dots, f_{l-1} \}$ (possibly empty). Let $\caE_I \subset
\caE$ be the edges between successive faces of $I$.

Let $D$ be a finite connected region. We say $D$ can be cut along a
self-avoiding $I$ if $I_{\inn} \subset D$, $\caE_I \subset D$, and the initial
and final faces lie on the outer boundary of $D$. Define
\[
D \setminus I := D \setminus \Bigl(I_{\inn} \cup \bigcup_{e \in \caE_I} e\Bigr).
\]
Then $D \setminus I = C \sqcup C'$ for uniquely determined disjoint connected regions $C$ and $C'$, with all inner boundary components of $\Sigma_D$ inherited from $C$ or $C'$. We say $D$ is obtained by gluing $C$ and $C'$ along $I$, writing $D= C\sqcup_I C'$.

All inner boundary components of $\Sigma_C$, resp. $\Sigma_{C'}$, identify with inner boundary components of $\Sigma_D$.  
The outer boundaries of $\Sigma_C$ and $\Sigma_{C'}$ each contain a \emph{boundary interval} $\caI$ and $\caI'$ determined by $I$. Let $m_{C,I}$ (resp. $m_{C',I}$) be the marked boundary points of $\Sigma_C$ (resp. $\Sigma_{C'}$) lying on edges of $\caE_I$, and define $\caI$ (resp. $\caI'$) as the minimal closed subinterval of $\partial \Sigma_C$ (resp. $\partial \Sigma_{C'}$) containing the balls $B_{1/3}(m)$ for each $m \in m_{C,I}$ (resp. $m_{C',I}$).  

Gluing $\caI$ to $\caI'$ yields a surface homeomorphic to $\Sigma_D$, which we identify with $\Sigma_D$. The induced gluing map
\[
\gl : A(\Sigma_C \sqcup \Sigma_{C'}) \simeq A(\Sigma_C) \otimes A(\Sigma_{C'}) \longrightarrow A(\Sigma_D)
\]
is as in Section~\ref{ssubsec:gluing}.

Let $\anchor_C$ be an anchor of $\Sigma_C$ whose attachment point to the outer boundary component $\caS^{\anchorino_C}_0$ of $\Sigma_C$ also lies on the outer boundary of $\Sigma_D$. Let $\anchor_D$ be an anchor for $\Sigma_D$ such that $\caS^{\anchorino_D}_0$ is the outer boundary of $\Sigma_D$ and such that $\caS^{\anchorino_C}_{\kappa} = \caS^{\anchorino_D}_{\kappa + \lambda}, \quad \text{for } \kappa = 1, \dots, m$ for some \emph{offset} $\lambda \in \{ 0, 1, \dots, n-m \}$.  
If the graph of $\anchor_D$ contains the graph of $\anchor_C$ as a subgraph, we say that $\anchor_D$ \emph{extends} $\anchor_C$ with offset $\lambda$. More generally, any anchor $\anchor'_D \sim \anchor_D$ is also said to extend $\anchor_C$ with offset $\lambda$.  

Recall that $P^{(X_1, \dots, X_n)}_D$ is the orthogonal projector onto the skein subspace $H_D^{(X_1, \dots, X_n)} \subset \caH_D$ defined in Eq.~\eqref{eq:skein subspace with free outer boundary}.  The following \emph{Inclusion Lemma} is crucial for manipulating Drinfeld insertions defined on different regions.

\begin{lemma}[Inclusion {\cite[Lemma 6.7]{bols2025levinI}}] 
\label[inclusion lemma]{lem:inclusion lemma} 
Let $D=C\sqcup_I C'$ (so $D$ is obtained by gluing $C$ and $C'$ along a self-avoiding dual path $I$), and let $\anchor_D$ extend $\anchor_C$ with offset $\lambda$. For any $X_1, \dots, X_n, Y_1, \dots, Y_n \in \Irr Z(\caC)$ and any $\beta \in Z(\caC)\big( X_{\lambda+1} \otimes \cdots \otimes X_{\lambda+m} \rightarrow Y_1 \otimes \cdots \otimes Y_m \big)$ we have
\[
\Dr^{\anchorino_C}_C[\beta] \, P^{(X_1, \dots, X_n)}_D 
= \Dr_{D}^{\anchorino_D}\big[ \id_{X_1 \otimes \cdots \otimes X_{\lambda}} \otimes \beta \otimes \id_{X_{\lambda + m + 1} \otimes \cdots \otimes  X_n} \big].
\]
\end{lemma}

\begin{remark} \label[remark]{rem:adding vacuum lines}
    With care, the Inclusion Lemma can be extended with more flexible ways of defining the extension of anchors. We mention one trivial but useful relation that is not covered by the Inclusion Lemma as stated (see also discussion at Eq. (58) of \cite{bols2025levinI}). 

    Let $C$ be a finite connected region with anchor $\anchor$ and assume that $\caS^{\anchorino}_{k+1} = \caS_e$ for some edge $e$ and number $k$. Let $C' \supset C$ be the finite subregion obtained from $C$ by filling in the puncture at $e$, and let $\anchor'$ be the unique subanchor of $\anchor$ which is an anchor on $C'$. It follows immediately from the definitions and the choice of vacuum boundary condition $p^{\I}$ in Eq. \eqref{eq:p^one defined} that
    $$  P_{C'}^{(X_1, \cdots, X_m)} = P_{C}^{(X_1, \cdots, X_k, \I, X_{k+1}, \cdots, X_m)}, $$
    where $m+1 = \abs{\Bd(\Sigma_{C'})}$. Moreover, given a morphism $\beta : X_1 \otimes \cdots \otimes X_m \rightarrow Y_1 \otimes \cdots \otimes Y_m$ we have
    $$ \Dr^{\anchorino'}_{C'} \left[ \adjincludegraphics[valign=c, height = 0.8cm]{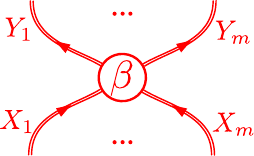} \right] P_{C'}^{(X_1, \cdots, X_m)} = \Dr^{\anchorino}_{C} \left[  \adjincludegraphics[valign=c, height = 0.8cm]{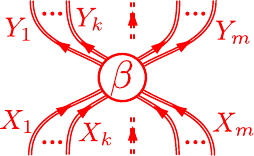}  \right] P_{C}^{(X_1, \cdots, X_k, \I, X_{k+1}, \cdots, X_m)}.$$
\end{remark}

\section{Simple objects of \texorpdfstring{$\DHR_f$}{DHR} from string operators}  \label{sec:simple objects of DHR}

We recall the construction of string operators from \cite{bols2025levinI} in \Cref{sec:string operators}, and use these to construct a representative set of simple objects for $\DHR_f$ in \Cref{subsec:representative set for DHR}.

\subsection{String operators} \label{sec:string operators}

\subsubsection{Links} \label{ssubsec:links}

Recall from \cite[~Section 6.3]{bols2025levinI} that a \emph{link} is a finite dual path $L = (f_1, \cdots, f_l)$ consisting of at least five faces and such that faces $f_i$ and $f_j$ share an edge if and only if $|i-j|\le 1$, and such that the first four faces of $L$ lie on a straight line, as do the last four faces of $L$.

Let $e_{\ii} = \partial_{\ii} L$ denote the edge between faces $f_1$ and $f_2$ and $e_{\f} = \partial_{\f} L$ the edge between the last two faces $f_{l-1}$ and $f_l$. The puncture at $e_{\ii}$, resp. $e_{\f}$, is called the \emph{initial}, resp. \emph{final}, puncture of $L$.

The region $C^L$ is the region (see Section \ref{ssubsec:regions and string-net subspaces}) consisting of the \emph{bulk faces} $F_L := \{f_3, \cdots, f_{l-2}\}$ of $L$, all edges belonging to faces of $L$ except $e_{\ii}$ and $e_{\f}$, as well as all the vertices belonging to such edges. The associated extended surface $\Sigma_L := \Sigma_{C^{L}}$ (see Section \ref{ssubsec:extended surface of region}) is a twice punctured disk, with punctures at $\partial_{\ii}L$ and $\partial_{\f} L$.

For any such link $L$ we fix an anchor $\anchor_L$ of $C_L$ (see Sections \ref{ssubsec:skein modules on punctured disks} and \ref{ssubsec:characterization of skein subspaces}) whose equivalence class up to the outer boundary component $\caS_0^{\anchorino_L}$ (Definition \ref{def:equivalence of anchors}) is uniquely determined by the condition that the underlying graph of $\anchor_L$ lies in the right strip $\Sigma_{L, r} := \Sigma_{C^{L, r}}$, where $C^{L, r}$ is the subregion of $C^L$ consisting of all vertices and edges on the right hand side of $C^L$ w.r.t. the orientation of $L$. (See \cite[Section 6.3]{bols2025levinI})

\begin{figure}[ht]
    \centering
    \includegraphics[width=0.4\textwidth]{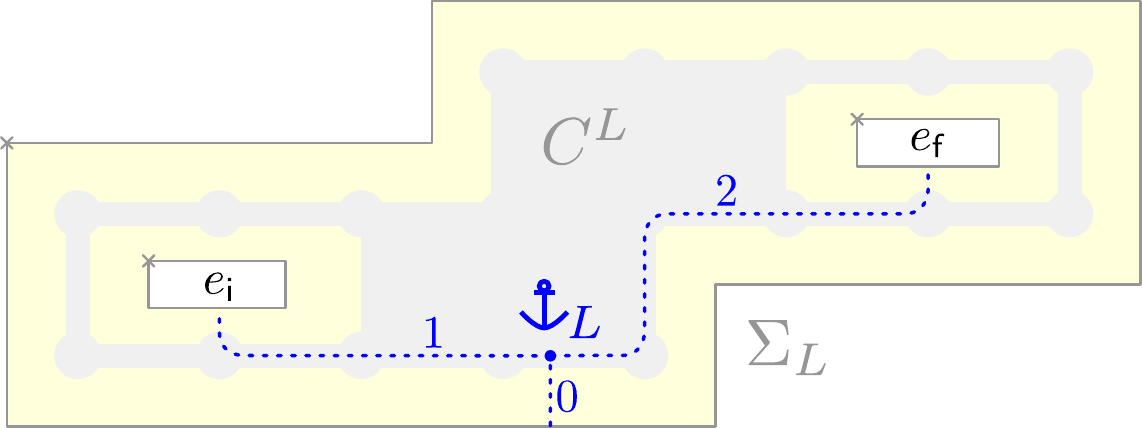}
    \caption{A link $L$ together with the region $C^L$ indicated in light grey, and the associated surface $\Sigma_L$ in light yellow. The anchor $\anchor_L$ is shown in blue.}
    \label{fig:link}
\end{figure}

We write $\supp(L) = V_{C^L}$ for the \emph{support} of the link $L$. We will often write $L$ instead of $C^L$ in notations depending on the region $C^L$. For example, $\caA_L := \caA_{C^L}, \caH_L := \caH_{C^L}$, $H_L := H_{C^L}$, $\sigma_L := \sigma_{C^L}$, $\Psi_L^{\anchorino} := \Psi_{C^L}^{\anchorino}$, etc. Consistent with this notation, we have $F_L = F_{C^L} = \{ f_3, \cdots, f_{l-2} \}$. 
We also suppress the standardised choice of anchor $\anchor_L$, writing $\Dr_L[\beta] = \Dr_{C^L}^{\anchorino_L}[\beta]$.

The boundary components of $\Sigma_L$ will always have the linear ordering induced by $\anchor_L$, so that $H_L(p_0, p_1, p_2) = H_L(\underline p)$ with $p_{\kappa} = p_{\caS^{\anchorino_L}_{\kappa}}$ for $\kappa = 0, 1, 2$. We also write $\frt^L_{\kappa} := \frt_{\caS_{\kappa}^{\anchorino_L}}$ for the corresponding $\Tube$-action on these boundary components. (see Section \ref{ssubsec:Tube actions on collars})

Two links $L_1 = (f_1, \cdots, f_k)$ and $L_2 = (f'_1\cdots, f'_l)$ are said to be \emph{composable} if $\partial_{\ii}(L_1)=\partial_\f(L_2)$ and the \emph{composite} $L_2 \wedge L_1 := (f'_1, \cdots, f'_{l-1}, f'_l, f_2, \cdots, f_k)$ is again a link.

\subsubsection{Hopping, pair creation, and the definition of unitary gates} \label{ssubsec:hopping and pair creation}

For each $X \in \Irr Z(\caC)$ we fix a unitary morphism $\zeta_X : X^* \rightarrow \bar X$, and represent it and its inverse graphically as
\begin{equation} \label{eq:star to bar}
    \zeta_X = \,\,\, \adjincludegraphics[valign=c, height = 0.8cm]{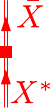} \,\,\, = \,\,\, \adjincludegraphics[valign=c, height = 0.8cm]{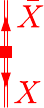} \,\,,\quad\quad \zeta^{-1}_X = \zeta^{\dag}_X = \,\,\, \adjincludegraphics[valign=c, height = 0.8cm]{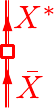} \,\,\, = \,\,\, \adjincludegraphics[valign=c, height = 0.8cm]{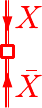}\,.
\end{equation}
We also put
\begin{equation} \label{eq:graphical gagger duals}
    \ev_X^{\dag} = \,\,\, \adjincludegraphics[valign=c, width = 1.5cm]{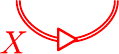} \,\,\, , \,\,\, \coev_X^{\dag} = \,\,\, \adjincludegraphics[valign=c, width = 1.5cm]{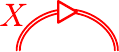} \,\, 
\end{equation}
and define Drinfeld insertion operators (see Section \ref{subsec:Drinfeld insertions})
\begin{align}
	\Dr_L^{(\I\I \rightarrow X \bar X)} &:= d_X^{-1/2} \, \Dr_{L} \left[ \,\,\,\,\, \adjincludegraphics[valign=c, height = 0.8cm]{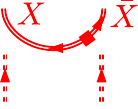} \right], \\
	\Dr_L^{(X \bar X \rightarrow \I \I)} &:= (\Dr_L^{(\I\I \rightarrow X \bar X)})^* = d_X^{-1/2} \, \Dr_{L} \left[ \,\,\,\,\, \adjincludegraphics[valign=c, height = 0.8cm]{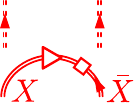} \right], \\
	\Dr_L^{(X \I \rightarrow \I X)} &:= \Dr_{L}\left[ \adjincludegraphics[valign=c, height = 1.0cm]{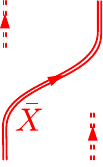} \right], \quad \Dr_L^{(\I X \rightarrow X \I)} := (\Dr_L^{(\I X \rightarrow X \I)})^* = \Dr_{L}\left[ \adjincludegraphics[valign=c, height = 1.0cm]{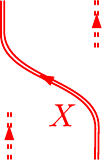} \right].
\end{align}

By \cite[~Lemma 6.8]{bols2025levinI} these Drinfeld insertions are partial isometries with
\begin{align}\label{eq:P_L defs1}
	P_L^{\I \I} &:= \frt^L_1( p^{\I} ) \frt^L_2(p^{\I}) B_L
        = (\Dr_L^{(\I\I \rightarrow X \bar X )})^* \Dr_L^{(\I\I \rightarrow X \bar X )}, \\
        P_L^{X \bar X} &:= \frt^L_0(P^{\I}) \frt^L_1( p^X ) \frt^L_2(p^{\bar X}) B_L
        = \Dr_L^{(\I\I \rightarrow X \bar X )} (\Dr_L^{(\I\I \rightarrow X \bar X )})^* \\
	P_L^{X \I} &:= \frt^L_1( p^{X} ) \frt^L_2(p^{\I}) B_L 
        = (\Dr_L^{(X \I \rightarrow \I X)})^* \Dr_L^{(X \I \rightarrow \I X)}, \\
        P_L^{\I X} &:= \frt^L_1( p^{\I} ) \frt^L_2(p^{X}) B_L 
        = \Dr_L^{(X \I \rightarrow \I X)} (\Dr_L^{(X \I \rightarrow \I X)})^*.
\label{eq:P_L defs4}
\end{align}
where $B_L = \prod_{f \in F_L} B_f$ imposes $B_f$ for the bulk faces of $L$.

Write further $P^{\I ?}_L := P_L^{\I \I} + P_L^{\I X}$,
$P^{X ?}_L := P_L^{X \I} + P_L^{X \bar X }$,
and $P^{??}_L = P^{X ?}_L + P^{\I ?}_L = P^{\I\I} + P_L^{X \I} + P_L^{\I X} + P_L^{X \bar X}$, and define a self-adjoint unitary gate
\begin{equation} \label{eq:unitary gate defined}
    u_L^X := \Dr_L^{(\I \I \rightarrow X \bar X)} + \Dr_L^{(X \bar X \rightarrow \I \I)} + \Dr_L^{(X \I \rightarrow \I X)} + \Dr_L^{(\I X \rightarrow X \I)} + (\I - P^{??}_L)
\end{equation}
associated to the link $L$. By construction, $u_L$ satsifies the following intertwining properties:
\begin{equation} \label{eq:intertwining properties of u_L}
    u^X_L P_L^{\I \I} = P_L^{X \bar X} u^X_L, \quad \quad u^X_L P_L^{ \I X} = P_L^{X \I } u^X_L.
\end{equation}

These unitary gates also satisfy the \emph{concatenation lemma}:
\begin{lemma}[\cite{bols2025levinI}, Lemma 6.9] \label[concatenation lemma]{lem:concatenation lemma}
    Let $X\in \Irr Z(\caC)$ and let $L_1$ and $L_2$ be composable links with $L = L_2 \wedge L_1$. Then
    $$ u^X_{L_2} u^X_{L_1} P_L^{\I ?} = u_L^X P_L^{\I ?}. $$
\end{lemma}

\subsubsection{String operators as endomorphisms of \texorpdfstring{$\caA$}{A}.} \label{ssubsec:definition of string operators}

A \emph{chain} $\scrC = ( L_n )_{n \in \N}$ is a half-infinite sequence of links $L_n$ such that $L^{\scrC}_{\interval{m}{n}} := L_n \wedge L_{n-1} \wedge \cdots \wedge L_m$ is a well-defined link for all natural numbers $m \leq n$. We write $\scrC[\interval{m}{n}]=(L_i)_{i=m}^n$ for the finite subchains of $\scrC$ as well as $\scrC_n := \scrC[\interval{1}{n}]$, and $\partial_{\ii} \scrC = \partial_{\f} L_1$ for the initial puncture of the chain, which is the final puncture of the first link. We also write $\supp(\scrC) = \bigcup_{n \in \N} \supp(L_n)$ for the support of the chain $\scrC$.

For any $X \in \Irr Z(\caC)$ and natural numbers $m \leq n$ we construct a finite unitary circuit $$U^X_{\scrC[\interval{m}{n}]} = u^X_{L_n} \times \cdots \times u^X_{L_m}$$
out of the gates of Eq. \eqref{eq:unitary gate defined}. We also put $\rho^X_{\scrC_n} := \Ad[ (U^X_{\scrC[\interval{1}{n}]})^* ]$.

\begin{lemma}[\cite{bols2025levinI}, Lemma 6.10] \label[lemma]{lem:limiting endomorphism}
    Let $\scrC$ be a chain and $X \in \Irr Z(\caC)$.

    The limit
	$$ \rho_{\scrC}^X := \lim_{n \uparrow \infty} \, \rho_{{\scrC}_n}^X,$$
	exists pointwise in norm and defines a unital *-endomorphism of $\caA$. This endomorphism is supported on $\supp({\scrC})$ in the sense that if $x \in \caA_{\supp({\scrC})^c}$ then $\rho_{\scrC}^X(x) = x$.
\end{lemma}

\subsection{Representative sets of simple objects of \texorpdfstring{$\DHR_f$}{DHRf}}
\label{subsec:representative set for DHR}

By \cite[Proposition 7.2, Proposition 7.10]{bols2025levinI}, $\pi_{\scrC}^X = \pi^{\I} \circ \rho_{\scrC}^X$ is an anyon representation localized in any cone that contains the support of $\scrC$. 
Let $\bar \rho_{\scrC}^X$ be its unique extension to $\caB$, as given by \Cref{lem:extension to B}, satisfying
$$\bar \rho_{\scrC}^X \circ \pi^{\I} = \pi^{\I} \circ \rho_{\scrC}^X.$$

Recall that $\pi^{\I}(\caA)$ is WOT-dense in $\caB$, so intertwiners among the extensions $\{\bar \rho_{\scrC}^X\}_{X \in \Irr Z(\caC)}$ are exactly the same as the intertwiners among the anyon representations $\{ \pi_{\scrC}^X \}_{X \in \Irr Z(\caC)}$. Therefore, we get an equivalent formulation of the main theorem of \cite{bols2025levinI}.
\begin{theorem}[{\cite[Thm. 2.3]{bols2025levinI}}]
\label[theorem]{thm:bijection of simples}
	Let $\scrC$ be a chain supported in an allowed cone. Then the set of endomorphisms $\{ \bar \rho_{\scrC}^X \}_{X \in \Irr Z(\caC)}$ is a complete set of representative simple objects for $\DHR_f$. In particular, under \Cref{ass:bounded spread Haag duality}, $\DHR_f$ is finite semisimple.
\end{theorem}

\section{Fusion}\label{sec:fusion}

The result of the previous section implies that $Z(\caC)$ and $\DHR_f$ are isomorphic as semisimple unitary categories. In this section, we show that $Z(\caC)$ and $\DHR_f$ have \emph{isomorphic $F$-symbols}. It is well known that this implies that $Z(\caC)$ and $\DHR_f$ are monoidally equivalent. We will actually be able to conclude that they are unitarily monoidally equivalent by Proposition \ref{prop:isomophic symbols implies isomorphic categories}.

\begin{notation} \label{not:fixed chain}
    Throughout this section we fix an arbitrary chain $\scrC = (L_n)$ which is supported in an allowed cone, 
    and denote $e_n = \partial_{\ii}L_n$.
    We will write $\bar \rho^X = \bar \rho^X_{\scrC}$ for the corresponding simple endomorphisms, and simplify notation by writing $u_n^X = u_{L_n}^X$, $U_n^X = U_{\scrC_n}^X$, $U^X_{n \to m} = U^X_{\scrC[n \to m]}$, and $P_n^{\bullet \bullet} = P_{L_n}^{\bullet \bullet}$, as well as $P_{n \to m}^{\bullet \bullet} = P_{L^{\scrC}_{n \to m}}^{\bullet \bullet}$.
\end{notation}

\subsection{Local states and anyon creation}

We start with some basic general observations that will be useful throughout the manuscript.

\subsubsection{Local states and eventually constant sequences of operators}

\begin{definition} \label[definition]{def:eventually constant}
    Let $\caK \subset \caH$ be a subspace of a Hilbert space $\caH$. 
    A sequence $(B_n) \subset B(\caH)$ is \emph{eventually constant on} $\caK$ if for every $\ket\Psi \in \caK$ there is $N\in \N$, such that for all $n\ge N$,
    \begin{equation}\label{eq:def evt constant locally}
        B_n \ket{\Psi} = B_N \ket{\Psi}.
    \end{equation}
\end{definition}

\begin{lemma}\label[lemma]{lem:evt constant product}
    If $(A_n) \subset B(\caK)$ and $(B_n) \subset B(\caH)$ are eventually constant sequences on $\caK\subset \caH$ then, for any $k \in \Z$, the sequence $(B_{n+k}A_n)$ is eventually constant on $\caK$.
    If moreover, $\overline \caK= \caH$ and the sequences $(A_n), (B_n)$ are uniformly bounded, 
    then the limits $A = \lim_n A_n$ and $B = \lim_n B_n$ exist in the strong operator topology, and
    \begin{equation}\label{eq:evt constant product}
        BA = \lim_n B_{n+k}A_n.
    \end{equation}
\end{lemma}

\begin{proof}
    We prove the claim for $k \geq 0$. The case $k < 0$ is treated analogously. Let $\ket\Psi \in \caK$, and take  $N$ satisfying \eqref{eq:def evt constant locally} for $(A_n)$.
    Since $A_{N}\ket{\Psi}\in \caK$, there is $M$ such that $B_{m} A_N \ket{\Psi} = B_M A_{N} \ket{\Psi}$ for all $m \ge M$.
    Then for all $n \geq n' \ge N$ such that $n'+k \geq M$, 
    \begin{equation}\label{eq:evt constant product sequence}
        B_{n+k} A_n \ket\Psi = 
        B_{n+k} A_{N} \ket\Psi = 
        B_{n'+k} A_{N} \ket\Psi = 
        B_{n'+k} A_{n'} \ket\Psi,
    \end{equation}
    proving the first claim.
   The limits are clearly defined pointwise on $\caK$. This defines a bounded operator on $\caK$ assuming a uniform bound for each sequence which extends to a bounded operator on $\caH$ if $\caK$ is dense. 
   Then Eq. \eqref{eq:evt constant product} follows by evaluating $BA$ on the dense subspace $\caK$ as in Eq. \eqref{eq:evt constant product sequence}.
\end{proof}

\begin{definition} \label[definition]{def:eventually constant on local states}
    A sequence $(A_n) \subset \caA^{\loc}$ is \emph{eventually constant on local states} if $(\pi^{\I}(A_n))$ is eventually constant on $\Hloc$.
\end{definition}
Recall that string operators are defined on $x\in \caA^{\loc}$ by the eventually constant sequence
$$\rho^X_{\scrC_n}(x) = (U_{\scrC_n}^X)^* x U_{\scrC_n}^X.$$
Trivially, this sequence is eventually constant on local states.

\subsubsection{Long strings create anyons}

Referring to the puncture at the initial edge $e_{n-1}$ of a link $L_{n-1}$, the unitary gate $u_{n-1}^X$ either does not do anything or creates an excitation of type $X$ at $e_{n-1}$ when acting on an arbitrary state that satisfies the vacuum conditions at $e_{n-1}$.
This is by design so that 
    \begin{equation}\label{eq:creating excitations}
	U_{n}^X 
        P_{n}^{\1\1} 
        = 
        P_{n}^{X?} U_n^X P_{n}^{\1\1}
        = p_n^X U_n^X P_n^{\1\1}.
    \end{equation}
This fact was used in the proof of \cite[Lemma 7.6]{bols2025levinI} as part of showing that string operators produce irreducible anyon representations.
The following lemma is a convenient packaging of the same fact applied to multiple string operators acting on local states.

Recall that we fixed a chain $\scrC = (L_n)$ and use the simplified Notation \ref{not:fixed chain}.

\begin{lemma}[Excitation] \label[excitation lemma]{lem:strings produce anyons}
    Let $\ket\Psi = \pi^{\I}(x) \ket\Omega$ for some $x \in \caA^{\loc}$. 
    Then there is $N$ such that for all $k \geq 1$ and $n \geq N+k$, and for all $X_0,\ldots, X_k\in \Irr\ZC$,
    \begin{equation} \label{eq:strings produce anyons}
        \pi^{\I}(U_{n}^{X_0} \cdots U_{n+k}^{X_k}) \ket \Psi = 
        \pi^{\I}(P_D^{(X_k,\ldots,X_0)} U_{n}^{X_0} \cdots U_{n+k}^{X_k}) \ket \Psi,
    \end{equation}
	where $D= \bigcup_{\kappa=1}^{k} C^{L_{n+\kappa}}$ with the boundary components enumerated so that $\caS_\kappa$ is the puncture at $\partial_\ii L_{n+k-\kappa+1}$ for $\kappa=1,\ldots, k+1$.
    Moreover, \eqref{eq:strings produce anyons} holds if $D$ is replaced by any finite region $D'$ 
    as long as $D'$ has the same inner boundary components as $D$ and is disjoint from the support of $\scrC\mn{1}{n-1}$ and the support of $x$.
\end{lemma}

\begin{remark} \label[remark]{rem:single string produces single anyon}
    The enumeration of boundary components above is chosen to match the enumeration induced by the standard anchor on the link $L_{n+1}$ in the case $k=1$ where $D=C^{L_{n+1}}$. 
    For example, if $X_0=\I$ the lemma states that $\pi^{\I}(U_{n+1}^{X_1}) \ket\Psi = \pi^{\I}(P_{L_{n+1}}^{X_1 \I } U_{n+1}^{X_1}) \ket\Psi $, noting that $P_{L_{n+1}}^{X_1 \I } = P_{L_{n+1}}^{(X_1, \I)}$ with the chosen enumeration.
\end{remark}

\begin{proof}
    For any $X \in \Irr Z(\caC)$, let us write $p_n^X = \frt_{\caS_{e_n}}(p^X)$ for all $n\in \N$, where $e_n = \partial_\ii L_n$.
    Let $n$ be arbitrary. 
    Using Eq. \eqref{eq:creating excitations}
    we find for any $\kappa < n$ that
    \begin{equation*}
        U_{n-\kappa}^X P_{n + \kappa \to n - \kappa}^{\I \I} 
        = P^{\I X}_{{n+\kappa} \to {n-\kappa+1}} U_{n-\kappa}^X P_{n + \kappa \to n - \kappa}^{\I \I}.
    \end{equation*}
    By iterations of \cref{lem:concatenation lemma} and the intertwining properties of unitary gates \eqref{eq:intertwining properties of u_L},
    \begin{align*}
        U_{n+\kappa}^X P_{n + \kappa \to n - \kappa}^{\I \I}  
        & = U^X_{{n+\kappa} \to {n-\kappa+1}}U_{n-\kappa}^X P_{n + \kappa \to n - \kappa}^{\I \I} \\
        &= U^X_{{n+\kappa} \to {n-\kappa+1}}  P^{\I X}_{{n+\kappa} \to {n-\kappa+1}} U_{n-\kappa}^X P_{n + \kappa \to n - \kappa}^{\I \I} \\
        &= u^X_{L_{{n+\kappa} \to {n-\kappa+1}}}  P^{\I X}_{{n+\kappa} \to {n-\kappa+1}} U_{n-\kappa}^X P_{n + \kappa \to n - \kappa}^{\I \I} \\
        &=  P^{X \I}_{{n+\kappa} \to {n-\kappa+1}}  u^X_{L_{{n+\kappa} \to {n-\kappa+1}}}   U_{n-\kappa}^X P_{n + \kappa \to n - \kappa}^{\I \I} \\
        & = P^{X\I }_{{n+\kappa} \to {n-\kappa+1}} U_{n+\kappa}^X P_{n + \kappa \to n - \kappa}^{\I \I}  \\
        & = P^{\I \I}_{{n+\kappa-1} \to {n-\kappa+1}} p_{n+\kappa}^{X } U_{n+\kappa}^X P_{n + \kappa \to n - \kappa}^{\I \I}. 
    \end{align*}
    It follows that for $n>k$,
    \begin{equation*}
        U_{n}^{X_0} \cdots U_{n+k}^{X_k} P^{\I \I}_{{n+k} \to {n-k}}
        = p_{n}^{X_0} \cdots p_{n+k}^{X_k} U_{n}^{X_0} \cdots U_{n+k}^{X_k} P^{\I \I}_{{n+k} \to {n-k}}.
    \end{equation*}

    Since $x$ has finite support, we may take $N$ such that $L_n \cap \supp(x) = \emptyset$ for all $n\ge N$.
    Then $\pi^{\I} \big( P_{n+k \to n-k}^{\I \I} \big) \ket\Psi = \ket\Psi$ for all $n\ge N+k$. 
    Now the statement follows if we take $D'$ as under the hypothesis for given $k\ge 1$ and $n\ge N+k$, 
    since $\pi^{\I}(P_{D'})\ket\Psi = \ket{\Psi}$ because $D'$ has disjoint support from $x$, and  since $P_{D'}$ commutes with $U_{n+\kappa}^{X_\kappa}$ for all $\kappa=0, \ldots, k$. Indeed, this is true of $B_f$ for every face in $D'$ by Lemma 6.5  of \cite{bols2025levinI}, and it is easily extended to cover the case of edges in $D'$ that do not belong to a face of $D'$.
\end{proof}

\subsection{Isomorphisms of fusion spaces} \label{subsec:isomorphism of fusion spaces}

We now define maps
\begin{equation*}
    \Phi_{XY}^Z : \ZC(X \otimes Y \to Z) \to \DHR(\bar \rho^X \otimes \bar \rho^Y \to \bar \rho^Z)
\end{equation*}
for all $X,Y,Z \in \Irr(\ZC)$ as follows. For $\alpha : X \otimes Y \to Z$, and
every $n \geq m$, we define for any link (with its standard anchor)
\begin{equation*}
	\Dr_{L}[\al] = \Dr_{L} \left[ \adjincludegraphics[valign=c, height = 0.8cm]{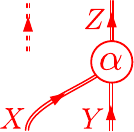} \right],
    \quad \quad
	\Dr_{L}[\al^\dagger] = \Dr_{L} \left[ \adjincludegraphics[valign=c, height = 0.8cm]{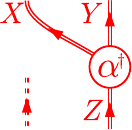} \right]
\end{equation*}
and we use the following shorthand notation:
\begin{equation*}
    \Dr_n[\alpha] = \Dr_{L_n}[\alpha], \quad \quad 
    \Dr_{\interval{m}{n}}[\alpha] = \Dr_{L_{\interval{m}{n}}^{\scrC}}[ \alpha].
\end{equation*}
Then define
\begin{equation}
	\Phi_n[\alpha] := (U_{n}^Z)^* \Dr_{n+1}\big[ \alpha \big] U_{n}^Y U_{n+1}^X.
\end{equation}

\begin{lemma}\label[lemma]{lem:Phi defined is intertwiner}
    Let $X,Y,Z \in \Irr Z(\caC)$ and $\alpha : X \otimes Y \to Z$.
    The sequence $( \Phi_n[\alpha] )_{n \in \N}$ is eventually constant on local states.
    The limit,
    \begin{equation*}
        \Phi_{XY}^Z[\alpha] := \lim_{n}\pi^{\1}(\Phi_n[\alpha]) \in \pi^{\1}(\caA_{\Lambda})'',
    \end{equation*}
    is an intertwiner  $\Phi_{XY}^Z[\alpha] : \bar\rho^X \times \bar\rho^Y \to \bar \rho^Z$.
    
    Also the sequence $(\Phi_n[\alpha]^*)_{n \in \N}$ is eventually constant on local states, and 
    \begin{equation*}
        \Phi_Z^{XY}[\alpha^\dagger] := \big(\Phi_{XY}^{Z}[\alpha] \big)^* = \lim_n \pi^\1(\Phi_n[\alpha]^*).
    \end{equation*}
\end{lemma}

\begin{proof}
    Let $\ket\Psi \in \Hloc$, and take $N$ as per \Cref{lem:strings produce anyons}, such that for all $n \ge N+1$,
    \begin{equation}
        \pi^{\I}(U_n^Y U_{n+1}^X) \ket\Psi 
        = \pi^{\I}(P^{(X,Y)}_{L_{n+1}} U_n^Y U_{n+1}^X) \ket\Psi.
    \end{equation}
    Now let $n \geq N+3$ and consider the region $D$ obtained from $L_{\interval{N+2}{n+1}}$ with additional punctures corresponding to punctures of $L_{N+2}$ and $L_{n+1}$, and enumerate the boundary components according to the punctures $(e_{n+1}, e_n, e_{N+2}, e_{N+1})$. See \Cref{fig:D region and anchor}.

    \begin{figure}
    \begin{center}
        \includegraphics[width=10cm]{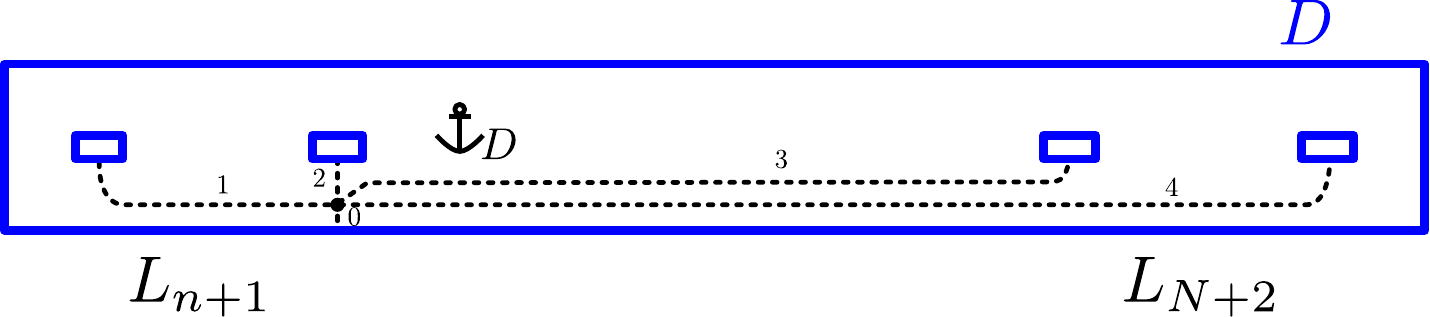}
    \end{center}
    \caption{Schematic representation of the region $D$ and its anchor $\anchor_D$. The enumeration of the punctures of $D$ induced by $\anchor_D$ corresponds to the ordering $(e_{n+1}, e_n, e_{N+2}, e_{N+1})$. The anchor $\anchor_D$ extends the standard anchor $\anchor_{L_{n+1}}$ and the standard anchor $\anchor_{L_{N+2}}$ with offset $2$.}
    \label{fig:D region and anchor}
    \end{figure}
    
    From \cref{lem:strings produce anyons} we obtain
    \begin{equation*}
        \pi^{\I}(U_{N+1}^Y U_{N+2}^X) \ket\Psi 
        = \pi^{\I}(P_D^{(\I,\I,X,Y)} U_{N+1}^Y U_{N+2}^X) \ket\Psi.
    \end{equation*}
    Take an anchor $\anchor_D$ for $D$ as in \Cref{fig:D region and anchor}. By repeated use of \cref{lem:concatenation lemma} and \cref{lem:inclusion lemma}, as well as Remark \ref{rem:adding vacuum lines}, we find 
    \begin{align*}
        (U^Z_{\interval{N+2}{n}})^* &\Dr_{n+1}\big[\alpha \big] U_{\interval{N+2}{n}}^Y U_{\interval{N+3}{n+1}}^X 
P_D^{(\I,\I,X,Y)} \\
         = & \Dr_{\interval{N+2}{n}}^{(Z \1 \to \1 Z)} \Dr_{n+1}\big[ \alpha \big] \Dr_{\interval{N+2}{n}}^{(\1 Y \to  Y \1 )} \Dr_{\interval{N+3}{n+1}}^{(\1 X\to X\1)}
        P_D^{(\I,\I,X,Y)} \\
        = & \Dr_{D}^{\anchorino_D}\left[ \adjincludegraphics[valign=c, width = 2.0cm]{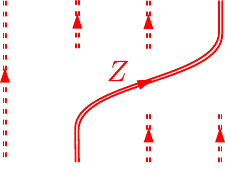} \right] \\
        \times & \Dr_{D}^{\anchorino_D}\left[ \adjincludegraphics[valign=c, width = 2.0cm]{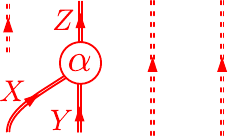} \right] \\
        \times & \Dr_{D}^{\anchorino_D}\left[ \adjincludegraphics[valign=c, width = 2.0cm]{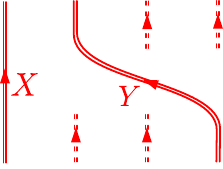} \right] \\
        \times & \Dr_{D}^{\anchorino_D}\left[ \adjincludegraphics[valign=c, width = 2.0cm]{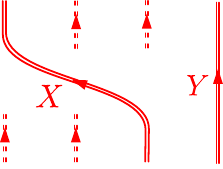} \right]  \, P_D^{(\I,\I,X,Y)} \\
        = & \Dr_{N+2}\big[\alpha \big] 
        P_D^{(\I,\I,X,Y)}.
    \end{align*}
    In combination with 
    \begin{align*}
        (U_n^Z)^* & = (U_{N+1}^Z)^* (U_{\interval{N+2}{n}}^Z)^*, \\
        U_n^Y U_{n+1}^X &= U^Y_{\interval{N+2}{n}} U^X_{\interval{N+3}{ n+1}} U_{N+1}^Y U_{N+2}^X,
    \end{align*}
    we obtain
    \begin{equation*}
        \pi^{\I}(\Phi_n[\alpha]) \ket\Psi
        = \pi^{\I}(\Phi_{N+1}[\alpha]) \ket\Psi
    \end{equation*}
    for all $n \geq N+3$. Since $\ket\Psi$ was arbitrary, we conclude that $\Phi_n[\alpha]$ is eventually constant on local states.

    We previously noted that $\rho^Y(x)$ is the limit of a sequence which is eventually constant on local states. Since $\rho^Y(x) \in \caA^{\loc}$, also $\rho^X(\rho^Y(x))$ is the limit of a sequence eventually constant on local states.
    So we compute for $n\ge N$ such that $\Dr_n[\alpha ]$ commutes with $x$,
    \begin{align*}
        \Phi_n[\alpha] \rho_{n+1}^X(\rho_{n}^Y(x)) 
        &= (U_n^Z)^* \Dr_{n+1}\big[ \alpha \big] U_{n}^Y U_{n+1}^X (U_{n+1}^X)^* (U_n^Y)^* x U_n^Y  U_{n+1}^X \\
        &= (U_n^Z)^* x \Dr_{n+1}\big[ \alpha \big] U_n^Y  U_{n+1}^X \\
        &= (U_n^Z)^* x (U_n^Z)^* U_n^Z \Dr_{n+1}\big[ \alpha \big] U_n^Y  U_{n+1}^X \\
        &= \rho_n^Z(x) \Phi_n[\alpha],
    \end{align*}
    Using \Cref{lem:evt constant product} we conclude
    $$
    \Phi_{XY}^Y[\alpha] \bar\rho^X(\bar\rho^Y(x)) = \bar\rho^Z(x) \Phi_{XY}^Y[\alpha].
    $$
    This extends to all of $\caA$.
    
    Finally, note that by \cref{lem:Drinfeld insertion multiplicativity}, $\Phi_n\big[\alpha\big]^* = (U_{n+1}^X)^*(U_n^Z)^* \Dr_{n+1}\big[\alpha^\dagger\big] U_n^Z$.
    The proof showing that $\Phi_n[\alpha]^*$ is eventually constant on local states runs completely analogously to the above, and is left to the reader. 
    It follows that $\big(\Phi_{XY}^{Z}[\alpha] \big)^* = \Phi_Z^{XY}[\alpha^\dagger]$.
\end{proof}

\begin{remark} \label[remark]{rem:fusion channels with punctures}
    Observe that $\pi^\1(\Phi_n[\alpha]) \to \Phi_{XY}^Z[\alpha]$ converges in the SOT-closed subalgebra $\pi^\1(\caA_\Lambda)''$ and that $\bar \rho^W$ is SOT-continuous on this subalgebra for any $W\in \Irr\ZC$. Therefore,
    \begin{equation}\label{eq:rho of Phi}
        \bar \rho^W(\Phi_{XY}^Z[\alpha]) = \lim_n \bar\rho^W(\pi^\1(F_n[\alpha])) = \lim_n \pi^\1\Big[(U_{n+1}^W)^* \Phi_n[\alpha] U_{n+1}^W\Big].
    \end{equation}
    Moreover, the sequence $\big( \, (U_{n+1}^W)^* \Phi_n[\alpha] U_{n+1}^W \, \big)_{n \in \N}$ is eventually constant on local states.
    We omit the proof which is identical to the above with the additional presence of $W$ excitations.
\end{remark}

The first observation about the maps $\Phi_{XY}^Z$ is that they provide an isomorphism of fusion rules (\Cref{prop:fusion rules in DHR}). This will follow from the following Lemma.

\begin{lemma}\label[lemma]{lem:Phi dagger lemma}
    For each $X,Y,Z \in \Irr \ZC$ and $\alpha, \delta : X \otimes Y \to Z$ we have
    \begin{equation*}
        \Phi_{XY}^{Z} \big[\alpha \big] \Phi^{XY}_{Z} \big[\delta^\dagger \big]
        = d_Z^{-1}\tr\big\{ \alpha \circ \delta^\dagger \big\} \times 1_{\caH},
    \end{equation*}
    and
    \begin{equation*}
        \Phi_Z^{XY} \big[\delta^\dagger \big] \Phi_{YX}^Z \big[\alpha \big] = \lim_n \pi^\1\Big( (U_{n+1}^X)^* (U_{n}^Y)^* \Dr_{n+1}\big[ \delta^\dagger \circ \alpha \big] U_n^Y U_{n+1}^X \Big).
    \end{equation*}
\end{lemma}

\begin{proof}
    Throughout this proof we identify $\caA$ with its image $\pi^{\I}(\caA)$ and drop $\pi^{\I}$ from the notation.
    We verify the identity on the dense subspace of local states.
    Let $\ket{\Psi} \in \Hloc$. From \cref{lem:strings produce anyons}, \Cref{rem:single string produces single anyon} we find 
    $U_n^Z \ket\Psi = 
    P_{n+1}^{\I Z} U_n^Z \ket\Psi$
    for all $n$ large enough. Using \Cref{lem:evt constant product} and \cref{lem:Drinfeld insertion multiplicativity} we compute, 
    \begin{align*}
        \Phi_{XY}^{Z} \big[\alpha \big] \Phi^{XY}_{Z} \big[\delta^\dagger \big] \ket{\Psi}
        & = (U_n^Z)^* \Dr_{n+1}\big[ \alpha \big] U_{n}^Y U_{n+1}^X (U_{n+1}^X)^* (U_n^Y)^* \Dr_{n+1}\big[\delta^\dagger \big] U_n^Z \ket\Psi \\
        &= (U_n^Z)^* \Dr_{n+1}\left[ \adjincludegraphics[valign=c, height = 1.0cm]{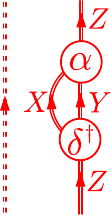} \right] U_n^Z \ket\Psi \\
        &= d_Z^{-1} \tr\big\{ \alpha \circ \delta^\dagger \big\} \times (U_n^Z)^* \Dr_{n+1}\big[ \id_{\I} \otimes \id_Z \big] P_{n+1}^{\I Z} U_n^Z \ket\Psi\\
        &= d_Z^{-1} \tr\big\{ \alpha \circ \delta^\dagger \big\} \times \ket\Psi,
    \end{align*}
    where we use that $\alpha \circ \delta^\dagger = d_Z^{-1} \tr\big\{ \alpha \circ \delta^\dagger \big\} \times \id_Z$ and the face that $\Dr_{n+1}[\id_{\I} \otimes \id_Z] P_{n+1}^{\I Z} = P_{n+1}^{\I Z}$.

    Similarly, by \Cref{lem:strings produce anyons} we have $U_n^Y U_{n+1}^X \ket{\Psi} = P_{n+1}^{(X,Y)} U_n^Y U_{n+1}^X \ket{\Psi}$ for all $n$ large enough.
    Arguing by \Cref{lem:evt constant product} and \Cref{lem:Drinfeld insertion multiplicativity} we find that
    \begin{align*}
        \Phi_Z^{XY} \big[\delta^\dagger \big] \Phi_{XY}^Z \big[\alpha \big] \ket{\Psi} 
        &= (U_{n+1}^X)^* (U_{n}^Y)^* \Dr_{n+1}\big[ \delta^\dagger \big] U_n^Z (U_n^Z)^* \Dr_{n+1}\big[\alpha \big] U_n^Y U_{n+1}^X \ket{\Psi} \\
        &= (U_{n+1}^X)^* (U_{n}^Y)^* \Dr_{n+1}\big[ \delta^\dagger \circ \alpha \big] U_n^Y U_{n+1}^X \ket{\Psi},
    \end{align*}
    finishing the proof.
\end{proof}

\begin{proposition}\label[proposition]{prop:fusion rules in DHR}
    Let $X, Y \in \Irr Z(\caC)$ and let $$\{ \pi_{Z, \kappa} : X \otimes Y \to Z \}_{Z \in \Irr Z(\caC), \kappa = 1, \cdots, N_{XY}^Z}$$ provide an orthogonal direct sum decomposition of $X \otimes Y$. 
    Then $$\{ \Phi_{X Y}^Z \big[ \pi_{Z, \kappa} \big] : \bar \rho^X \times \bar \rho^Y \to \bar \rho^Z \}_{Z \in \Irr Z(\caC), \kappa = 1, \cdots, N_{XY}^Z}$$ provides an orthogonal direct sum decomposition of $\bar \rho^X \times \bar \rho^Y$. 
    In particular, $\SSS_f$ is closed under the tensor product and $\{\Phi_{XY}^Z\}$ provides and isomorphism of fusion rules as defined in Section  \ref{subsec:isomorphism of F- and R-symbols}.
\end{proposition}

\begin{proof}
    It is immediate from \Cref{lem:Phi dagger lemma}, that 
    \begin{equation*}
        \sum_{Z,\kappa} \Phi^{Z}_{XY} \big[\pi_{Z,\kappa}\big]^* \Phi_{XY}^Z\big[\pi_{Z,\kappa}\big] =
        \sum_{Z,\kappa} \Phi_{Z}^{XY}\big[\pi_{Z,\kappa}^\dagger\big] \Phi_{XY}^Z\big[\pi_{Z,\kappa}\big] =
        \id_{\bar\rho^X\times \bar\rho^Y},
    \end{equation*}
    and
    \begin{equation*}
        \Phi^{XY}_Z\big[\pi_{Z,\kappa}\big]   \Phi_{Z'}^{XY}\big[\pi_{Z',\lambda}\big]^* = 
        \Phi_{XY}^Z\big[\pi_{Z,\kappa}\big]   \Phi_{Z'}^{XY}\big[\pi_{Z',\lambda}^\dagger \big] = 
        \delta_{Z,Z'}\delta_{i,i'} \times \id_{\bar\rho^Z}
    \end{equation*}
    for all $Z,Z' \in \Irr\ZC.$
\end{proof}

\subsection{Isomorphism of \texorpdfstring{$F$}{F}-symbols} \label{subsec:monoidality}

$F$-symbols for $\ZC$ are unitary maps defined by the commuting diagram
\begin{equation*}
    \begin{tikzcd}
        \ZC(X \otimes (Y \otimes Z) \to W) \, \arrow[rr, "\simeq"] \arrow[d, "- \circ \alpha_{X,Y,X}"'] & & 
        \displaystyle\bigoplus_{V\in \Irr(\ZC)} \ZC(X\otimes V \to W) \otimes \ZC(Y \otimes Z \to V) \arrow[d, "F_{XYZ}^{W}"] \\
        \ZC((X \otimes Y) \otimes Z \to W) \, \arrow[rr, "\simeq"]  & & 
        \displaystyle\bigoplus_{U\in \Irr(\ZC)} \ZC(U\otimes Z \to W) \otimes \ZC(X \otimes Y \to U)  
    \end{tikzcd}
\end{equation*}
where the left vertical arrow is given by the precomposition with the associator $\alpha$, and horizontal isomorphisms from right to left are given component-wise by composition: $\xi \otimes \eta \mapsto \xi \circ (\id_X \otimes \eta)$.
$F$-symbols are defined similarly for $\DHR_f$, where the associator is the identity. 

The following Lemma shows that $\{\Phi_{XY}^Z\}$ intertwines the $F$-symbols of $Z(\caC)$ and $\DHR_f$ in the sense described in Appendix \ref{app:equivalence}.

\begin{lemma}\label[lemma]{lem:Phi is monoidal}
    Let $D_n = C^{L_n} \cup C^{L_{n+1}}$, and $U,V,W,X,Y,Z \in \Irr \ZC$.
    For all $\xi: U \otimes Z \to W$ and $\eta : X \otimes Y \to U$,
    and all $\gamma: X \otimes V \to W$ and $\delta : Y \otimes Z \to V$,
    \begin{gather}
	    \Phi^W_{UZ}[\xi] \circ (\Phi^U_{XY}[\eta] \times \id_{\bar \rho^Z}) = \lim_n \pi^\1\left( (U^W_n)^* \Dr_{D_{n+1}} \left[ \adjincludegraphics[valign=c, height=1.0cm]{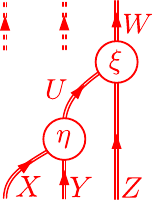} \right]  U_{n}^Z U_{n+1}^Y U_{n+2}^X \right), \label{eq:first fusion tree correspondence} \\
        \Phi^W_{XV}[\gamma] \circ (\id_{\bar \rho^X} \times \Phi^V_{YZ}[\delta])
										= \lim_n \pi^\1  \left( (U_{n}^W)^* \Dr_{D_{n+1}}  \left[ \adjincludegraphics[valign=c, height=1.0cm]{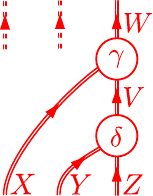} \right] U_{n}^Z U_{n+1}^Y U_{n+2}^X \right), \label{eq:second fusion tree correspondence}
    \end{gather}
    where the limits are in the SOT-topology.
\end{lemma}

\begin{proof}
    We verify both identities on the dense subspace of local states and omit $\pi^{\I}$ from notation throughout this proof. Let $\ket\Psi \in \pi^\1(\caA^{\loc})\ket{\Omega}$, and take $n$ large enough according to \Cref{lem:strings produce anyons} such that 
    $$ 
        U_{n}^Z U_{n+1}^Y U_{n+2}^X \ket{\Psi} =
        P_{D_{n+1}}^{(X,Y,Z)} U_{n}^Z U_{n+1}^Y U_{n+2}^X \ket{\Psi}.
    $$
    Using this together with $\Phi^U_{XY}[\eta] \times \id_{\bar \rho^Z} = \Phi^U_{XY}[\eta]$ and \Cref{lem:evt constant product} we find that the left hand side of \eqref{eq:first fusion tree correspondence} becomes
    \begin{align*}
        \Phi^W_{UZ}[\xi] \Phi^U_{XY}[\eta]  \ket{\Psi} 
        &= (U_n^W)^* \Dr_{n+1}\big[ \xi \big] U_{n}^Z U_{n+1}^U (U_{n+1}^U)^* \Dr_{n+2}\big[ \eta \big] U_{n+1}^Y U_{n+2}^X \ket{\Psi} \\
        &= (U_n^W)^* \Dr_{n+1}\big[ \xi \big] \Dr_{n+2}\big[ \eta \big] P_{D_{n+1}}^{(X,Y,Z)} U_{n}^Z U_{n+1}^Y U_{n+2}^X \ket{\Psi} \\
        &= (U_n^W)^* \Dr_{D_{n+1}}\big[ \xi \circ (\eta \otimes \id_Z) \big] U_{n}^Z U_{n+1}^Y U_{n+2}^X \ket{\Psi},
    \end{align*}
    where we used \Cref{lem:Drinfeld insertion multiplicativity} in combination with \cref{lem:inclusion lemma} applied to $L_n, L_{n+1} \subset D_{n+1}$ with an anchor $\anchor_{D_{n+1}}$ extending $\anchor_{L_n}$ and $\anchor_{L_{n+1}}$. This shows the first identity.
    
    For the second identity we have $\id_{\bar \rho^X} \times \Phi^V_{YZ}[\delta] = \bar \rho^X (\Phi^V_{YZ}[\delta])$,
    so in this case we find
    \begin{align*}
        & \Big( \Phi^W_{XV}[\gamma] \circ (\id_{\bar \rho^X} \times \Phi^V_{YZ}[\delta]) \Big) \ket{\Psi} \\
        &= (U_{n+1}^W)^* \Dr_{n+2}\big[\gamma  \big] U_{n+1}^V U_{n+2}^X (U_{n+2}^X)^* (U_n^V)^* \Dr_{n+1}\big[ \delta \big] U_n^Z U_{n+1}^Y U_{n+2}^X \ket{\Psi} \\
        &= (U_{n}^W)^*  \big( u_{n+1}^W \Dr_{n+2}\big[\gamma  \big] u_{n+1}^V \Dr_{n+1}\big[ \delta \big] P_{D_{n+1}}^{(X,Y,Z)} \big) U_n^Z U_{n+1}^Y U_{n+2}^X \ket{\Psi} \\
        &= (U_{n}^W)^* \Dr_{D_{n+1}}\big[\gamma \circ (\id_X \otimes \delta)   \big] U_n^Z U_{n+1}^Y U_{n+2}^X \ket{\Psi}.
    \end{align*}
    Here, the second equality is a simple algebraic reduction using definitions and disjointness of support. The last equality follows from a repeated application of \cref{lem:inclusion lemma} and \cref{lem:Drinfeld insertion multiplicativity}, which shows: 
    \begin{align*}
        & u_{n+1}^W \Dr_{n+2}\big[\gamma  \big] u_{n+1}^V \Dr_{n+1}\big[ \delta \big] P_{D_{n+1}}^{(X,Y,Z)} \\
        = & \Dr_{D_{n+1}}^{\anchorino_{D_{n+1}}} \left[ \adjincludegraphics[valign=c, width = 1.6cm]{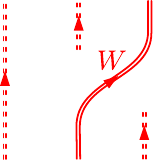} \right] \\
        \times & \Dr_{D_{n+1}}^{\anchorino_{D_{n+1}}} \left[ \adjincludegraphics[valign=c, width = 1.6cm]{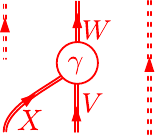} \right] \\
        \times & \Dr_{D_{n+1}}^{\anchorino_{D_{n+1}}} \left[ \adjincludegraphics[valign=c, width = 1.6cm]{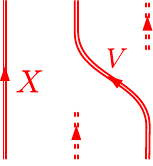} \right] \\
        \times & \Dr_{D_{n+1}}^{\anchorino_{D_{n+1}}} \left[ \adjincludegraphics[valign=c, width = 1.8cm]{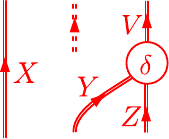} \right]  \, P_{D_{n+1}}^{(X,Y,Z)} \\
        = & \Dr_{D_{n+1}}\big[\gamma \circ (\id_X \otimes \delta)   \big] 
        P_{D_{n+1}}^{(X,Y,Z)}.
    \end{align*}
\end{proof}

\section{Braiding} \label{sec:braiding}

In this section, we show that the identification $\{ \Phi_{X Y}^Z \}$ of fusion spaces also intertwines $R$-symbols. This will imply that $Z(\caC)$ and $\DHR_f$ are equivalent as braided $\rm C^*$-tensor categories by Proposition \ref{prop:isomophic symbols implies isomorphic categories}.

Note that the notation \ref{not:fixed chain} is no longer in force because we will be dealing with more than one chain.

\subsection{Transporters} \label{subsec:transporters}

In order to compute the braiding of $\DHR_f$ from Eq. \eqref{eq:braiding defined}, we construct explicit unitary transporters between $\bar \rho^X_\scrC$ and $\bar \rho^X_{\scrC'}$, for arbitrary $X\in \Irr\ZC$ and  sufficiently nice half-infinite chains $\scrC$ and $\scrC'$. Although it would suffice to do so for the fixed fiducial chain and a single auxilliary chain, we describe the general construction for completeness, cf. \Cref{rem:no haag}.

\subsubsection{Isotopy Lemma} \label{ssubsec:isotopy lemma}

The central ingredient for the transportation of string operators is a form of isotopy invariance enjoyed by the unitary gates of the string operators.

\begin{definition} \label[definition]{def:simple isotopy}
	There is a simple isotopy between links $L$ and $L'$ if $\partial_{\ii} L = \partial_{\ii} L'$ and $\partial_{\f} L = \partial_{\f} L'$, and there is a region $E$, such that
	\begin{enumerate}
		\item $C^L, C^{L'} \subset E$,
		\item There is a region $C \subset E$ such that $E$ is obtained by gluing $L$ to $C$ along a self-avoiding dual path (see Section \ref{subsec:inclusion lemma}),
		\item there is a region $C' \subset E$ such that $E$ is obtained by gluing $L'$ to $C'$ along a self-avoiding dual path,
        \label{item:stupid face assumption on bridges}
		\item There is an anchor $\anchor_E$ for $E$ which extends both $\anchor_L$ and $\anchor_{L'}$ with offset zero.
	\end{enumerate}
	The region $E$ is said to support a simple isotopy between $L$ and $L'$.
\end{definition}

The definition is designed so that \cref{lem:inclusion lemma} can be applied. Note that it follows from the assumptions that $\Sigma_E$ is a twice punctured disk whose punctures coincide with $\partial_{\ii} L$ and $\partial_{\f} L$. The subregions $C$ and $C'$ glued to $L$ and $L'$ respectively to obtain $E$ have associated regions $\Sigma_C$ and $\Sigma_{C'}$ which are disks. See Figure \ref{fig:simple isotopy of links}.

\begin{figure}[h]
    \begin{center}
        \includegraphics[width=8cm]{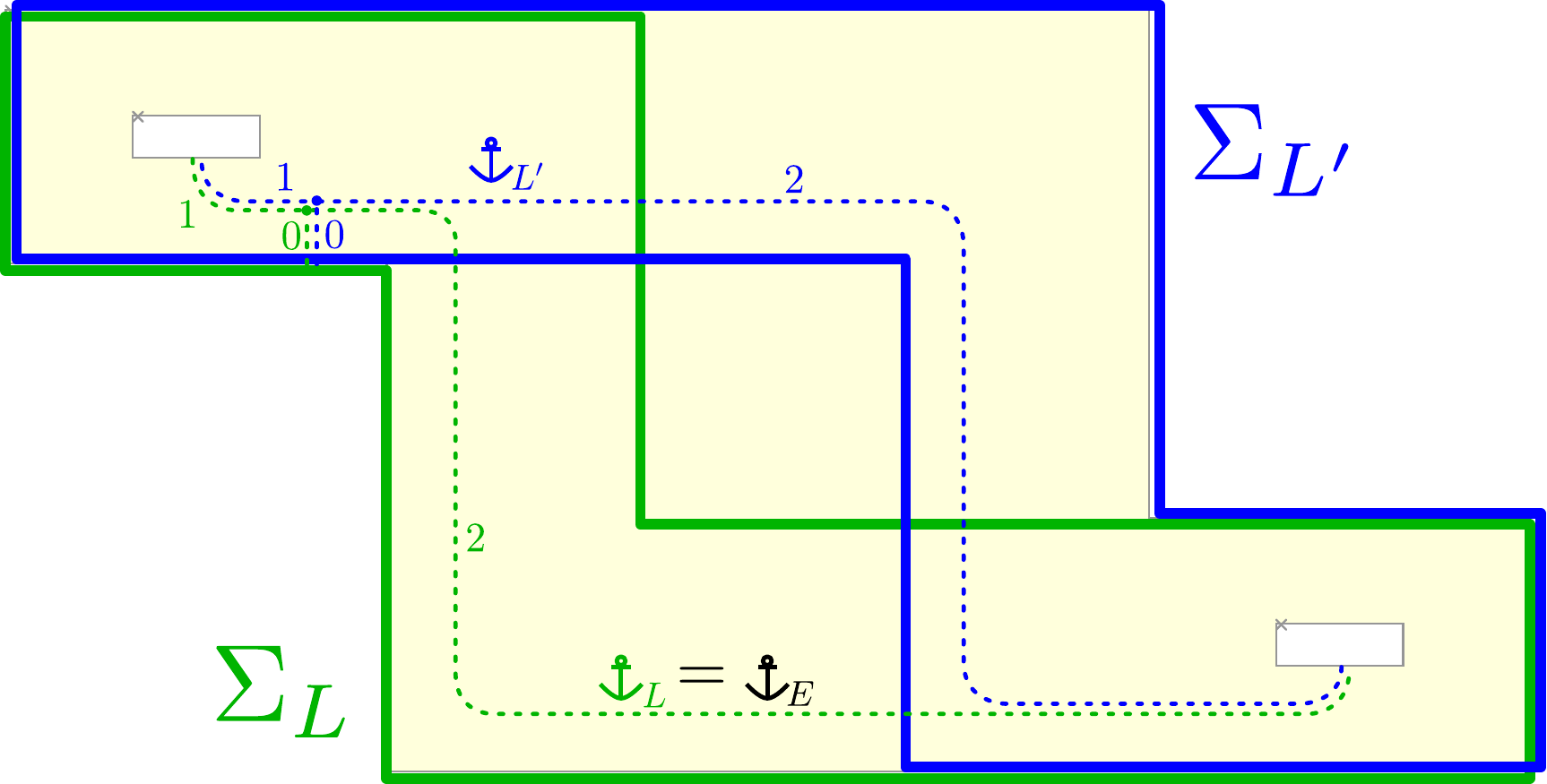}
    \end{center}
    \caption{Isotopy of links $L$ and $L'$ supported by region $E$. The regions $\Sigma_L, \Sigma_{L'}$ and $\Sigma_E$ are shown in green, blue, and light yellow respectively. A choice of anchors $\anchor_E = \anchor_{L}$ and $\anchor_{L'}$ is also shown, together with the enumeration of their strands.}
    \label{fig:simple isotopy of links}
\end{figure}

Recall that for any region $E$, the projector $P_E \in \caA_E$ is the orthogonal projector onto the skein subspace $H_E \subset \caH_E$ associated to $E$. See Section \ref{ssubsec:skein subspaces defined}.
\begin{lemma}[Isotopy] \label[isotopy lemma]{lem:isotopy lemma}
	If there is a simple isotopy between links $L$ and $L'$ supported by a region $E$, then
	$$ u_L^X P_E = u_{L'}^X P_E $$
	for any $X \in \Irr Z(\caC)$. 
\end{lemma}

\begin{proof}
	Recall the definition Eq. \eqref{eq:unitary gate defined} of the unitary $u_{L}^X$. The Drinfeld insertions appearing in that definition are partial isometries with initial and final projections given by Eqs. \eqref{eq:P_L defs1}-\eqref{eq:P_L defs4}. It is clear that $P_L^{\bullet \bullet} P_E = P_{L'}^{\bullet \bullet}P_E$ for all these projectors, so it is sufficient to show the following four equalities:
	\begin{align*}
		\Dr_L^{(\I \I \rightarrow X \bar X)} P_E &= \Dr_{L'}^{(\I \I \rightarrow X \bar X)} P_E, \quad \Dr_L^{(X \bar X \rightarrow \I \I)} P_E = \Dr_{L'}^{(X \bar X \rightarrow \I \I)} P_E, \\
		\Dr_L^{(X \I \rightarrow \I X)} P_E &= \Dr_{L'}^{(X \I \rightarrow \I X)} P_E, \quad \Dr_L^{(\I X \rightarrow X \I)} P_E = \Dr_{L'}^{(\I X \rightarrow X \I)}.
	\end{align*}

	Let us show the first of these equalities. The proof of the others is completely analogous. We use the enumeration of boundary components induced by the anchors $\anchor_L, \anchor_{L'}$ and $\anchor_E$. Since $E$ is obtained by gluing a disk $C$ to the link $L$ along a self-avoiding dual path, and $\anchor_E$ extends $\anchor_L$ with no offset, the inclusion Lemma \ref{lem:inclusion lemma} yields
	\begin{align*} \Dr_L^{(\I\I \rightarrow X \bar X)} P_E &= d_X^{-1/2} \, \Dr_{L} \left[ \,\,\,\,\, \adjincludegraphics[valign=c, width=1.5cm]{pair_creation_morphism.pdf} \right] P_E^{(\I \I)} \\  &=  d_X^{-1/2} \, \Dr_E^{\anchorino_E} \left[ \,\,\,\,\, \adjincludegraphics[valign=c, width=1.5cm]{pair_creation_morphism.pdf} \right] P_E^{(\I \I)} =  d_X^{-1/2} \,  \Dr_E^{\anchorino_E} \left[ \,\,\,\,\, \adjincludegraphics[valign=c, width=1.5cm]{pair_creation_morphism.pdf} \right] P_E
	\end{align*}
	where in the first and last equalities we noted that the Drinfeld insertions are supported on $P_L^{(\I \I)}$ and $P_E^{(\I \I)}$, and that $P_L^{\I \I} P_E = P_E^{(\I \I)}$. For the same reason we have
	$$ \Dr_L^{(\I\I \rightarrow X \bar X)} P_E =  d_X^{-1/2} \,  \Dr_E^{\anchorino_E} \left[ \,\,\,\,\, \adjincludegraphics[valign=c, width=1.5cm]{pair_creation_morphism.pdf} \right] P_E. $$
	Combining these result yields the first equality above.
\end{proof}

In order to apply this isotopy below, we will also need the following variation of the concatenation lemma:
\begin{lemma} \label[lemma]{lem:variation on concatenation}
	Let $X \in \Irr Z(\caC)$ and let $L_1$ and $L_2$ be composable links with $L = L_2 \wedge L_1$. Then
	$$  u_{L_1}^X u_{L_2}^X P_L^{X \I} = u_{L}^X P_L^{X \I}.$$
\end{lemma}

\begin{proof}
	As in the proof of the concatenation Lemma 6.9 in \cite{bols2025levinI} we have
	$$ u_{L_1}^X u_{L_2}^X P_L^{X \I} = \Dr_{L_1}^{(X \I \rightarrow \I X)} \Dr_{L_2}^{(X \I \rightarrow \I X)} P_L^{X \I}  $$
	and
	$$ u_L^{X} P_L^{X \I} = \Dr_{L}^{(X \I \rightarrow \I X)} P_L^{X \I}. $$
	Consider the region $D = C^{L_1} \cup C^{L_2}$, and let $\anchor_D$ be as in the proof of Lemma 6.9 of \cite{bols2025levinI}, so that $\anchor_D$ extends the anchor $\anchor_{L_1}$ with offset 0, and $\anchor_D$ extends the anchor $\anchor_{L_1}$ with offset 1.
	As we get $L$ from filling in the middle puncture of $D$, we have
        $$B_L = B_D \frt_{\caS^{\anchorino_D}_2}(p^\1) = \frt_{\caS^{\anchorino_D}_2}(p^\1) B_D.$$
    Making use of this observation as well as \cref{lem:inclusion lemma} and \cref{lem:Drinfeld insertion multiplicativity} we find
    \begin{align*}
        \Dr_{L_1}^{(X \I \rightarrow \I X)} \Dr_{L_2}^{(X \I \rightarrow \I X)} P_L^{X \I} = &\Dr_D \left[ \adjincludegraphics[valign=c, width = 1.5cm]{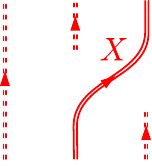} \right] \\
        \times &\Dr_D \left[ \adjincludegraphics[valign=c, width = 1.5cm]{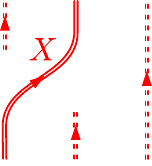} \right] \, P_L^{X \I} \\
        = &\Dr_L^{(X \I \rightarrow \I X)}  \, P_L^{X \I},
    \end{align*}
    where the last equality is evident as in Remark \ref{rem:adding vacuum lines}. 
\end{proof}

\subsubsection{Construction of transporters}

We now construct explicit transporters between string operators $\bar \rho_{\caC}^X$ and $\bar \rho_{\scrC'}^X$ supported on chains $\scrC, \scrC'$ which admit a so-called \emph{bridge}, see Figure \ref{fig:bridge}

\begin{definition} \label[definition]{def:bridge}
	Let $\scrC = \{ L_n \}_{n \in \N}$ and $\scrC' = \{ L'_n \}_{n \in \N}$ be chains. A bridge from $\scrC$ to $\scrC'$ is a pair $( \{ Q_n \}_{n \geq N}, \{ E_n \}_{n \geq N} \} )$ for some $N \in \N$, where each $Q_n$ is a link and each $E_n$ is a region such that for all $n \geq N$,
	\begin{enumerate}
		\item  $L_n$ and $Q_{n+1}$ are composable, i.e. $Q_{n+1} \wedge L_n$ is a link,
		\item $Q_{n}$ and $L'_n$ are composable, i.e. $L'_n \wedge Q_n$ is a link,
		\item The region $E_n$ supports simple isotopy between $Q_{n+1} \wedge L_n$ and $L'_n \wedge Q_{n}$,
		\item For any finite region $C$ there is $n_0 \geq N$ so that $E_n \cap C = \emptyset$ if $n \geq n_0$.
        \label{item:bridge eventually disjoint}
        \item The region $E_n$ is disjoint from the supports of $\scrC\mn{1}{n-2}$ and $\scrC'\mn{1}{n-2}$.
        \label{item:bridge untangled}
	\end{enumerate}	
    We say the bridge $( \{Q_n\}_{n \geq N}, \{ E_n \}_{n \geq N})$ is supported in a cone $\Lambda$ if the supports of all links $\{L_n, L'_{n}\}_{n \in \N}$ and $\{Q_n\}_{n \geq N}$ are subsets of $\Lambda$.
\end{definition}

\begin{figure}
    \begin{center}
        \includegraphics[width=10cm]{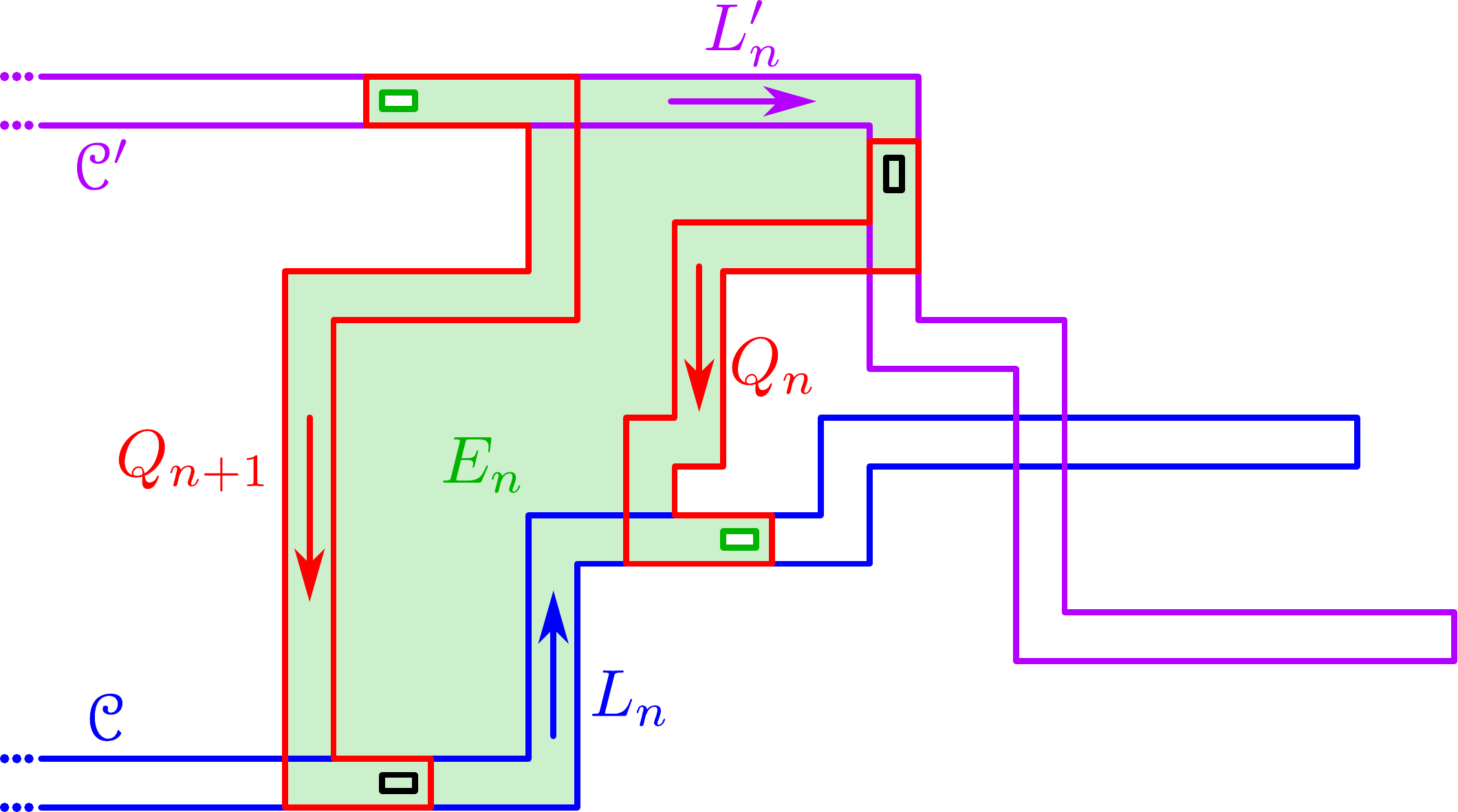}
    \end{center}
    \caption{Part of a bridge $( \{ Q_n \}, \{ E_n \} )$ between chains $\scrC = (L_n)$ and $\scrC' = (L'_n)$.}
    \label{fig:bridge}
\end{figure}

\begin{proposition} \label[proposition]{prop:explicit transporters}
	Let $\scrC$ and $\scrC'$ be chains and suppose $( \{ Q_n \}_{n \geq N}, \{ E_n \}_{n \geq N} )$ is a bridge from $\scrC$ to $\scrC'$. Let $X \in \Irr Z(\caC)$. For each $n \geq N$, define the unitary
	\begin{equation} \label{eq:finite transporter}
		T_n = (U_{\scrC'_n}^X)^* u_{Q_{n+1}}^X U_{\scrC_n}^X.
	\end{equation}
	Then the sequence $\{ T_n \}_{n \geq N}$is eventually constant on local states and the limit
	\begin{equation} \label{eq:limiting intertwiner}
		T := \lim_{n \uparrow \infty} \pi^{\I}(T_n)
	\end{equation}
	is a unitary intertwiner from $\bar \rho_{\scrC}^X$ to $\bar \rho_{\scrC'}^X$.

	Moreover, if the bridge $( \{ Q_n \}_{n \geq N}, \{ E_n \}_{n \geq N} )$ is supported in a cone $\Lambda$, then $T \in \caR(\Lambda)$. 
\end{proposition}

\begin{proof}
	Throughout this proof we identify $\caA$ with its image $\pi^{\I}(\caA)$ and drop $\pi^{\I}$ from the notation. Since the object $X \in \Irr Z(\caC)$ is fixed, we also drop it from the notation where no ambiguity arises.

	Consider a local state $\ket{\Psi} := x \ket{\Omega}$ for some $x \in \caA^{\loc}$. 

    Let $n \ge \max\{N+2,n_0\}$ where $N$ is given by \cref{lem:strings produce anyons} depending on $x$, and $n_0$ granted by \cref{item:bridge eventually disjoint} of \Cref{def:bridge} for the support of $x$. Taking into account \cref{item:bridge untangled} of Definition \ref{def:bridge},
    we see that the hypothesis of \cref{lem:strings produce anyons} is satisfied for $E_n$ (applying the lemma with values $n-1$ and $k=1$ with $X_0=X$ and $X_1=\I$, where the region $D=L_{(n-1)+1}\subset E_n$).
    This gives    
    \begin{equation} \label{eq:first property}
        U_{\scrC_{n-1}} \ket{\Psi} = P_{E_n}^{(\I, X)} U_{\scrC_{n-1}} \ket{\Psi}.
    \end{equation}
	Then we compute,
	\begin{align*}
		T_n \ket{\Psi} &= U_{\scrC'_n}^* \, u_{Q_{n+1}} \, U_{\scrC_n} \, \ket{\Psi} 
        = U_{\scrC'_n}^* \, u_{Q_{n+1}} \, u_{L_n} \, U_{\scrC_{n-1}} \, \ket{\Psi} \\
		&\stackrel{\eqref{eq:first property}}{=}  \, U_{\scrC'_n}^* \, u_{Q_{n+1}} \, u_{L_n} \, P_{E_n}^{(\I, X)} \, U_{\scrC_{n-1}} \, \ket{\Psi} \, 
        \stackrel{Lem. \ref{lem:variation on concatenation}}{=} \, U_{\scrC'_n}^* \, u_{Q_{n+1} \wedge L_n} \, P_{E_n}^{(\I, X)} \, U_{\scrC_{n-1}} \, \ket{\Psi} \\
        & \stackrel{Lem. \ref{lem:isotopy lemma}}{=} \, U_{\scrC'_n}^* \, u_{L'_n \wedge Q_n} \, P_{E_n}^{(\I, X)} \, U_{\scrC_{n-1}} \, \ket{\Psi} \, 
        \stackrel{Lem. \ref{lem:variation on concatenation}}{=} \, U_{\scrC'_n}^* \, u_{L'_n} \, u_{Q_n} \, P_{E_n}^{(\I, X)} \, U_{\scrC_{n-1}} \, \ket{\Psi}  \\
        &= U_{\scrC'_{n-1}}^* \, u_{Q_n} \, U_{\scrC_{n-1}} \, \ket{\Psi} = T_{n-1} \ket{\Psi}.
	\end{align*}
	Since this holds for all $n \geq \max\{N+2, n_0\}$, we conclude that $(T_n)_{n \geq N}$ is eventually constant on local states.

	Let us now show that the limit $T = \lim_{n} T_n$ is an intertwiner from $\bar \rho_{\scrC}^X$ to $\bar \rho^X_{\scrC'}$. Fix $y \in \caA^{\loc}$ and note that for all $n \in \N$ large enough so that $\supp Q_{n+1} \cap \supp(y) = \emptyset$ we have
	\begin{equation*}
		T_n \, \big( U_{\scrC_n}^* y U_{\scrC_n} \big) = U_{\scrC'_n} u_{Q_{n+1}} \, y U_{\scrC_n} = U_{\scrC'_n} y \, u_{Q_{n+1}} U_{\scrC_n} = \big( U_{\scrC'_n} y U_{\scrC'_n}^* \big) \, T_n.
	\end{equation*}
	Since $( T_n )$, $( U_{\scrC_n}^* y U_{\scrC_n} )$, and $( U_{\scrC'_n}^* y U_{\scrC'_n} )$ are sequences which are eventually constant on local states, converging to $T, \bar \rho^X_{\scrC}(y)$, and $\bar \rho^X_{\scrC'}(y)$ respectively, we conclude from \Cref{lem:evt constant product} that
	$$ T \, \bar \rho^X_{\scrC}(y) \ket{\Psi} = \bar \rho^X_{\scrC'}(y) \, T \ket{\Psi} $$
	for all $y \in \caA^{\loc}$ and any local state $\ket{\Psi} \in \caA^{\loc} \ket{\Omega}$. We then conclude that $T \in \DHR_f( \bar \rho^X_{\scrC} \rightarrow \bar \rho^X_{\scrC'} )$ by density.

	To see that $T$ is unitary, we first note that by changing the orientation of the links $\{Q_n\}$ we obtain a bridge from $\scrC'$ to $\scrC$. Then it follows in the same way as above that the sequence $( T_n^* )_{n \geq N}$ is eventually constant on local states and therefore converges strongly to $T^*$. Using unitarity of $T_n$ we find
	$$ T^* T \ket{\Psi} = \lim_n \, T_n^* T_n \ket\Psi = \ket\Psi,$$
	and similarly $T T^* \ket{\Psi} = \ket{\Psi}$ for any local state $\ket{\Psi} \in \caA^{\loc} \ket{\Omega}$ by \Cref{lem:evt constant product}. By density we conclude that $T$ is unitary.

	Finally, if $\Lambda \subset \R^2$ contains the supports of all links $\{L_n, L'_{n}\}_{n \in \N}$ and all links $\{Q_n\}_{n \geq N}$, then each $T_n$ is supported in $\Lambda$ and it follows that $T \in \caR(\Lambda)$.
\end{proof}

\begin{remark} \label[remark]{rem:transporter with punctures}
    The endomorphisms $\bar \rho_{\scrC}^Y$ are SOT-continuous on $\pi^{\1}(\caA_\Lambda)$ for all chains $\scrC$, cones $\Lambda$, and $Y\in \Irr\ZC$, so we observe that in the case of the above lemma where the SOT-convergence $T_n\to T$ takes place in $\pi^\1(\caA_\Lambda)''$ for an allowed cone $\Lambda$, then  
    \begin{equation}\label{eq:rho of T}
        \bar \rho_{\scrC}^Y(T) = \lim_n \bar\rho_{\scrC}^Y(\pi^\1(T_n)) = \lim_n \pi^\1\Big[(U_{\scrC_{n+1}}^Y)^* T_n U_{\scrC_{n+1}}^Y\Big].
    \end{equation}
    The sequence $\big( \, U_{\scrC_{n+1}}^Y)^* T_n U_{\scrC_{n+1}}^Y \, \big)_{n \geq N}$ is eventually constant on local states.
    A similar statement holds for $T^*$.
    We omit the proof, with the remark that it is identical to the above with the additional presence of $Y$ excitations.
\end{remark}

\subsection{Isomorphism of \texorpdfstring{$R$}{R}-symbols} \label{subsec:braiding}

For this subsection we fix a \emph{fiducial chain} $\scrC = (L_n)$ as follows. Identify the set of faces $\caF$ of $\Z^2$ with the dual lattice $(\Z^2)^*$ so that the face $f$ whose bottom left vertex is $(x, y) \in \Z^2$ is given the coordinate $(x, y) \in (\Z^2)^*$. For each $n \in \N$ we take $L_n = ( (-6n-2+i, 0) )_{i = 0}^7$, directed to the right. See figure \ref{fig:fiducial and right chain}.
With the fiducial chain $\scrC$ fixed, we write $\bar \rho^X := \bar \rho^X_{\scrC}$ for any $X \in \Irr Z(\caC)$, and let $\Phi_{XY}^Z : Z(\caC)(X \otimes Y \to \Z) \rightarrow \DHR_f( \bar \rho^X \times \bar \rho^Y \to \bar \rho^Z )$ be the maps constructed in Section \ref{subsec:isomorphism of fusion spaces} for the fiducial chain $\scrC$.

\begin{figure}
    \begin{center}
        \includegraphics[width=14cm]{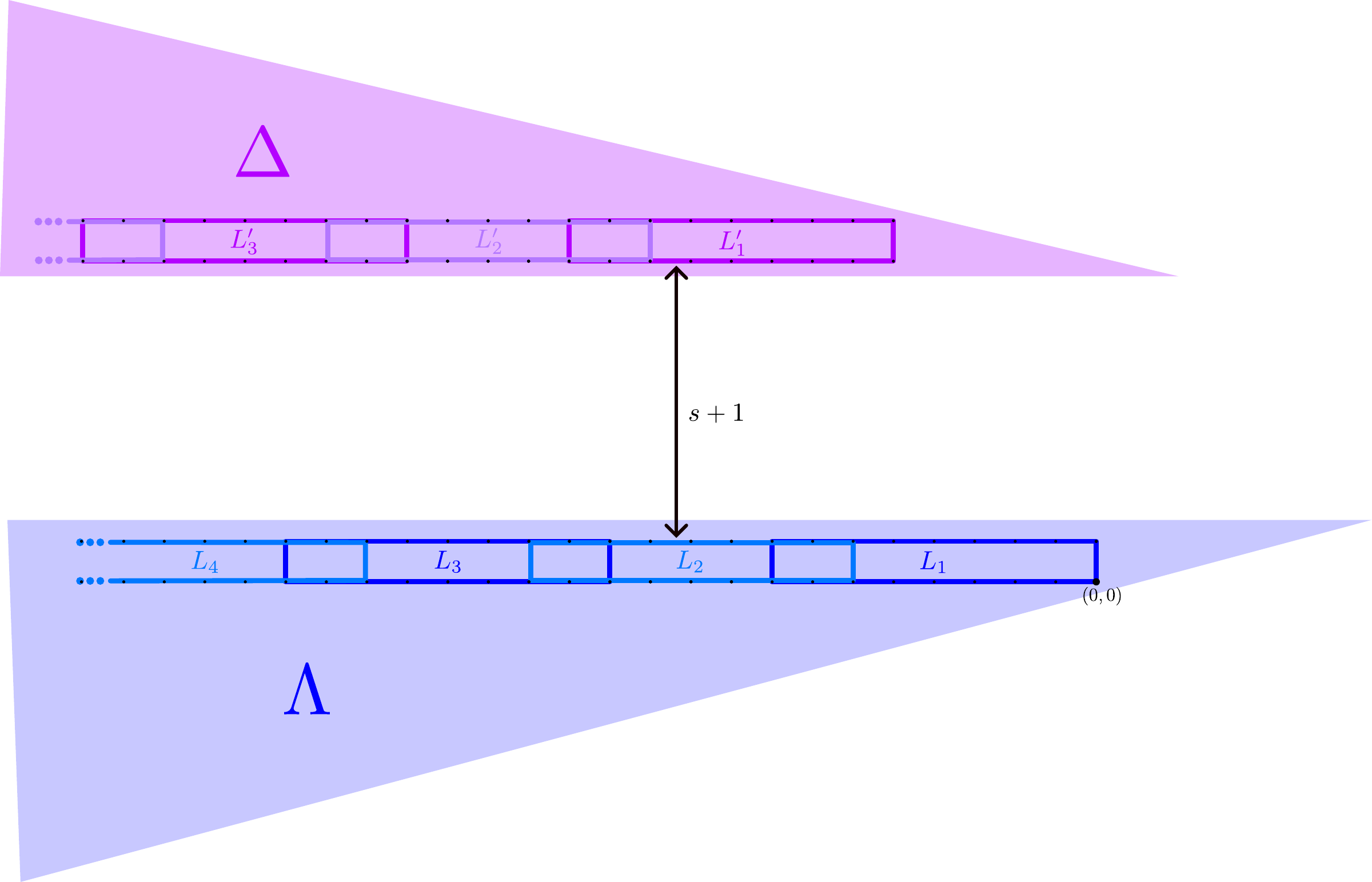}
    \end{center}
    \caption{The fiducial chain $\scrC$ in blue, separated a distance $s+1$ from the right chain $\scrC'$ in purple above it. Also shown are cones $\Lambda$ and $\Delta$ satisfying the assumptions of the braiding construction in Section \ref{ssubsec:tensor product and braiding}.}
    \label{fig:fiducial and right chain}
\end{figure}

In order to compute $R$-symbols for $\DHR_f$, we must compute the braidings $b_{Y, X} := b_{\bar \rho^Y, \bar \rho^X} : \bar \rho^Y \times \bar \rho^X \rightarrow \bar \rho^X \times \bar \rho^Y$ for any $X, Y \in \Irr Z(\caC)$. In order to do this, we fix a second chain $\scrC' = (L'_n)$ with links $L'_n$ obtained from $L_n$ by translating $s+2$ lattice spacings up and five lattice spacings to the left. This guarantees that the endomorphisms $\bar \rho_{\scrC}^X$ and $\bar \rho_{\scrC'}^X$ are localized in allowed cones $\Lambda$ and $\Delta$ respectively, satisfying the assumptions of the construction of the braiding in Section \ref{ssubsec:tensor product and braiding}, see Figure \ref{fig:fiducial and right chain}. 

To compute the braiding $b_{Y, X}$ we require a unitary transporter $T^X : \bar \rho^X \rightarrow \bar \rho^X_{\scrC'}$. We will take $T^X$ to be the transporter constructed in \Cref{prop:explicit transporters} using a bridge $( \{ Q_n \}_{n \in \N}, \{ E_n \}_{n \in N} )$ between $\scrC$ and $\scrC'$ with the links $Q_n$ and regions $E_n$ specified for each $n \in \N$ as follows:
\begin{center}
	   \includegraphics[width = 0.5\textwidth]{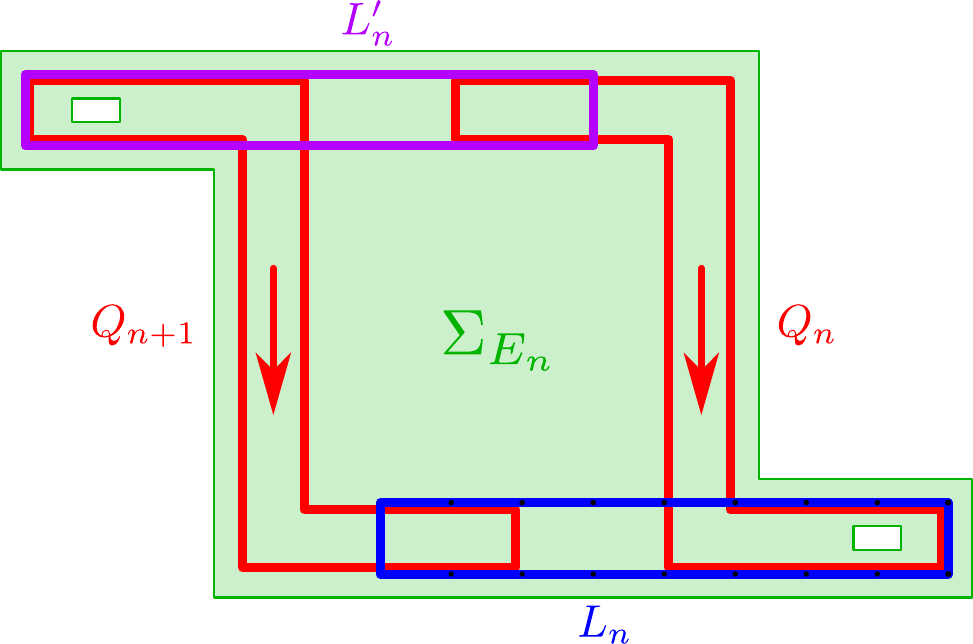}.
\end{center}
One easily checks that this indeed provides a bridge from $\scrC$ to $\scrC'$.

We then have
\begin{equation}
    b_{Y,X} = (T^X)^* \,  \bar \rho^Y( T^X ).
\end{equation}

Recall from appendix \ref{app:equivalence} that the $R$-symbols of $\DHR_f$ and $Z(\caC)$ are simply the actions of $- \circ b_{X, Y}$ and $- \circ \beta_{X, Y}$ on fusion spaces. The following lemma shows that the maps $\{\Phi_{XY}^Z\}$ intertwine the $R$-symbols in the sense described in Appendix \ref{app:equivalence}.

\begin{lemma}\label[lemma]{lem:Phi preserves braiding}
    For any $X,Y,W \in \Irr \ZC$, and $\alpha : X \otimes Y \to W  $,
    \begin{equation*}
        \Phi^W_{XY}[\alpha] \circ b_{Y,X}  = \Phi^W_{YX}\big[\alpha \circ \beta_{Y,X} \big]. 
    \end{equation*}
\end{lemma}

\begin{proof}
    From \Cref{lem:Phi defined is intertwiner} we see that the claim is equivalent to
    \begin{equation*}
        b_{Y,X}^* \circ \Phi_W^{XY}[\gamma] = \Phi_W^{YX}\big[\beta_{Y,X}^\dagger \circ \gamma \big],
    \end{equation*}
    where $\gamma = \alpha^\dagger : W \to X \otimes Y$.
    It suffices to verify this identity on the dense subspace $\pi^{\1}(\caA^{\loc})\ket{\Omega}$ of local states. We shall henceforth omit $\pi^{\I}$ from notation and consider a local state $\ket\Psi = x \ket\Omega$ with $x\in \caA^{\loc}$.
    
    From \Cref{lem:strings produce anyons} and \Cref{def:bridge} we can find an $N \in \N$ large enough so that the support of $x$ is disjoint from $E_n$ for all $n \ge N$ and
    \begin{equation} \label{eq:W is there}
        U_{\scrC_{n-1}}^W \ket{\Psi} = P_{E_{n}}^{(\I, W)} U_{\scrC_{n-1}}^W \ket{\Psi}
    \end{equation}
    for all $n \geq N$.
    We aim to show
    \begin{equation} \label{eq:crucial braiding equality}
        (u_{Q_n}^X)^*   (u_{L'_n}^X)^* u_{L_n}^Y u_{Q_{n+1}}^X  \Dr_{L_n}\big[ \gamma \big] P_{E_{n}}^{(\I, W)}
        = \Dr_{L_n}\big[ \beta_{Y,X}^\dagger \circ \gamma \big] P_{E_{n}}^{(\I, W)}
    \end{equation}
    for all such $n$. Fix $n \geq N$, and let $D_n$ be the region obtained from $E_{n}$ by placing an additional puncture at $\partial_{\ii} L_n$, so $\Sigma_{D_n}$ becomes a disk with three punctures. Write $K_n := Q_{n+1} \wedge L_n$ and $K'_n := L'_n \wedge Q_n$ so the region $E_n$ supports a simple isotopy between $K_n$ and $K'_n$. As in the discussion below \Cref{def:simple isotopy}, we may take $\anchor_{E_n} = \anchor_{K_n}$.
    Note that we have $P_{D_n}^{(\I,\I, W)} = P_{E_n}^{(\I, W)}$ where we use the ordering $( \partial_{\ii} L'_n, \partial_{\ii} L_n, \partial_{\f} L_n )$ of the three punctures of $D_n$.
    
    Recall from Section \ref{subsec:Drinfeld insertions} that the range of $P_{D_n}^{(\I, \I, W)}$ is spanned by vectors of the form
    \begin{align}
    \label{eq:paramatrising under Psi}
        \psi &=  \caD^{-3} \times \Psi_{D_n}^{\anchorino_{1}} \big( \al \otimes w_0 \otimes w^{\I} \otimes w^{\I} \otimes w^W \big) \\
    \notag
        &=  \big( d_{X_0} d_{W} \big)^{1/2} \,\sigma_{D_n}^{-1} \left( \adjincludegraphics[valign=c, height = 2.0cm]{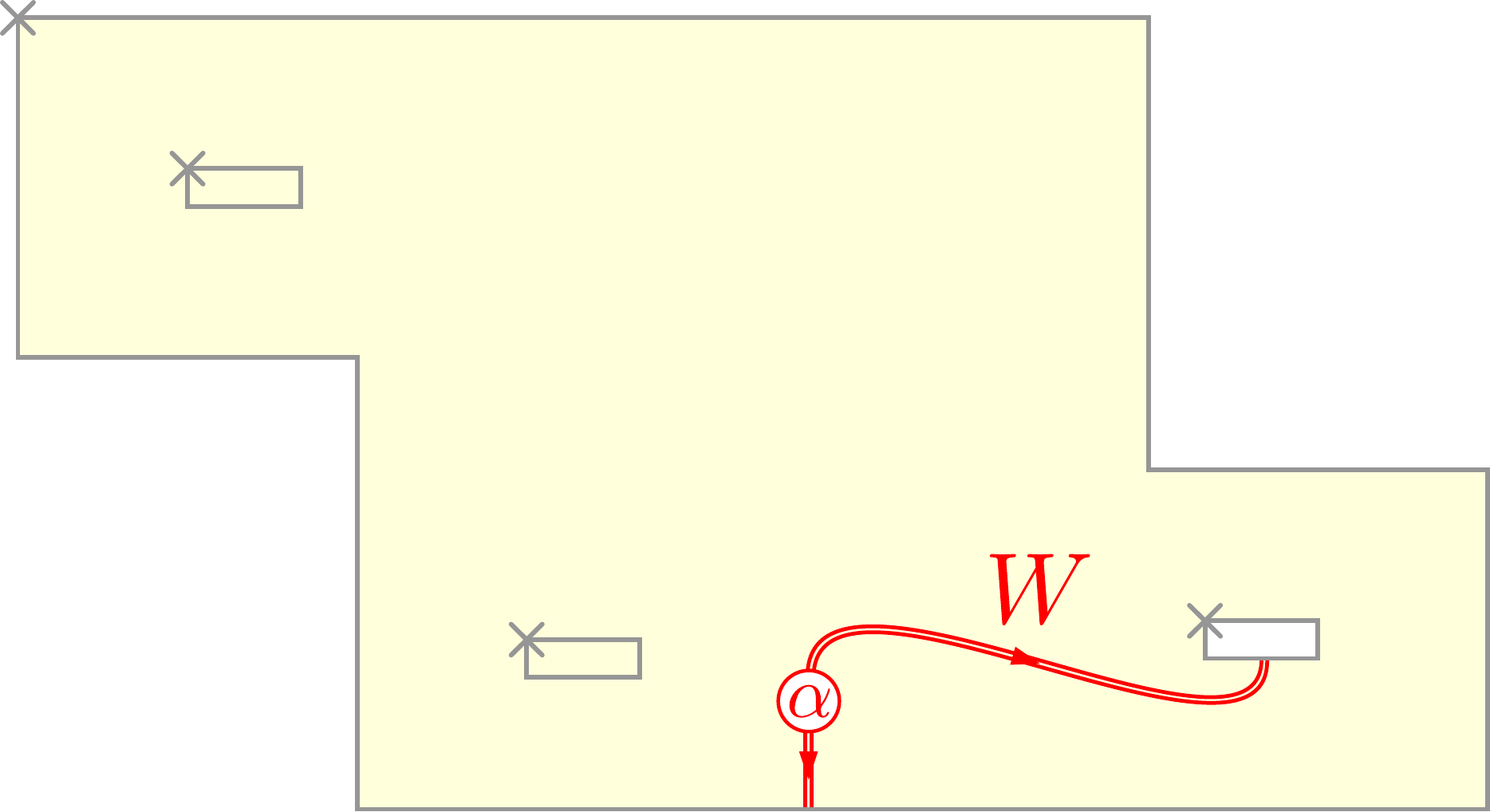} \right),
    \end{align}
    for arbitrary boundary conditions $w_0 : \bar W \rightarrow \chi^{\otimes \caS^{\anchorino_1}_0}$.
    Here, the anchor $\anchor_1$ is as shown in Figure \ref{fig:braiding anchors} and we use the graphical representation introduced in Eq. \eqref{eq:star to bar} for the unitary $\zeta_W : W^* \rightarrow \bar W$. We adopt the convention that we do not indicate vacuum lines and that the vacuum boundary condition $w^{\I}$ is imposed on boundary components with no (= vacuum) attaching strands (\cf \Cref{conv:Drinfeld strands attaching to boundary components}). We also fill in the punctures at such boundary components to emphasize that they are cloaked and any strand of a string diagram can be pulled across them by isotopy without changing without changing the vector being represented. We also drop the specification of the outer boundary condition $w_0$ from the figures as it will be fixed throughout the upcoming computation.

    Further dropping the notation for the isomorphism $\sigma_{D_n}$, we compute the action of the operator $(u_{Q_n}^X)^*   (u_{L'_n}^X)^* u_{L_n}^Y u_{Q_{n+1}}^X  \Dr_{L_n}\big[ \gamma \big]$ on a vector $\psi$ of this form. First, using \cref{lem:inclusion lemma} for $L_n \subset D_n$ with anchors $\anchor_{L_n}$ and $\anchor_1$ respectively, we get
    \begin{align*}
        \psi &\xmapsto{\Dr_n[\gamma]} \, \big( d_{X_0} d_{X} d_Y \big)^{1/2} \, \adjincludegraphics[valign=c, height = 2.0cm]{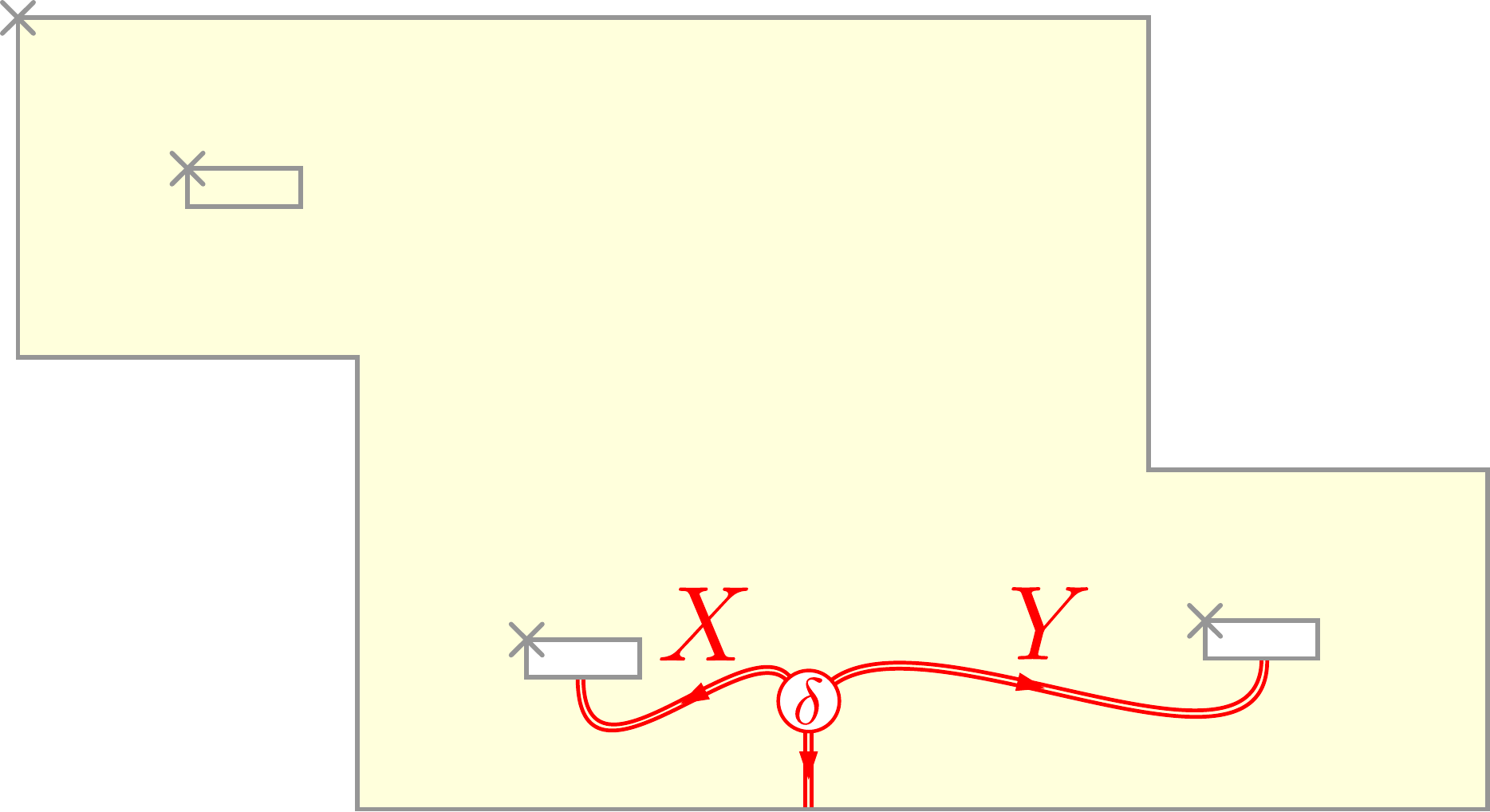} \\
        \intertext{Applying \cref{lem:inclusion lemma} to $Q_{n+1} \subset D_n$ using anchors $\anchor_{Q_{n+1}}$, and $\anchor_{1}$ respectively, now yields}
        &\xmapsto{u_{Q_{n+1}}^X} \, \big( d_{X_0} d_{X} d_Y \big)^{1/2} \, \adjincludegraphics[valign=c, height = 2.0cm]{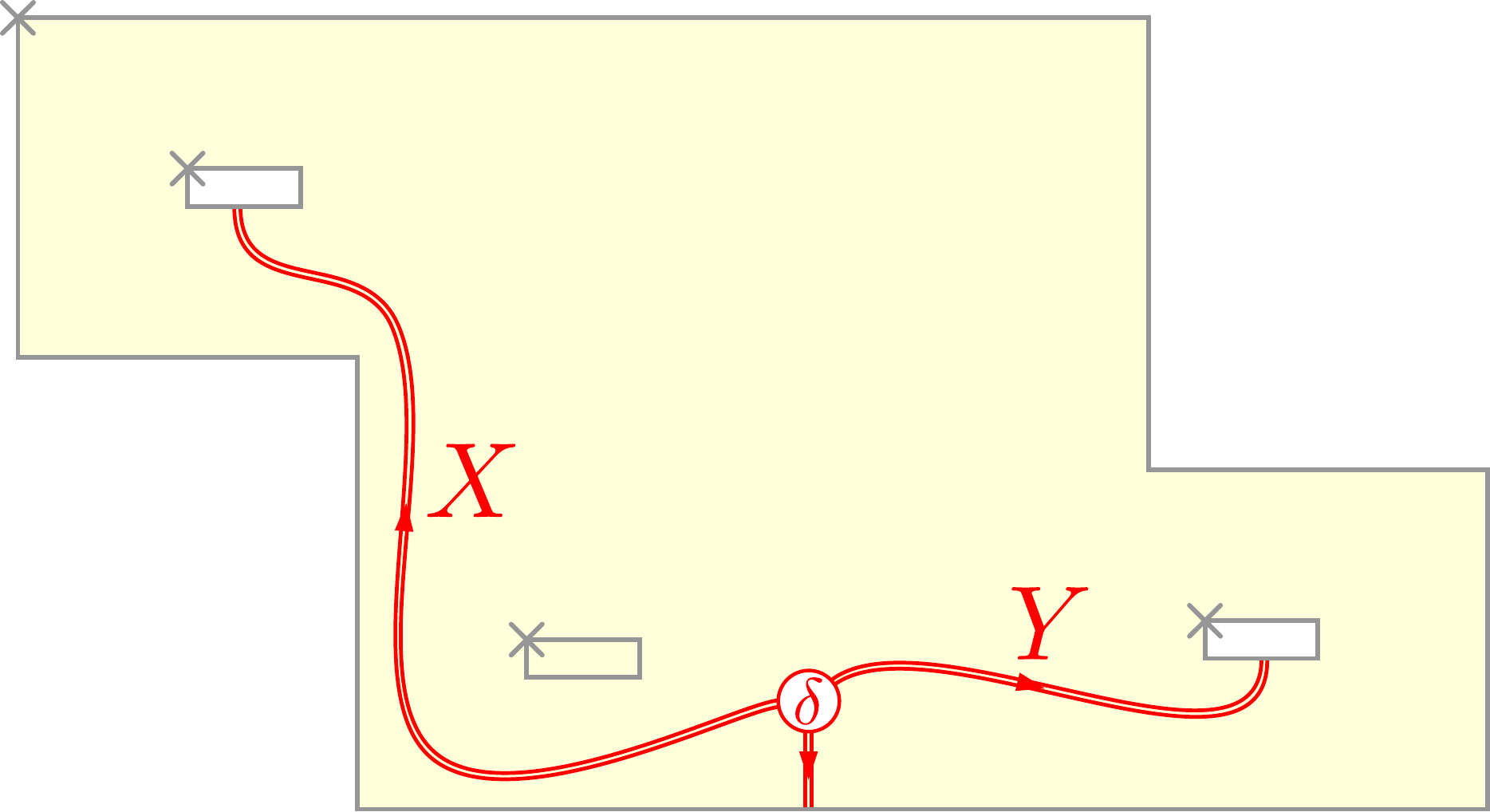}
        \intertext{A similar application of \cref{lem:inclusion lemma} shows}
        & \xmapsto{u_{L_n}^Y} \, \big( d_{X_0} d_{X} d_Y \big)^{1/2} \, \adjincludegraphics[valign=c, height = 2.0cm]{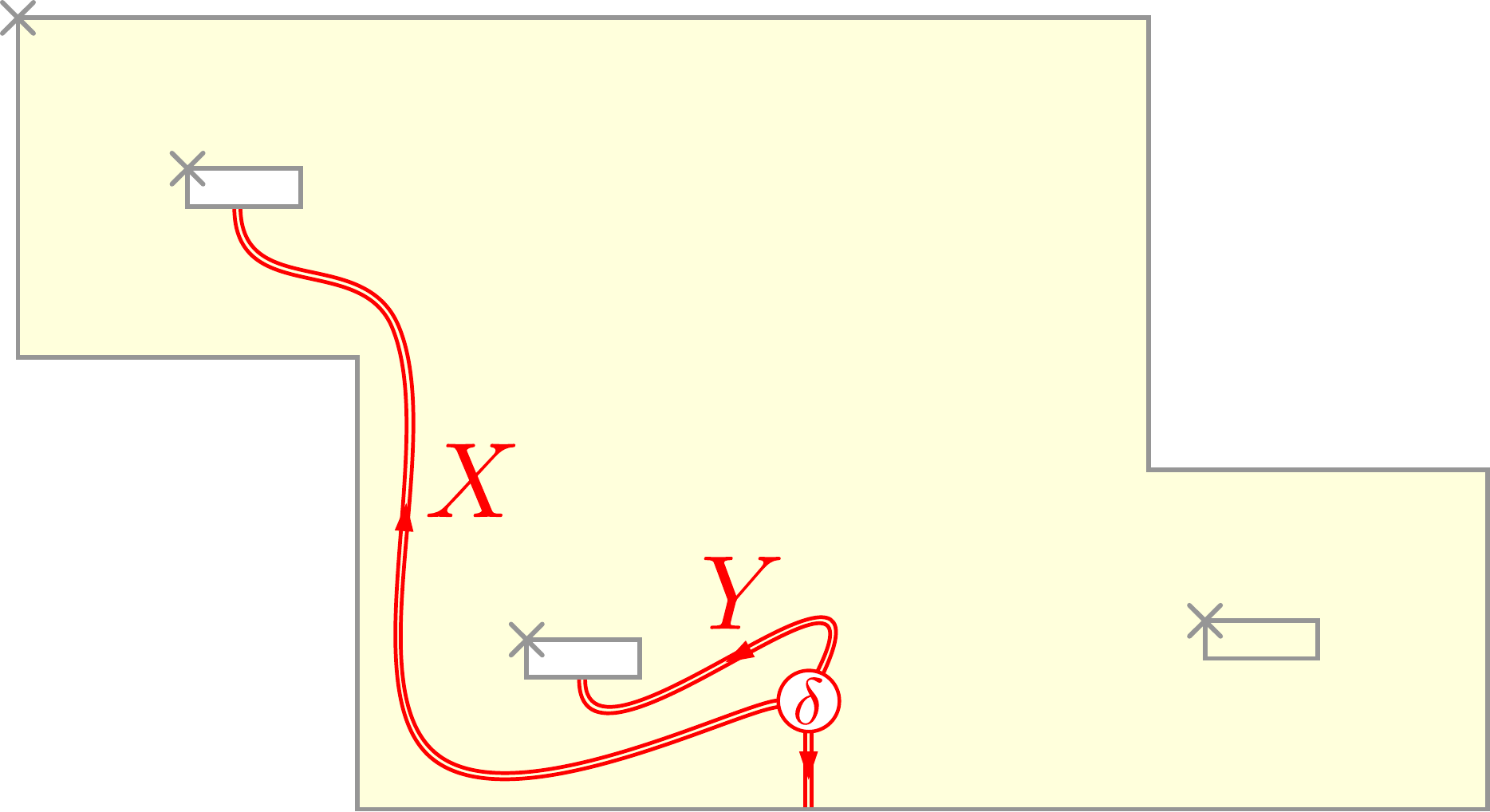}
        \intertext{Note that the resulting vector is in the range of $P_{K'_n}^{X\I}$. We can therefore apply \Cref{lem:variation on concatenation} and self-adjointness of $u_{L'_n}^X$ and $u_{Q_n}$ to see that the action of $(u_{Q_n}^X)^*   (u_{L'_n}^X)^*$ on this vector is the same as the action of $u_{K'_n}^X$.
        In order to apply \cref{lem:inclusion lemma} to understand the action of $u^X_{K'_n}$, we consider the anchor $\anchor_2$ shown in Figure \ref{fig:braiding anchors} which extends $\anchor_{K'_n}$. The corresponding $\Dr^{\anchorino_2}_{D_n}$-action is defined in terms of $\Psi^{\anchorino_2}_{D_n}$ and it is clear how the given state is parametrised under $\Psi^{\anchorino_2}_{D_n}$ (in the sense of Eq. \eqref{eq:paramatrising under Psi}), so we obtain}
        & \xmapsto{(u_{Q_n}^X)^*   (u_{L'_n}^X)^*} \, \big( d_{X_0} d_{X} d_Y \big)^{1/2} \, \adjincludegraphics[valign=c, height = 2.0cm]{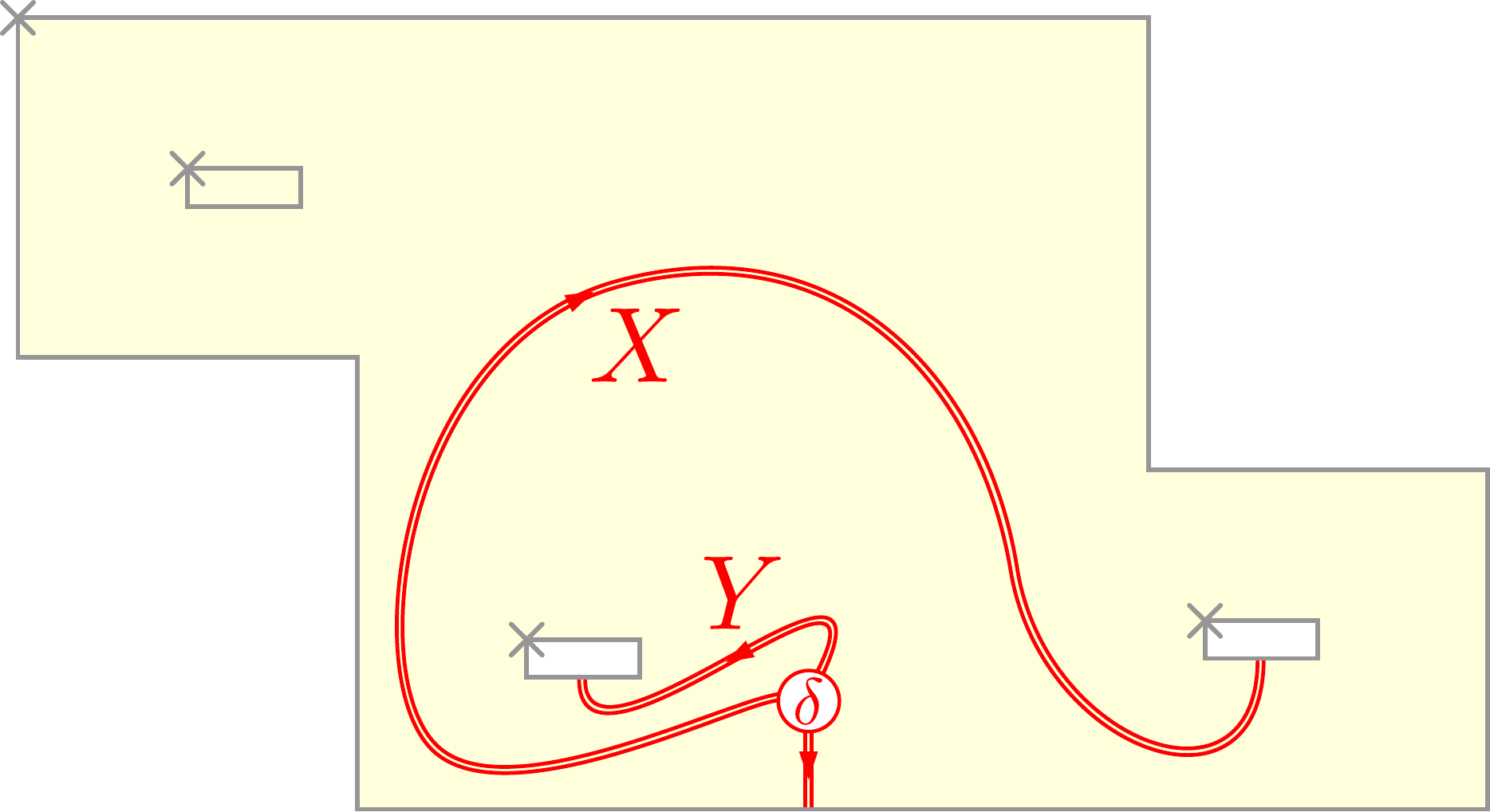}
        \intertext{Finally, using isotopy and Eq. \eqref{eq:normalization and cloaking properties}, and recalling \Cref{conv:Drinfeld strands attaching to boundary components}, we note that that the $X$-strand can be pulled through the puncture at $\partial_{\ii} L_n$, ending up below the $Y$-strand as follows:}
        &= \, \big( d_{X_0} d_{X} d_Y \big)^{1/2} \, \adjincludegraphics[valign=c, height = 2.0cm]{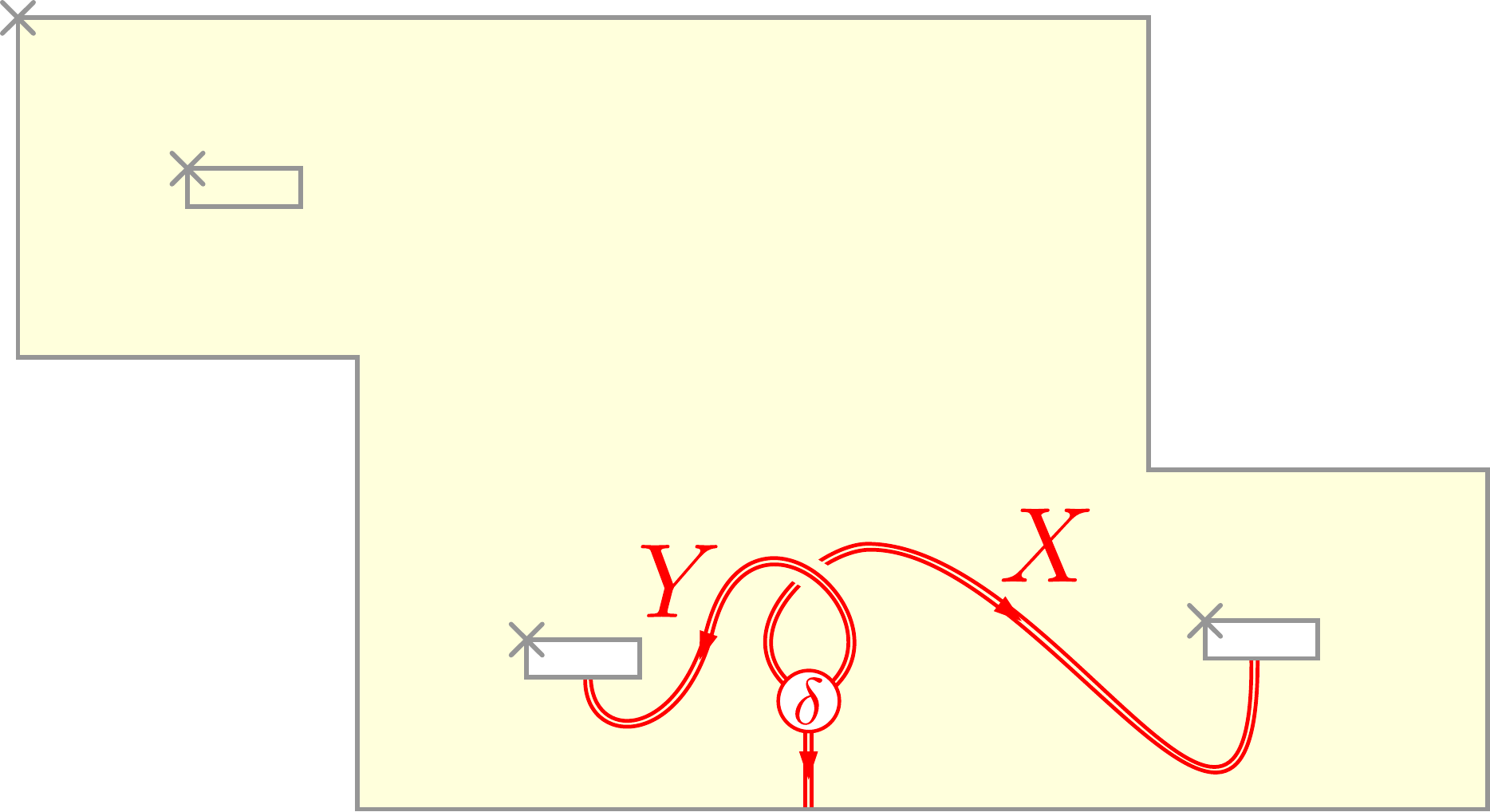} \\
        &= \Dr_{L_n} \big[ \beta_{Y, X}^{\dag} \circ \gamma \big] \, \psi,
    \end{align*}
    where we again used \cref{lem:inclusion lemma} in the last step. Since $\psi \in \Ran P_{D_n}^{(\I, \I, W)} = \Ran P_{E_n}^{(\I, W)}$ was arbitrary, this proves \eqref{eq:crucial braiding equality}.

    \begin{figure}
         \begin{center}
            \includegraphics[width=14cm]{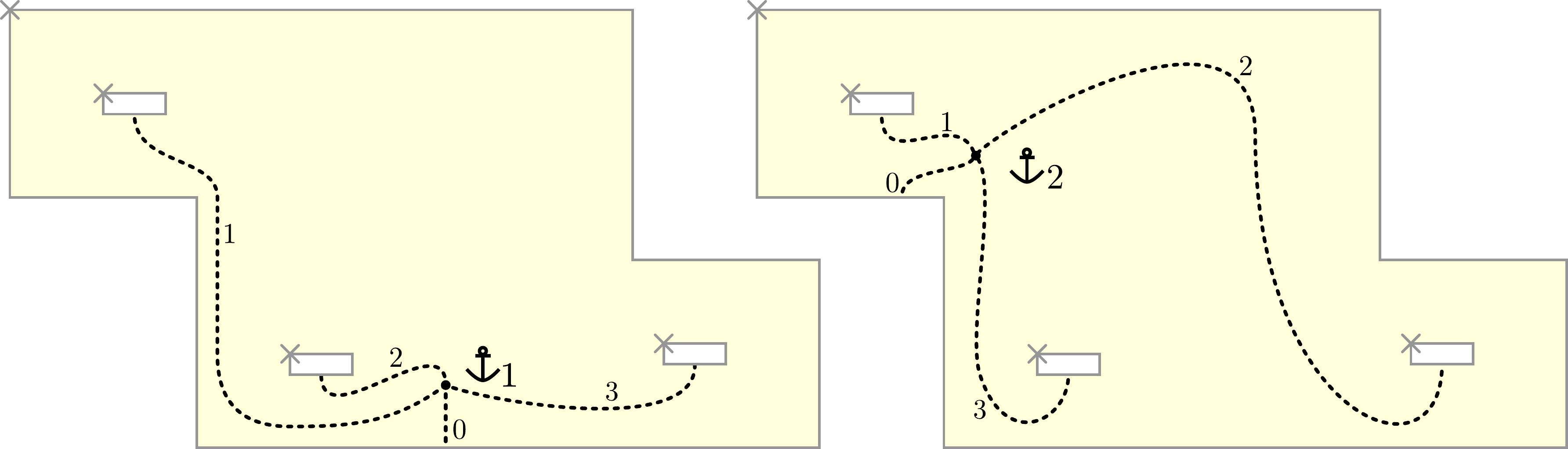}
        \end{center}
        \caption{Anchors used in the proof of \Cref{lem:Phi preserves braiding}.}
         \label{fig:braiding anchors}
    \end{figure}

    Let us now use the results \eqref{eq:W is there} and \eqref{eq:crucial braiding equality} to prove the claim of the lemma. From \Cref{rem:transporter with punctures} we have that the operator $\rho^Y( (T^X)^* )$ is the limit of a sequence which is eventually constant on local states. So, using \Cref{lem:evt constant product}, we compute
    \begin{align*}
        & \bar\rho^Y \big(  T^{X}  \big)^* T^X \, \Phi_W^{XY}[\gamma] \, \ket{\Psi} \\
        &= (U_{\scrC_n}^Y)^* (U_{\scrC_{n-1}}^X)^* (u_{Q_{n}}^X)^* U_{\scrC'_{n-1}}^X U_{\scrC_n}^Y 
        \times (U_{\scrC'_n}^X)^* u_{Q_{n+1}}^X  U_{\scrC_n}^X
        \times (U_{\scrC_n}^X)^* (U_{\scrC_{n-1}}^Y)^* \Dr_{L_n}\big[ \gamma \big] U_{\scrC_{n-1}}^W \ket{\Psi} \\
        &= (U_{\scrC_n}^Y)^* (U_{\scrC_{n-1}}^X)^* 
        (u_{Q_{n}}^X)^*   (u_{L'_n}^X)^* u_{n}^Y u_{Q_{n+1}}^X  
        \Dr_{L_n}\big[ \gamma \big] U_{\scrC_{n-1}}^W \ket{\Psi} \\
        &= (U_{\scrC_n}^Y)^* (U_{\scrC_{n-1}}^X)^* \Dr_{L_n}\big[ \beta_{Y,X}^\dagger \circ \gamma \big] U_{\scrC_{n-1}}^W \ket{\Psi}
        \\ &= \Phi_W^{YX}\big[\beta_{Y,X}^\dagger \circ \gamma \big] \ket{\Psi}.
    \end{align*}
    Since the local state $\ket{\Psi}$ was arbitrary, the claim now follows by density.
\end{proof}

\section{Proof of Theorem \ref{thm:main theorem}} \label{sec:proof of main theorem}

\begin{proof}
    Fix the fiducial chain $\scrC$ as in Section \ref{subsec:braiding} and write $\bar \rho^X = \bar \rho^X_{\scrC}$ for all $X \in \Irr Z(\caC)$.

    Recall from Section \ref{subsec:DHR} that under \Cref{ass:bounded spread Haag duality} the category $\SSS_f$ is a unitary braided tensor category. Theorem \ref{thm:bijection of simples} moreover implies that $\DHR_f$ is \emph{finite} semisimple, the set $\Irr \DHR_f := \{ \bar \rho^X \}_{X \in \Irr Z(\caC)}$ being a representative set of simple objects for $\DHR_f$. The map $X \mapsto \bar \rho^X$ provides a bijection from $\Irr Z(\caC)$ to $\Irr \DHR_f$. By Lemmas \ref{lem:Phi dagger lemma} and \ref{lem:Phi is monoidal} we have that the maps $\{ \Phi_{XY}^Z \}_{X, Y, Z \in \Irr Z(\caC)}$ are intertwiners of $F$ and $R$-symbols of $Z(\caC)$ and $\DHR_f$ with respect to these sets of simples. \Cref{prop:fusion rules in DHR} says that $\{\Phi_{XY}^Z\}$ preserve orthogonal direct sum decompositions. It follows from Proposition \ref{prop:isomophic symbols implies isomorphic categories} that $Z(\caC)$ and $\DHR_f$ are equivalent as unitary braided tensor categories.
\end{proof}

\newpage
\appendix

\section{\texorpdfstring{$F$}{F}- and \texorpdfstring{$R$}{R}-symbols determine braided tensor structure} \label{app:equivalence}

We give a detailed account of this well-known statement.

\subsection{Isomorphism of \texorpdfstring{$F$}{F}- and \texorpdfstring{$R$}{R}-symbols} \label{subsec:isomorphism of F- and R-symbols}

Let $(\caC, \otimes, \al, \I, l, r)$ be a semisimple monoidal category and let $\Irr \caC$ be a set of representative simple objects containing $\I$. The $F$-symbols of $\caC$ with respect to $\Irr \caC$ are the invertible linear maps
$$ F_{i j k}^l :  \bigoplus_{n \in \Irr \caC} \, \caC(i \otimes n \to l) \otimes \caC(i \otimes k \to n)      \rightarrow     \bigoplus_{m \in \Irr \caC} \, \caC(m \otimes k \to l) \otimes \caC(i \otimes j \to m) $$
defined for all $i, j, k, l \in \Irr \caC$ by the commuting diagram
\begin{equation*}
    \begin{tikzcd}
        \caC(i \otimes (j \otimes k) \to l) \, \arrow[rr, "\simeq"] \arrow[d, "-\, \circ \, \alpha_{i,j,k}"'] & & 
        \displaystyle\bigoplus_{n \in \Irr \caC} \caC(i \otimes n \to l) \otimes \caC(j \otimes k \to n) \arrow[d, "F_{ijk}^{l}"] \\
        \caC((i \otimes j) \otimes k \to l) \, \arrow[rr, "\simeq"]  & & 
        \displaystyle\bigoplus_{m \in \Irr \caC} \caC(m \otimes k \to l) \otimes \caC(i \otimes j \to m).
    \end{tikzcd}
\end{equation*}
If $\caC$ is equipped with a braiding $\beta$, then the $R$-symbols of $\caC$ with respect to $\Irr \caC$ are the invertible morphisms
$$ R_{ij}^k : \caC(j \otimes i \to k) \xrightarrow{- \circ \beta_{i, j}} \caC(i \otimes j \to k),$$
defined for all $i, j, k \in \Irr \caC$.

Suppose now $(\caC, \otimes, \psup{\caC}{\al}, \I, \psup{\caC}{l}, \psup{\caC}{r})$ and $(\caD, \otimes, \psup{\caD}{\al}, \I, \psup{\caD}{l}, \psup{\caD}{r})$ are semisimple monoidal categories with representative sets of simples $\Irr \caC$ and $\Irr \caD$, respectively.
In addition, suppose that there is a bijection $f : \Irr \caC \to \Irr \caD$ such that $f(\I) = \I$, which implies that $\caC$ and $\caD$ are equivalent as linear categories. We say $\caC$ and $\caD$ have the same fusion rules if
$$ \caC(i \otimes j \to k) \simeq \caD( f(i) \otimes f(j) \to f(k) ) $$
for all $i, j, k \in \Irr \caC$, in which case $f$ provides an isomorphism of fusion rules. 

Let $\psup{\caC}{F}$ and $\psup{\caD}{F}$ be the $F$-symbols of $\caC$ and $\caD$ with respect to the representative sets of simples $\Irr \caC$ and $\Irr \caD$. Writing $\und i = f(i)$ to lighten notation, we say a collection of morphisms
$$ \phi_{i, j}^k : \caC(i \otimes j \to k) \rightarrow \caD( \und{i} \otimes \und{j} \to \und{k} ), $$
intertwines the $F$-symbols if the following diagram commutes for all $i, j, k, l \in \Irr \caC$:
\begin{equation}\label{eq:isomorphic F-symbols}
\begin{tikzcd}
    \bigoplus_n \caC( i \otimes n \to l ) \otimes \caC(j \otimes k \to n) \arrow{rr}{\psup{\caC}{F}_{i j k}^l} \arrow{dd}{ \bigoplus_n \, \phi_{i, n}^l \otimes \phi_{j, k}^{n}} \, && \,
    \bigoplus_m \caC(m \otimes k \to l) \otimes \caC(i \otimes j \to m) \arrow{dd}{\bigoplus_m \, \phi_{m, k}^l \otimes \phi_{i, j}^m} \\ && \\
    \bigoplus_n \, \caD( \und{i} \otimes \und{n} \to  \und{l} ) \otimes \caD( \und{j} \otimes \und{k} \to \und{n} ) \arrow{rr}{\psup{\caD}{F}_{\und{i} \und{j} \und{k}}^{\und{l}}} \, && \,
    \bigoplus_m \, \caD( \und{m} \otimes \und{k} \to \und{l}) \otimes \caD( \und{i} \otimes \und{j} \to \und{m} )
\end{tikzcd}
\end{equation}
If each $\phi_{i, j}^k$ is moreover invertible, then we say $\{ \phi_{i, j}^k \}$ provides an isomorphism of $F$-symbols, which in particular implies that $f$ provides an isomorphism of fusion rules.

If $\caC$ and $\caD$ are also braided and have $R$-symbols $\psup{\caC}{R}$ and $\psup{\caD}{R}$ with respect to $\Irr \caC$ and $\Irr \caD$ respectively, then the maps $\{ \phi_{i, j}^k \}_{i, j, k \in \Irr \caC}$ are said to intertwine the $R$-symbols if the following diagram commutes for all $i, j, k \in \Irr \caC$:
\begin{equation}\label{eq:isomorphic R-symbols}
    \begin{tikzcd}
        \caC(j \otimes i \to k) \arrow[d, "\phi_{j, i}^k"'] \, \arrow[rr, "\psup{\caC}{R}_{i \, j}^k"] &&
        \, \caC(i \otimes j \to k) \arrow[d, "\phi_{i, j}^k"] \\
        \caD( \und j \otimes \und i \to \und k) \,\arrow[rr, "\psup{\caD}{R}_{\und{i} \, \und{j}}^{\und{k}}"] \, &&
        \caD( \und i \otimes \und j \to \und k)
    \end{tikzcd}
\end{equation}
If each $\phi_{i, j}^k$ is moreover invertible, then we say $\{ \phi_{i, j}^k \}$ provides an isomorphism of $R$-symbols.

We say $\caC$ and $\caD$ have isomorphic $F$-symbols if there is an isomorphism of $F$-symbols. We say $\caC$ and $\caD$ have isomorphic $F$- and $R$-symbols if there is an isomorphism of $F$-symbols that intertwines the $R$-symbols.

This appendix is devoted to proving the following result, whose proof appears in  Section \ref{sec:proof of isomorphic symbols implies isomorphic categories} below.
\begin{proposition} \label{prop:isomophic symbols implies isomorphic categories}
    Let $\caC$ and $\caD$ be semisimple monoidal categories. If $\caC$ and $\caD$ have isomorphic $F$-symbols, then they are monoidally equivalent. If $\caC$ and $\caD$ are moreover braided and have isomorphic $F$- and $R$-symbols, then they are braided monoidally equivalent.

    Suppose that $\caC$ and $\caD$ are unitary. We say that $\{ \phi_{i, j}^k \}_{i, j, k \in \Irr \caC}$ preserve orthogonal direct sum decompositions if 
    $( \phi_{i, j}^{k_{\kappa}}(\pi_{\kappa}) : \und i \otimes \und j \to \und{k_{\kappa}} )$ is an orthogonal direct sum decomposition of $\und i \otimes \und j$ for any orthogonal direct sum decomposition $( \pi_{\kappa} : i \otimes j \to k_{\kappa} )$ of $i \otimes j$ into simple objects for any $i, j \in \Irr \caC$. If this is the case and $\{ \phi_{i, j}^k \}_{i, j, k \in \Irr \caC}$ is an isomorphism of $F$-symbols, then there is a unitary monoidal equivalence $\caC \simeq \caD$. If $\caC$ and $\caD$ are moreover unitarily braided and $\{ \phi_{i j}^k\}$ also intertwines $R$-symbols, then there is a braided unitary monoidal equivalence $\caC \simeq \caD$.
\end{proposition}
Note that if $\caC$ and $\caD$ are unitary fusion categories equipped with their canonical spherical traces then their morphism spaces are Hilbert spaces equipped with the trace inner product. Unitarity of the equivalence can then be conditioned on the unitarity of $\phi_{i,j}^k$.

\subsection{Functors of semisimple categories and their natural transformations}

Recall that a semisimple category $\caC$ has direct sum decompositions. That is, if $\Irr \caC$ is a representative set of simples for $\caC$, then each object $a \in \Ob \caC$ admits a finite set of morphisms $\pi_{\kappa} : a \to i_{\kappa}$ and $\iota_{\kappa} : i_{\kappa}  \to a$ for $\kappa = 1, \cdots, n$ to and from simples objects $i_{\kappa} \in \Irr \caC$ such that
$$ \pi_{\lambda} \circ \iota_{\kappa} = \delta_{\lambda, \kappa} \, \id_{i_{\kappa}} \quad \quad \text{and} \quad \quad \sum_{\kappa} \iota_{\kappa} \circ \pi_{\kappa} = \id_a.$$
If $\caC$ is unitary, then the direct sum decomposition can be taken to be \emph{orthogonal}, meaning that $\iota_{\kappa} = \pi_{\kappa}^{\dag}$ for all $\kappa$. The collection of maps $(\pi_{\kappa}, \iota_{\kappa})_{\kappa}$ is called a (orthogonal) direct sum decomposition of $a$.
The fact that any object of a semisimple category can be decomposed into simples implies that functors of semisimple categories and their natural transformations are determined by their values on simple objects. This hinges on the fact that morphisms are determined by their compositions which is expressed by the Yoneda Lemma. This subsection is devoted to reviewing these facts.

Given a linear category $\caC$ and any object $c \in \Ob  \caC$ we get a functor $\caC(c \to -) : \caC \rightarrow \Vec$ which sends an object $a$ to the vector space $\caC(c \to a)$ and sends a morphism $f : a \to b$ to the linear transformation $\caC(c \to a) \rightarrow \caC(c \to b)$ given by post-composition with $f$. Similarly we get a functor $\caC(- \to c) : \caC^{\op} \rightarrow \Vec$ which sends morphisms to their pre-composition. 

These assignments yield the \emph{Yoneda embeddings}
$$ \yo^{\op}_{\caC} : \caC^{\op} \to \Fun( \caC \to \Vec ), \quad \yo_{\caC} : \caC \to \Fun(\caC^{\op} \to \Vec) $$
which send morphisms in $\caC$ to \emph{natural transformations} given by pre- and post-composition respectively.
These maps sending morphisms to natural transformations are invertible by the \emph{Yoneda Lemma} (see for example \cite[~Lemma 2.9.1]{penneysSemisimpleCatNotes}):
\begin{lemma}[Yoneda Lemma]
    Let $F : \caC \to \Vec$ be a linear functor. For each $c \in \Ob \caC$ the map
    $$
        \Nat( \caC(c \to -) \implies F ) \rightarrow F(c) : \eta \mapsto \eta_c(\id_c)
    $$
    is an isomorphism. Similarly, if $F : \caC^{\op} \to \Vec$ is a linear functor, then
    $$
        \Nat( \caC(- \to c) \implies F ) \rightarrow F(c) : \eta \mapsto \eta_c(\id_c)
    $$
    is an isomorphism. 
    In both cases, the inverse maps $v \in F(c)$ to the natural transformation given by $\eta_a(f) = F(f) \big( v \big)$.
\end{lemma}

By applying the Yoneda Lemma to $F= \caC(b \to -)$, resp. $F=\caC(- \to b)$, it implies that the Yoneda embeddings are fully faithful (see for example \cite[~Corollary 2.9.3]{penneysSemisimpleCatNotes}).

\begin{lemma} \label{lem:morphisms determined by postcomposition on simples}
    Let $\caC$ be semisimple with representative set of simples $\Irr \caC$. Any morphism $f : a \to b$ is completely determined by its action by pre-composition on the spaces $\caC(b \to i)$ for $i \in \Irr \caC$.
\end{lemma}

\begin{proof}
    Since the Yoneda embedding is faithful, $f$ is determined by its action on morphism spaces given by precomposition.
    Since $\caC$ is semisimple,
    the action by precomposition with $f$ on any $\caC(b \to c)$ is completely determined by the actions on $\caC(b \to i)$ for $i \in \Irr \caC$. 
    Indeed, let $(\iota_{\kappa}, \pi_{\kappa})$ be a direct sum decomposition of $c$. Then we find for any $g \in \caC(b \to c)$ that
    $$ g \circ f = \id_c \circ g \circ f = \sum_{\kappa} \pi_{\kappa} \circ (\iota_{\kappa} \circ g) \circ f. $$
\end{proof}

\begin{lemma} \label[lemma]{nat transf from simple components}
    Let $F, G : \caC \rightarrow \caD$ be linear functors between semisimple categories $\caC$ and $\caD$. Let $\Irr \caC$ be a representative set of simples for $\caC$. Any collection of morphisms $\{ \eta_i : F(i) \rightarrow G(i) \}_{i \in \Irr \caC}$ uniquely determines a natural transformation $\eta : F \implies G$ by the formula
    $$ \eta_a = \sum_{\kappa} \, G(\iota_{\kappa}) \circ \eta_{i_{\kappa}} \circ F(\pi_{\kappa})$$
    for any direct sum decomposition $(\pi_{\kappa}, \iota_{\kappa})$ of $a$ into simples $i_{\kappa} \in \Irr \caC$.

    If the $\{\eta_i\}_{i \in \Irr \caC}$ are invertible, then $\eta$ is a natural isomorphism. If $F$ and $G$ are dagger functors and the $\{ \eta_i \}_{i \in \Irr \caC}$ are unitary, then $\eta$ is a unitary natural isomorphism.
\end{lemma}

\begin{proof}
    Note first that the $\eta_i$ satisfy naturality conditions amongst themselves. That is, for any $f : i \to j$ we have that
    \begin{equation} \label{eq:naturality for simple components}
        \begin{tikzcd}
        	{F(i)} & {G(i)} \\
        	{F(j)} & {G(j)}
        	\arrow["{\eta_i}", from=1-1, to=1-2]
        	\arrow["{F(f)}"', from=1-1, to=2-1]
        	\arrow["{G(f)}", from=1-2, to=2-2]
        	\arrow["{\eta_j}", from=2-1, to=2-2]
        \end{tikzcd}
    \end{equation}
    commutes. Indeed, if $i \neq j$ then $f = 0$ and the diagram commutes. If $i = j$ then $f = \lambda \id_i$ for some $\lambda \in \C$ and both compositions in the diagram yield $\lambda \eta_i$.

    For every $a \in \Ob \caC$, fix a direct sum decomposition $(\pi^a_{\kappa}, \iota^a_{\kappa})_{\kappa}$ into simples $i^a_{\kappa} \in \Irr \caC$. 
    We define
    $$ \eta_a := \sum_{\kappa} G(\iota^a_{\kappa}) \circ \eta_{i^a_{\kappa}} \circ F(\pi^a_{\kappa}) $$
    for all $a \in \Ob \caC$. 
    We now verify that this indeed defines a natural transformation. For $f : a \to b$ we have
    \begin{align*}
        G(f) \circ \eta_a &= \sum_{\kappa} \, G(f) \circ G(\iota^a_{\kappa}) \circ \eta_{i^a_{\kappa}} \circ F(\pi^a_{\kappa}) \\
        &= \sum_{\kappa, \lambda} \, G(\iota^b_{\lambda}) \circ G(\pi^b_{\lambda} \circ f \circ \iota^a_{\kappa}) \circ \eta_{i^a_{\kappa}} \circ F(\pi^a_{\kappa}) \\
        &= \sum_{\kappa, \lambda} \, G(\iota^b_{\lambda}) \circ \eta_{i^b_{\lambda}} \circ F(\pi^b_{\lambda} \circ f \circ  \iota^a_{\kappa}) \circ F(\pi^a_{\kappa} ) \\
        &= \eta_b \circ F(f)
    \end{align*}
    where we used Eq. \eqref{eq:naturality for simple components} in the third step.

    It remains to verify that the $\eta_a$ are independent of the choice of direct sum decomposition. This follows from a similar computation as the one used to verify naturality, and is left to the reader.

    Finally, if the $\{ \eta_i \}$ are invertible then $\sum_{\kappa} F(\iota^a_{\kappa}) \circ \eta_{i^a_{\kappa}}^{-1} \circ G(\pi^a_{\kappa})$ is an inverse of $\eta_a$. Similarly, if $F$ and $G$ are dagger functors and the $\{ \eta_i \}$ are unitary, then we can use an orthogonal direct sum decomposition $\iota^a_{\kappa} = (\pi_{\kappa}^a)^{\dag}$ of $a$ to find that $\eta_a^{\dag} = \sum_{\kappa} F(\iota^a_{\kappa}) \circ \eta_{i^a_{\kappa}}^{\dag} \circ G(\pi^a_{\kappa})$ is the inverse of $\eta_a$.
\end{proof}

\begin{lemma} \label{lem:functor from assignment of simples}
    Let $\caC$ and $\caD$ be semisimple. Given an assignment $f : \Irr \caC \to \Ob \caD$, there is a  linear functor $F: \caC \to \caD$, unique up to a canonical natural isomorphism, which extends $f$ on objects. If $\caC$ and $\caD$ are moreover unitary categories, then $F$ may be taken to be a dagger functor. In either case, if $f : \Irr \caC \to \Irr \caD$ is a bijection, then $F$ is an equivalence. 
\end{lemma}

\begin{proof}
    For each object $a\in \Ob \caC$, fix a direct sum decomposition $(\pi^a_\kappa, \iota^a_\kappa)$ of $a$ into simples $i^a_{\kappa} \in \Irr \caC$ and
    define 
    $$
        F(a) = \bigoplus_\kappa f(i^a_\kappa),
    $$
    where the right hand side is an arbitrary choice of direct sum of the objects $f(i^a_{\kappa})$ in $\caD$, given by a direct sum decomposition $(\tilde \pi^a_{\kappa}, \tilde \iota^a_{\kappa})$. By functoriality, $F$ is determined on morphisms $f : a \to b$ by
    \begin{equation} \label{eq:F determined on morphisms}
        F(f) = \sum_{\kappa, \lambda} \, \tilde \iota^b_{\kappa} \circ F( \pi^b_{\kappa} \circ f \circ \iota^a_{\lambda} ) \circ \tilde \pi^a_{\lambda},
    \end{equation} 
    where we note that $\pi^b_{\kappa} \circ f \circ \iota^a_{\lambda}$ is either zero or a multiple of the identity.
    With this definition on morphisms, one easily checks that $F(\id_a) = \id_{F(a)}$ and that $F(g) \circ F(f) = F(g \circ f)$ whenever $g$ and $f$ are composable. Suppose $G$ is another linear functor that extends $f$ on objects. Then by Lemma \ref{nat transf from simple components} the identity maps $\eta_i := \id_{F(i)} : F(i) \to G(i)$ for all ${i \in \Irr \caC}$ determine a canonical natural isomorphism $\eta : F \implies G$.

    For the unitary case, note that a unitary semisimple category admits all orthogonal direct sums. We can therefore assume that all direct sum decompositions in the construction above are orthogonal. That $F$ is a dagger functor then follows by noting that $F( \pi^b_{\kappa} \circ f \circ \iota^a_{\lambda} )^{\dag} = F( \pi^a_{\lambda} \circ f^{\dag} \circ \iota^b_{\kappa} )$, since the dagger of this morphism is given by complex conjugation.
    It is a simple exercise to check that $F$ is faithful and that if $f : \Irr \caC \to \Irr \caD$ is a bijection, then $F$ is essentially surjective and full.
\end{proof}

\subsection{Construction of a tensorator}

Under the assumptions of Proposition \ref{prop:isomophic symbols implies isomorphic categories}, Lemma \ref{lem:functor from assignment of simples} gives linear equivalence $F : \caC \to \caD$ such that $F(i) = f(i)$ for each $i \in \Irr \caC$, which we may take to be a unitary equivalence in the case where $\caC$ and $\caD$ are unitary.

By \Cref{nat transf from simple components}, the given isomorphism of $F$-symbols $\{\phi_{i j}^k\}_{i, j, k \in \Irr \caC}$ uniquely determines a natural transformation
$$ \phi_{i, j}^{-} : \caC(i \otimes j \to -) \implies \caD( F(i) \otimes F(j) \to F(-) ) $$
for every pair $i,j \in \Irr\caC$. The component morphisms are given by
\begin{equation} \label{eq:components of phi nat transf}
    \phi_{i, j}^c(\al) = \sum_{\kappa} F( \iota_{\kappa} ) \circ \phi_{i, j}^{i_{\kappa}}( \pi_{\kappa} \circ \al )
\end{equation}
for any direct sum decomposition $(\pi_{\kappa}, \iota_{\kappa})_{\kappa}$ of $c$ into simples $i_{\kappa} \in \Irr(\caC)$.

The Yoneda Lemma gives a canonical isomorphism
\begin{align*} 
    \Nat\big( \caC(i \otimes j \to -) \implies \caD(F(i) \otimes F(j) \to F(-)) \big) &\xrightarrow{\simeq} \caD(F(i) \otimes F(j) \to F(i \otimes j)) \\
     \rho &  \mapsto \,\rho_{i \otimes j}(\id_{i \otimes j}).
\end{align*}
Applying this to the natural transformation $\phi_{i j}^{-}$ yields the morphism
\begin{equation} \label{eq:simple tensorator defined}
    \tau_{i, j} := \phi_{i, j}^{i \otimes j}( \id_{i \otimes j}) : F(i) \otimes F(j) \to F(i \otimes j).
\end{equation}
It follows from the Yoneda isomorphism (or the naturality of $\phi_{i, j}^{-}$) that
\begin{equation} \label{eq:basic equation for simple tau}
    F(g) \circ \tau_{i, j} = \phi_{i, j}^k(g)
\end{equation}
for all $g \in \caC(i \otimes j \to k)$.

By \Cref{nat transf from simple components} the maps $\{\tau_{i, j}\}_{i,j\in\Irr\caC}$ uniquely determine a natural transformation
$$\tau : F(-) \otimes F(-) \implies F(- \otimes -).$$
We will show that this $\tau$ can serve as a tensorator to equip the equivalence $F$ with the structure of a monoidal equivalence.

\begin{lemma} \label{lem:monoidal structure axiom}
    The natural transformation $\tau$ satisfies the following commutative diagram for all $a, b, c \in \Ob \caC$:
    \begin{equation}\label{eq:monoidal structure axiom}
    \begin{tikzcd}[column sep=huge, row sep=large]
        (F(a)\otimes F(b))\otimes F(c)
          \arrow[r, "\psup{\caD}{\alpha}_{F(a),F(b),F(c)}"]
          \arrow[d, "\tau_{a,b}\otimes \id_{F(c)}"']
        &
        F(a)\otimes(F(b)\otimes F(c))
          \arrow[d, "\id_{F(a)}\otimes \tau_{b,c}"]
        \\
        F(a\otimes b)\otimes F(c)
          \arrow[d, "\tau_{a \otimes b, c}"']
        &
        F(a)\otimes F(b\otimes c)
          \arrow[d, "\tau_{a, b \otimes c}"]
        \\
        F((a \otimes b)\otimes c)
          \arrow[r, "F(\psup{\caC}{\alpha}_{a,b,c})"]
        &
        F(a \otimes (b \otimes c))
    \end{tikzcd}
    \end{equation}
\end{lemma}

\begin{proof}
    The two paths from top left to bottom right are components of natural transformations $(F(-) \otimes F(-)) \otimes F(-) \implies F( - \otimes (- \otimes -) )$ which by \Cref{nat transf from simple components} are completely determined by their simple components. It is therefore sufficient to show that \eqref{eq:monoidal structure axiom} commutes for simple objects.
    This will follow from the assumption that the $\{ \phi_{i, j}^k \}$ provide an isomorphism of $F$-symbols.

    Let $i,j,k \in \Irr\caC$, and note that by Lemma \ref{lem:morphisms determined by postcomposition on simples}, the diagram \eqref{eq:monoidal structure axiom} with $a=i, b=j$, and  $c=k$, commutes, if and only if, the top square of
    \begin{equation} \label{eq:diagram I}
        \begin{tikzcd}
            {\caD \big( (F(i) \otimes F(j)) \otimes F(k) \to F(l) \big)} &&& {\caD \big( F(i) \otimes ( F(j) \otimes F(k) ) \to F(l) \big)} \\
            \\
            {\caD \big( F( (i \otimes j) \otimes k) \to F(l) \big)} &&& {\caD \big( F( i \otimes (j \otimes k) ) \to F(l) \big)} \\
            \\
            {\caC \big( (i \otimes j) \otimes k \to l \big)} &&& {\caC \big( i \otimes (j \otimes k) \to l \big)}
            \arrow["{- \circ \psup{\caD}{\alpha}_{\und i, \und j, \und k}}"', from=1-4, to=1-1]
            \arrow["{- \circ \tau_{i \otimes j, k} \circ (\tau_{i, j} \otimes \id_{\und k})}"', from=3-1, to=1-1]
            \arrow["{- \circ \tau_{i, j \otimes k} \circ (\id_{\und i} \otimes \tau_{j, k})}"', from=3-4, to=1-4]
            \arrow["{- \circ F( \psup{\caC}{\alpha}_{i, j, k} )}"', from=3-4, to=3-1]
            \arrow["F"', from=5-1, to=3-1]
            \arrow["F"', from=5-4, to=3-4]
            \arrow["{- \circ \psup{\caC}{\al}_{i, j, k}}"', from=5-4, to=5-1]
        \end{tikzcd}
    \end{equation}
    commutes for all $l \in \Irr \caC$. The bottom square commutes by functoriality of $F$.
    Let $l\in \Irr\caC$ be given, and observe that for any $n, m \in \Irr \caC$, any $\eta : i \otimes j \to m$ and $\xi : m \otimes k \to l$, we have
    \begin{align*}
        F(\xi \circ (\eta \otimes \id_k)) \circ \tau_{i \otimes j, k} \circ (\tau_{i, j} \otimes \id_k) &= F(\xi) \circ \tau_{m, k} \circ (F(\eta) \otimes \id_{\und k}) \circ ( \tau_{i, j} \otimes \id_{\und k} ) \\
        &= \phi_{m, k}^l(\xi) \circ \big( \phi_{i, j}^m(\eta) \otimes \id_{\und k} \big),
    \end{align*}
    where we used naturality of $\tau$ in the first step, and Eq. \eqref{eq:basic equation for simple tau} in the last. This shows that the following diagram commutes:
    \begin{equation} \label{eq:diagram II}
        \begin{tikzcd}
        	\begin{array}{c} \bigoplus_m \, \caD \big( F(m) \otimes F(k) \to F(l) \big) \\ \quad  \otimes \caD \big( F(i) \otimes F(j) \to F(m) \big) \end{array} && {\caD \big( (F(i) \otimes F(j)) \otimes F(k) \to F(l) \big)} \\
        	\\
        	&& {\caD \big( F( (i \otimes j) \otimes k) \to F(l) \big)} \\
        	\\
        	{\bigoplus_m \, \caC( m \otimes k \to l ) \otimes \caC(i \otimes j \to m)} && {\caC \big( (i \otimes j) \otimes k \to l \big)}.
        	\arrow["\simeq", from=1-1, to=1-3]
        	\arrow["{- \circ \tau_{i \otimes j, k} \circ (\tau_{i, j} \otimes \id_{\und k})}"', from=3-3, to=1-3]
        	\arrow["{\bigoplus_m \, \phi_{m, k}^l \otimes \phi_{i, j}^m}"', from=5-1, to=1-1]
        	\arrow["\simeq", from=5-1, to=5-3]
        	\arrow["F"', from=5-3, to=3-3]
        \end{tikzcd}
    \end{equation}
    An analogous computation shows commutativity of
    \begin{equation} \label{eq:diagram III}
        \begin{tikzcd}
        	{\caD \big( F(i) \otimes ( F(j) \otimes F(k) ) \to F(l)  \big)} && \begin{array}{c} \bigoplus_n \, \caD \big( F(i) \otimes F(n) \to F(l) \big) \\ \quad \otimes \caD \big( F(j) \otimes F(k) \to F(n) \big) \end{array} \\
        	\\
        	{\caD \big( F( i \otimes (j \otimes k) ) \to F(l) \big)} \\
        	\\
        	{\caC \big( i \otimes (j \otimes k) \to l \big)} && {\bigoplus_n \, \caC(i \otimes n \to l) \otimes \caC( j \otimes k \to n )}.
        	\arrow["\simeq", from=1-3, to=1-1]
        	\arrow["{- \circ \tau_{i, j \otimes k} \circ (\id_{\und i} \otimes \tau_{j, k})}"', from=3-1, to=1-1]
        	\arrow["F"', from=5-1, to=3-1]
        	\arrow["\simeq", from=5-3, to=5-1]
        	\arrow["{\bigoplus_n \phi_{i, n}^l \otimes \phi_{j, k}^n}"', from=5-3, to=1-3]
        \end{tikzcd}
    \end{equation} 
    Placing diagrams \eqref{eq:diagram II} and \eqref{eq:diagram III} on the sides of diagram \eqref{eq:diagram I}, we find that the latter commutes if and only if the outer square of the following diagram commutes:
    \begin{equation*}
        \begin{tikzcd}
        	{\caD \big( (F(i) \otimes F(j)) \otimes F(k) \to F(l) \big)} &&& {\caD \big( F(i) \otimes ( F(j) \otimes F(k) ) \to F(l)  \big)} \\
        	\begin{array}{c} \bigoplus_m \, \caD \big( F(m) \otimes F(k) \to F(l) \big) \\ \quad  \otimes \caD \big( F(i) \otimes F(j) \to F(m) \big) \end{array} &&& \begin{array}{c} \bigoplus_n \, \caD \big( F(i) \otimes F(n) \to F(l) \big) \\ \quad \otimes \caD \big( F(j) \otimes F(k) \to F(n) \big) \end{array} \\
        	\\
        	{\bigoplus_m \, \caC( m \otimes k \to l ) \otimes \caC(i \otimes j \to m)} &&& {\bigoplus_n \, \caC(i \otimes n \to l) \otimes \caC( j \otimes k \to n )} \\
        	{\caC \big( (i \otimes j) \otimes k \to l \big)} &&& {\caC \big( i \otimes (j \otimes k) \to l \big)}
        	\arrow["{- \circ \psup{\caD}{\alpha}_{\und i, \und j, \und k}}"', from=1-4, to=1-1]
        	\arrow["\simeq"', from=2-1, to=1-1]
        	\arrow["\simeq"', from=2-4, to=1-4]
        	\arrow["{\psup{\caD}F_{\und i \, \und j \, \und k}^{\und l}}"', from=2-4, to=2-1]
        	\arrow["{\bigoplus_m \, \phi_{m, k}^l \otimes \phi_{i, j}^m}"', from=4-1, to=2-1]
        	\arrow["\simeq", from=4-1, to=5-1]
        	\arrow["{\bigoplus_n \phi_{i, n}^l \otimes \phi_{j, k}^n}"', from=4-4, to=2-4]
        	\arrow["{\psup{\caC}{F}_{i \, j \, k }^l}"', from=4-4, to=4-1]
        	\arrow["\simeq", from=4-4, to=5-4]
        	\arrow["{- \circ \psup{\caC}{\al}_{i, j, k}}"', from=5-4, to=5-1]
        \end{tikzcd}\
    \end{equation*}
    But here the top and bottom squares are the definition of the $F$-symbols, and commutativity of the middle square is precisely the assumption that $\{ \phi_{i, j}^k \}$ intertwines $F$-symbols.
\end{proof}

\begin{lemma} \label{lem:braided structure axiom}
    Suppose $\caC$ and $\caD$ are equipped with braidings $\psup{\caC}{\beta}, \psup{\caD}{\beta}$ and the maps $\{ \phi_{i, j}^k \}$ provide an isomorphism of the corresponding $R$-symbols. Then the natural transformation $\tau$ satisfies the following commutative diagram for all $a, b \in \Ob \caC$:
    \begin{equation} \label{eq:braided structure axiom}
        \begin{tikzcd}
        	{F(a) \otimes F(b)} && {F(b) \otimes F(a)} \\
        	{F(a \otimes b)} && {F(b \otimes a)}
        	\arrow["{\psup{\caD}{\beta}_{F(a), F(b)}}", from=1-1, to=1-3]
        	\arrow["{\tau_{a, b}}"', from=1-1, to=2-1]
        	\arrow["{\tau_{b, a}}", from=1-3, to=2-3]
        	\arrow["{F(\psup{\caC}{\beta}_{a, b})}", from=2-1, to=2-3]
        \end{tikzcd}
    \end{equation}
\end{lemma}

\begin{proof}
    The two paths from top left to bottom right of the diagram \eqref{eq:braided structure axiom} are components of natural transformations $F(-) \otimes F(\bullet) \implies F(\bullet \otimes -)$ which, by \Cref{nat transf from simple components} are completely determined by their simple components. It is therefore sufficient to show commutativity of \eqref{eq:braided structure axiom} in the case where $a = i, b = j \in \Irr \caC$ are simple. In this case, and using By Lemma \ref{lem:morphisms determined by postcomposition on simples}, commutativity of \eqref{eq:braided structure axiom} is equivalent to commutativity of the top square of
    \begin{equation} \label{eq:diagram i}
        \begin{tikzcd}
        	{\caD \big( F(i) \otimes F(j) \to F(k) \big)} &&& {\caD \big( F(j) \otimes F(i) \to F(k) \big)} \\
        	{\caD \big( F(i \otimes j) \to F(k) \big)} &&& {\caD \big( F(j \otimes i) \to F(k) \big)} \\
        	{\caC( i \otimes j \to k )} &&& {\caC(j \otimes i \to k)}
        	\arrow["{- \circ \, \psup{\caD}{\beta}_{F(i), F(j)}}"', from=1-4, to=1-1]
        	\arrow["{- \circ \, \tau_{i, j}}", from=2-1, to=1-1]
        	\arrow["{- \circ \, \tau_{j, i}}"', from=2-4, to=1-4]
        	\arrow["{- \circ \, F \big( \psup{\caC}{\beta}_{i, j} \big)}"', from=2-4, to=2-1]
        	\arrow["F", from=3-1, to=2-1]
        	\arrow["F"', from=3-4, to=2-4]
        	\arrow["{- \circ \, \psup{\caC}{\beta}_{i, j}}"', from=3-4, to=3-1]
        \end{tikzcd}
    \end{equation}
    for all $i, j, k \in \Irr \caC$. The bottom square commutes by functoriality of $F$.

    By Eq. \eqref{eq:basic equation for simple tau} we see that the vertical compositions in this diagram are precisely the maps $\phi_{i, j}^k$ and $\phi_{j, i}^k$. Recalling further that the $R$-symbols are defined by precomposition with the braiding, commutativity the diagram \eqref{eq:diagram i} is precisely the statement that the $\{ \phi_{i, j}^k \}$ intertwine the $R$-symbols.
\end{proof}

It  remains to show that $\tau$ is  a natural isomorphism, and in the unitary case, that $\tau$ is unitary.
\begin{lemma} \label{lem:tau is isomorphism}
    The natural transformation $\tau$ is a natural isomorphism.
\end{lemma}

\begin{proof}
    We must check that all components $\tau_{a, b}$ are invertible. By \Cref{nat transf from simple components} it is sufficient to verify this for the simple components $\{\tau_{i, j}\}_{i, j \in \Irr \caC}$.

    By assumption, the maps $\phi_{i, j}^k$ are invertible. The functor $F$ is an equivalence, in particular it  acts invertibly on $\Hom$-spaces. Together with Eq. \eqref{eq:basic equation for simple tau} we conclude that precomposition with $\tau_{i, j}$ acts invertibly on $\caD(F(i \otimes j) \to F(k))$ for any $k \in \Irr \caC$. By \Cref{nat transf from simple components} we find that pre-composition by $\tau_{i, j}$ is an \emph{invertible natural} transformation.
    This natural transformation is precisely the Yoneda embedding $\yo_{\caD}^{\op}(\tau_{i, j})$. Since the Yoneda embedding is fully faithful it follows that $\tau_{i,j}$ is invertible.
\end{proof}

\begin{lemma} \label{lem:tau is unitary}
    If $\caC$ and $\caD$ are unitary monoidal categories and the isomorphisms $\{ \phi_{i, j}^k  \}_{i, j, k \in \Irr \caC}$ preserve direct sum decompositions, then $\tau$ is a unitary natural isomorphism.
\end{lemma}

\begin{proof}
    We must check that all components $\tau_{a, b}$ are unitary. By \Cref{nat transf from simple components} it is sufficient to verify this for the simple components $\{\tau_{i, j}\}_{i, j \in \Irr \caC}$.
    Writing $\tau_{i,j}$ in terms of Eq. \eqref{eq:simple tensorator defined} by specialising the formula \eqref{eq:components of phi nat transf} to an orthogonal direct sum decomposition $(\pi_{\kappa})$ of $i \otimes j$, we get
    \begin{align*}
        \tau_{i,j} \circ \tau_{i,j}^\dagger &= \sum_{\kappa, \lambda} \, F(\pi_{\kappa}^\dag) \circ \phi_{i, j}^{i_{\kappa}}( \pi_{\kappa} ) \circ \phi_{i, j}^{i_{\lambda}}( \pi_{\lambda} )^\dag \circ F(\pi_\lambda) = \sum_\kappa F(\pi_{\kappa}^\dagger) \circ  F(\pi_\kappa) = \id_{i\otimes j}, \\
    \intertext{and}
        \tau_{i,j}^\dagger \circ \tau_{i,j} &= \sum_{\kappa, \lambda} \, \phi_{i, j}^{i_{\lambda}}( \pi_{\lambda} )^\dag \circ F(\pi_\lambda) \circ F(\pi_{\kappa}^\dag) \circ \phi_{i, j}^{i_{\kappa}}( \pi_{\kappa} )  
        = \sum_{\kappa} \, \phi_{i, j}^{i_{\kappa}}( \pi_{\kappa})^\dag  \circ  \phi_{i, j}^{i_{\kappa}}( \pi_{\kappa} ) =  \id_{i\otimes j},
    \end{align*}
    using that $(\phi_{i,j}^k(\pi_\kappa)$ is an orthogonal direct sum decomposition of $\und i \otimes \und j$.
\end{proof}

\subsection{Proof of Proposition \ref{prop:isomophic symbols implies isomorphic categories}} \label{sec:proof of isomorphic symbols implies isomorphic categories}

\begin{proof}
    From Lemmas \ref{lem:monoidal structure axiom}, \ref{lem:braided structure axiom}, \ref{lem:tau is isomorphism}, and \ref{lem:tau is unitary} it follows that $\tau$ is an appropriate tensorator for the functor $F$ in all cases. A compatible unitor isomorphism $\upsilon : \I \to F(\I)$ exists and is uniquely determined by \cite[~Proposition 2.4.3]{etingof2015tensor}.
\end{proof}

\bibliographystyle{unsrturl}
\bibliography{bib}

\end{document}